\documentclass[12pt, draftclsnofoot, onecolumn]{IEEEtran}

\def\outputtype{0} 
\usepackage{ifthen}
\ifcase\outputtype
\or
   \usepackage[1-50]{pagesel} 
\or
   \usepackage[51-]{pagesel} 
\fi

\usepackage{amsmath, amsfonts}     
\usepackage{amssymb}               
\usepackage{cases}                 
\usepackage{relsize}               

\allowdisplaybreaks[4]

\makeatletter
\newcommand{\biggg}{\bBigg@{3}}

\newcommand{\Biggg}{\bBigg@{3.5}}

\newcommand{\bigggg}{\bBigg@{4}}

\newcommand{\Bigggg}{\bBigg@{4.5}}

\makeatother

\usepackage{algorithm}            
\usepackage{algorithmic}          

\usepackage{graphicx}             
\graphicspath{{./figs/}}          

\usepackage{adjustbox}            

\usepackage{makecell}            
\usepackage{array}               
\usepackage{threeparttable}      
\usepackage{stfloats}            

\usepackage[caption=false,font=normalsize]{subfig}

\usepackage{hyperref}            
\usepackage{xcolor}              
\usepackage{color}               
\definecolor{StandardBlue}{rgb}{0,0,1}

\usepackage{textcomp}            
\usepackage{url}                 
\usepackage{verbatim}            
\usepackage{cite}                
\usepackage{enumitem}            

\newcommand{\Rx}{R_{\xv}}

\newcommand{\xv}{{\bf X}}

\newcommand{\xvh}{{\widehat {\bf X}}}

\newcommand{\xvhs}{{\widehat {\bf X}}^{\star}}

\newcommand{\zv}{{\bf Z}}
\newcommand{\zvs}{{\bf Z}^{\star}}

\newcommand{\Kxh}{\mathsf{K}_{\xvh}}

\newcommand{\Kxs}{\mathsf{K}_{\xvh^{\star}}}

\newcommand{\dv}{{\sf d}}
\newcommand{\ev}{{\sf e}}
\newcommand{\av}{{\sf a}}

\newcommand{\vv}{{\sf v}}
\newcommand{\uv}{{\sf u}}

\newcommand{\Gammam}{\mathsf{\Gamma}}
\newcommand{\dsv}{\mathsf{d}^{\star}}
\newcommand{\Kz}{\mathsf{K}_{\zv}}
\newcommand{\Kzs}{\mathsf{K}_{\zv^{\star}}}

\newcommand{\Dm}{\mathsf{D}}

\newcommand{\Em}{{\sf E}}
\newcommand{\Dsm}{\mathsf{D}^{\star}}
\newcommand{\Am}{\mathsf{A}}
\newcommand{\Bm}{\mathsf{B}}
\newcommand{\Cm}{\mathsf{C}}

\newcommand{\Rm}{\mathsf{K}}

\newcommand{\tr}{\mathrm{tr}}
\newcommand{\rank}{\mathrm{r}}

\renewcommand{\det}{\mathrm{det}}
\newcommand{\diag}{\mathrm{diag}}

\newcommand{\zerov}{{\sf 0}}

\newcommand{\iso}{{\rm iso}}

\newcommand{\dd}{{\, \rm{d}}}

\newcommand{\T}{\mathrm{T}}

\newcommand{\lambdav}{\hbox{\boldmath$\lambda$}}

\newcommand{\Pim}{\hbox{\boldmath$\Pi$}}

\def\mindex#1{\index{#1}}

\def\sq{\hbox{\rlap{$\sqcap$}$\sqcup$}}
\def\qed{\ifmmode\sq\else{\unskip\nobreak\hfil
\penalty50\hskip1em\null\nobreak\hfil\sq
\parfillskip=0pt\finalhyphendemerits=0\endgraf}\fi\medskip}

\long\def\defbox#1{\framebox[.9\hsize][c]{\parbox{.85\hsize}{%
\parindent=0pt
\baselineskip=12pt plus .1pt      
\parskip=6pt plus 1.5pt minus 1pt 
#1}}}

\long\def\beginbox#1\endbox{\subsection*{}%
\hbox{\hspace{.05\hsize}\defbox{\medskip#1\bigskip}}%
\subsection*{}}
\def\endbox{}

\newsavebox{\junk}
\savebox{\junk}[1.6mm]{\hbox{$|\!|\!|$}}

\def\limsup{\mathop{\rm lim\ sup}}

\def\bfmath#1{{\mathchoice{\mbox{\boldmath$#1$}}%
{\mbox{\boldmath$#1$}}%
{\mbox{\boldmath$\scriptstyle#1$}}%
{\mbox{\boldmath$\scriptscriptstyle#1$}}}}

\def\bfmY{\bfmath{Y}}

\def\bfmhhaY{\bfmath{\hhaY}} 
\def\bfmhhaY{\hbox to 0pt{$\widehat{\bfmY}$\hss}\widehat{\phantom{\raise 1.25pt\hbox{$\bfmY$}}}}

\def\til={{\widetilde =}}

\def\FRAC#1#2#3{\genfrac{}{}{}{#1}{#2}{#3}}

\def\ddtp{{\mathchoice{\FRAC{1}{d^{\hbox to 2pt{\rm\tiny +\hss}}}{dt}}%
{\FRAC{1}{d^{\hbox to 2pt{\rm\tiny +\hss}}}{dt}}%
{\FRAC{3}{d^{\hbox to 2pt{\rm\tiny +\hss}}}{dt}}%
{\FRAC{3}{d^{\hbox to 2pt{\rm\tiny +\hss}}}{dt}}}}

\def\average#1,#2,{{1\over #2} \sum_{#1}^{#2}}

\def\eye(#1){{\bf(#1)}\quad}

\usepackage{amsthm}

\newtheoremstyle{mybold}
      {3pt}
      {3pt}
      {\itshape}
      {}
      {\bfseries}
      {.}
      { }
      {}

\newtheoremstyle{myremark}
      {3pt}
      {3pt}
      {\normalfont}
      {}
      {\bfseries}
      {.}
      { }
      {}

\theoremstyle{myremark}
\newtheorem{remark_}{Remark}  

\newenvironment{remark}
      {\begin{remark_}}
      {\hfill$\lozenge$\end{remark_}}

\theoremstyle{mybold}
\newtheorem{theorem}{Theorem}
\newtheorem{lemma}{Lemma}
\newtheorem{definition}{Definition}
\newtheorem{corollary}{Corollary}
\newtheorem{proposition}{Proposition}

\def\eq#1/{(\ref{eq:#1})}

\newcommand{\beqn}[1]{\notes{#1}%
\begin{eqnarray} \elabel{#1}}

\newcommand{\eeqn}{\end{eqnarray} }

\newcommand{\beq}[1]{\notes{#1}%
\begin{equation}\elabel{#1}}

\newcommand{\eeq}{\end{equation}}

\def\bdes{\begin{description}}
\def\edes{\end{description}}

\newcounter{rmnum}

\newcounter{anum}

{\end{list}}

\def\ass(#1:#2){(#1\ref{#1:#2})}

\def\ritem#1{
\item[{\sf \ass(\current_model:#1)}]
}

\newenvironment{recall-ass}[1]{%
\begin{description}
\def\current_model{#1}}{
\end{description}
}

\newcounter{problem}
\newcounter{save@equation}
\newcounter{save@problem}
\makeatletter
\newenvironment{problem}
{\setcounter{save@equation}{\value{equation}}%
   \setcounter{save@problem}{\value{problem}}%
   \setcounter{problem}{\value{equation}}%
   \let\c@equation\c@problem%
   \renewcommand{\theequation}{\arabic{equation}}%
   \subequations}
{\endsubequations%
   \setcounter{save@problem}{\value{equation}}%
   \setcounter{equation}{\value{save@equation}}}
\makeatother

\begin{document}

\title{On the Gaussian-Quadratic Rate-Distortion Function for Vector Sources with Individual Distortion Constraints}

\author{
Shuao Chen, Junyuan Gao, Yuxuan Shi, Yongpeng Wu, Giuseppe Caire, \\
H. Vincent Poor, and Wenjun Zhang
\thanks{
   S. Chen, Y. Wu, and W. Zhang are with the Department of Electronic Engineering, Shanghai Jiao Tong
   University, Shanghai 200240, China (e-mails:
   \{shuao.chen, yongpeng.wu, zhangwenjun\}@sjtu.edu.cn)
   (Corresponding author: Yongpeng Wu).
}
\thanks{
   J. Gao is with the Department of Electrical and Electronic Engineering, The Hong Kong Polytechnic University, Hong Kong SAR, China (e-mail: junyuan.gao@polyu.edu.hk).
}
\thanks{
   Y. Shi is with the Department of Networked Intelligence, Pengcheng Laboratory, Shenzhen 410083, China (e-mail: shiyx01@pcl.ac.cn). 
}
\thanks{
   G. Caire is with the Communications and Information Theory Group, Technische Universit{\"a}t Berlin, Berlin 10587, Germany (e-mail: caire@tu-berlin.de).
}
\thanks{
   H. V. Poor is with the Department of Electrical and Computer Engineering, Princeton University, Princeton, NJ 08544, USA (e-mail: poor@princeton.edu).
}
\thanks{Parts of this work were presented at the 2025 IEEE Global Communications Conference~\cite{chen2025joint}.
}
}

\maketitle

\vspace*{-8mm}
\begin{abstract}
This paper investigates the Gaussian-quadratic lossy compression with arbitrary source length under individual distortion constraints.
The rate-distortion function (RDF) is lower-bounded by a Hadamard inequality-based rate, which is tight if and only if the semidefinite condition (SDC) holds. Otherwise, this bound becomes loose, and analytical results are lacking. 
Moreover, the fundamental quantitative relationship between source correlations and the RDF remains incomplete. 
In this paper, we provide new theoretical results under different source covariance matrices and distortion constraints.
First, under arbitrary covariance and distortion constraints, we obtain the spectral properties of the optimal source reconstruction achieving the RDF, and a stronger scalar inequality version of the SDC. 
We propose a class of source covariance matrices based on hierarchical correlations and show that studying the two-type correlation (2-TC) model is sufficient to establish the analytical foundation for the broader class.
Under this covariance, we obtain the RDF with source correlations explicitly incorporated when the SDC holds, and analyze the SDC from the perspectives of distortion constraints and source correlations.
Next, under the 2-TC covariance and two-type distortion (2-TD) constraints, we establish the complete RDFs over seven regions on a distortion plane, with the optimal distortion (rate) allocations determined in each region.
It is revealed that the essence of pursuing the complete RDF lies in thoroughly analyzing the correlations between the optimal distortions.
Finally, under isotropic correlation and identical constraints, we provide the per-component compression rate and show that exploiting correlations can significantly reduce compression costs.
\end{abstract}


\section{Introduction}

\IEEEPARstart{T}{he} rapid growth of artificial intelligence (AI), Internet of Things (IoT), and edge computing has led to the generation of vast datasets containing components with heterogeneous significance, in which mission-critical elements coexist with non-essential data~\cite{baccour2022pervasive,shi2020communication}. 
The next generation of communication systems is expected to confront unprecedented challenges in handling the massive creation, transmission, and processing of data~\cite{agiwal2016next,chafii2023twelve}.
Given the potentially stringent requirements for processing latency, transmission efficiency, and storage limitations, effective compression of these heterogeneous data streams will be indispensable in enhancing overall system efficiency.
From an information-theoretic perspective, the fundamental trade-off in lossy compression lies between the compression rate and the fidelity of the reconstructed source~\cite{shannon1959coding}, and is mathematically characterized by the rate-distortion function (RDF)~\cite{berger1971rate}.

The choice of distortion constraint criterion is critical for establishing the RDF.
For an \(N\)-dimensional source \(\xv \in \mathbb{R}^N\) and its reconstruction \(\xvh \in \mathbb{R}^N\), the distortion matrix is defined as \(\Dm \triangleq \mathbb{E}[(\xv - \xvh)(\xv - \xvh)^{\T}]\), whose diagonal entries \(\{d_i\}_{i=1}^N\) correspond to the component-wise distortions.
The distortion constraint matrix is given by \(\Em=\diag(e_{1},\cdots,e_{N})\), where \(e_{i}\) specifies the permissible distortion for the \(i\)-th component of the source \(\mathbf{X}\). 
Two mostly adopted criteria are as follows~\cite{gray1973new,oohama2012distributed,oohama2014indirect}:

\begin{itemize}
\item \textbf{Sum distortion criterion:} This criterion constrains the sum distortion between the vector source and its reconstruction across all components,
\begin{align}
   \sum_{i=1}^{N} d_{i} \leq \sum_{i=1}^{N} e_{i}.
\end{align}
For a vector Gaussian source \(\xv\) with per-component quadratic distortions \(d_i = \mathbb{E}[|X_i - \widehat{X}_i|^2]\), the RDF is given in closed form and the optimal allocation of \(\{d_i\}_{i=1}^N\) is obtained via the well-known reverse water-filling principle~\cite{kolmogorov1956shannon,thomas2006elements}. 
\item \textbf{Individual distortion criteria:}
The core idea behind these criteria is to impose potentially different reconstruction distortion constraints on each component of the vector source,
\begin{align} \label{eq:individual constraints}
d_{i} \leq e_{i}, \quad \forall \, 1 \leq i \leq N.
\end{align}
Within information-theoretic research, these criteria have been extensively used in network coding, including distributed coding~\cite{berger1978multiterminal,tung1978multiterminal} and multiple description coding~\cite{gamal1982achievable,equitz1991successive}. 
These criteria are particularly useful for sources with heterogeneous fidelity requirements across components with different levels of importance. However, the RDF under these more general distortion constraints lacks a complete and closed-form expression, even in the case of Gaussian sources with quadratic distortion. Therefore, this paper focuses on a thorough investigation of the RDF and the explicit quantitative role of source correlation in the Gaussian-quadratic case.
\end{itemize}

Some existing works have explored the RDF for the Gaussian-quadratic lossy compression under individual distortion constraints. Specifically, for a vector source of length \(N=2\), Xiao et al.~\cite[Theorem 6]{xiao2005compression} derived a closed-form RDF. Later, Lapidoth~\cite[Theorem 3.1]{lapidoth2010sending} and Tinguely~\cite[Theorem 2.2]{tinguely2008transmitting} obtained the same result using an alternative approach.
For arbitrary \(N\), the current best-known results on the RDF use the Hadamard inequality to provide a lower bound~\cite{xiao2005compression}. The Hadamard inequality states that for any positive semidefinite (PSD) matrix \(\Dm \succeq \mathsf{0}\), the following holds~\cite{horn2012matrix}:
\begin{align} \label{eq:hadamard inequality}
\det(\Dm) \leq \prod_{i=1}^N d_i.
\end{align}
This leads to a Hadamard lower bound on the RDF
\begin{align} \label{eq:hadamard bound}
\Rx([N], \ev) \geq R^{\ell}_{\xv}([N], \ev),
\end{align}
where we define the rate as
\begin{align} \label{eq:hadamard bound value}
R^{\ell}_{\xv}([N], \ev) \triangleq \frac{1}{2} \log \frac{\det(\Rm)}{\det(\Em)},
\end{align}
and \(\Rm\) denotes the covariance matrix of the vector Gaussian source.
Equality in the Hadamard bound holds if and only if the semidefinite condition (SDC) defined by the Loewner order~\cite{horn2012matrix}
\begin{align} \label{eq:SDC}
\Rm \succeq \Em
\end{align}
is satisfied, where the optimal distortion matrix is \(\Dsm = \Em\).

Current research on Gaussian-quadratic lossy compression under individual distortion constraints still faces several issues.
First, when the SDC holds, the relationship between the distortion constraints and source correlation has not been fully characterized at the element-wise level. Such a characterization is especially important for revealing how the SDC can be satisfied in practice.
In particular, it remains to determine the set of diagonal matrices \(\Em\) (i.e., the individual distortion constraints) that ensure the SDC for a given Gaussian source covariance matrix \(\Rm\), and vice versa, to identify the set of PSD matrices \(\Rm\) that ensure the SDC for a given \(\Em\).
Second, a major limitation arises when the SDC is not satisfied, as the Hadamard bound in Eq.~\eqref{eq:hadamard bound} can become overly loose and provide no useful information about the RDF.
This issue becomes particularly severe when \(\det(\Em)>\det(\Rm)\), i.e., the source correlations are too strong or the imposed individual distortion constraints are too mild, under which the corresponding Hadamard rate in Eq.~\eqref{eq:hadamard bound value} turns negative.
Despite this, it is noteworthy that existing work often neglects the non-SDC case, resulting in a lack of analytical results, including closed-form expressions for the RDF.
Although numerical solutions for the RDF can be obtained, they are typically limited to point-wise evaluations and become computationally impractical as the source dimension increases.
More importantly, such methods are unable to provide analytical expressions for the RDF or the optimal distortion allocation, and thus fail to explicitly characterize the influence of source correlations and individual distortion constraints on the optimal compression.

To address these problems, this paper establishes new general performance bounds for the non-SDC case, provides an \textit{element-wise} characterization of SDC, and derives a \textit{complete} and \textit{closed-form} RDF under specific conditions.
Here, completeness refers to full coverage of all individual distortion constraints, and closed form refers to an explicit analytical expression rather than the solution of an optimization problem.
Specifically, we first consider the case where both the source covariance matrix and the individual distortion constraints are arbitrary. We obtain the spectral properties of the optimal source reconstruction achieving the RDF. A stronger version of the SDC, expressed as a scalar inequality, is also derived.
To make further progress toward an explicit RDF expression, we introduce a class of covariance matrices with hierarchical \(n\)-type correlation (\(n\)-TC), where \(1 \leq n \leq N-1\) and \(N\) is the number of source components.
We show that the methods and results developed for the 2-TC case naturally extend to this broader \((N-1)\)-TC class. 
Under a 2-TC source covariance and arbitrary distortion constraints, when the SDC holds, we derive the RDF that explicitly incorporates source correlations and distortion constraints. We then equivalently transform the SDC defined by the Loewner partial order into multiple scalar inequalities, enabling precise characterization of the SDC region from the perspectives of distortion constraints and source correlations. From the distortion-constraint perspective, we provide the probability that the SDC holds when the distortion constraints are independently and identically distributed. From the source-correlation perspective, we determine the range of inter-component correlations of a 2-TC source under which the SDC holds.
We next consider the case of 2-type distortion (2-TD) constraints under a 2-TC covariance matrix. We derive complete and closed-form RDFs for all seven regions of the distortion plane, with boundaries that ensure optimal distortion (rate) allocations. We then analyze the compression problem under extreme correlation and in the asymptotic regime of large source length. Finally, we provide the RDF under isotropic correlation and identical constraints.
In summary, our theoretical results refine the characterization of SDC, develop a framework for analyzing the non-SDC case, and quantify the impact of correlations on achievable compression, offering valuable insights for real-world system design and optimization.

We provide theoretical results based on different source covariances and imposed distortion constraints. The contributions of this paper are summarized as follows.
\begin{itemize}
   \item When the source covariance and distortion constraints are both arbitrary, we obtain the spectral properties of the optimal reconstruction and provide a stronger version of the SDC.
Specifically, Theorem~\ref{thm:det(Rm-Dsm)} shows that the SDC is satisfied and inactive if and only if the optimal source reconstruction is non-degenerate.
Furthermore, Theorem~\ref{thm:inertia relation} provides a general upper bound on the number of independent components in the optimal reconstruction. On one hand, it reveals how the imposed individual distortion constraints affect the structure of the optimal reconstruction. On the other hand, it suggests that when the SDC does not hold, a viable strategy for effective compression is to seek reconstructions with fewer independent components.
Moreover, the proof technique for Theorem~\ref{thm:inertia relation} offers a general framework for analyzing rank properties in constrained Max-Det optimization problems by translating the algebraic properties of the optimal solution into the geometric space of the dual variables.
Proposition~\ref{prop:weyl bound} provides a more intuitive interpretation of the SDC: the variances of the equivalent parallel sources from the original vector source should all be no smaller than the mildest distortion constraint among the components. In this case, allocating the full distortion to each original source component is optimal.
   \item We first propose a class of Gaussian source covariance matrices with arbitrarily hierarchical \(n\)-type correlations.
Theorem~\ref{thm:2TC_extension} shows that focusing on the 2-TC case is sufficient to establish the analytical foundation for Gaussian-quadratic compression under general \((N-1)\)-TC covariance.
Under the 2-TC covariance and arbitrary distortion constraints, we derive the relationship between correlations, distortion constraints, and the optimal compression rate when SDC holds, and further characterize the SDC region from both distortion-constraint and source-correlation perspectives.
Specifically, in Theorem~\ref{thm:RE_psd_conditions}, we obtain an element-wise characterization of the SDC and a closed-form RDF that incorporates source correlations.
From the distortion-constraint perspective, when individual distortion constraints are i.i.d., Theorem~\ref{thm:sdc exponential asymp} provides both an upper bound and an asymptotic approximation for the probability that the SDC holds. A key insight is that this probability decays exponentially with the source length, at a rate determined by the source correlations. This underscores the relevance of deriving the RDF when the SDC does not hold.
From the source-correlation perspective, Theorem~\ref{thm:max rho0} identifies the component-wise correlation for a 2-TC source to satisfy the SDC, revealing an elegant trade-off between source correlation and distortion constraints.
   \item Under the 2-TC covariance and 2-TD constraints, we provide complete and closed-form RDFs, refine our partial results for clarity, and analyze the compression problem under extreme correlation and in the asymptotic regime of infinite source length.
Specifically, we establish the structure of the optimal distortion matrix for this setting in Theorem~\ref{thm:dsm model}.
Theorem~\ref{thm:seven regions main} presents the RDF, by partitioning the entire distortion plane into seven distinct regions, with the optimal distortion allocations in each region rigorously established.
In Theorem~\ref{thm:iso tighter}, we derive a more tractable achievable upper bound for the RDF in a region, and show that it provably outperforms the known Hadamard bound.
An important insight from our theoretical results is that, under individual distortion criteria, the optimality of distortion allocation is achieved through the existence of correlations between the distortions of different components, which are overlooked by the Hadamard rate.
Moreover, the intricate relationships between distortion allocations across different regions highlight the critical importance of partitioning the distortion space to derive the RDF in general settings.
Interestingly, even when source components are positively correlated, their reconstruction distortions can be negatively correlated, suggesting that an overestimation in one component is likely to be accompanied by an underestimation in another.
The Gaussian-quadratic lossy compression problem we study, under extreme correlations, relates to some typical scenarios of interest, and in the asymptotic regime, it extends the existing two-component source compression problem.
   \item Under the isotropic correlation and identical constraints, Corollary~\ref{coro:alliso rd} shows that the RDF under individual distortion criteria is identical to that under the sum distortion criterion.
In the limit of large source length, an asymptotic expression for the average compression rate per component is provided.
When the SDC is satisfied, i.e., the source correlation is not sufficiently strong, the average rate converges to a non-zero quantity, jointly determined by the correlation and the distortion constraints.
In contrast, when the correlation is strong enough to violate the SDC, the average rate vanishes at a faster rate.
Our theoretical results quantitatively show how source correlations reduce storage and processing overhead while maintaining the desired reconstruction quality.
\end{itemize}

\subsubsection*{Organization}

The rest of the paper is organized as follows.
In Section~\ref{sec:System Model}, we introduce our problem setup and define the RDFs under both sum and individual distortion criteria.
In Section~\ref{sec:Prior Results}, we review the existing results on RDFs under both criteria.
Section~\ref{sec:Source Compression with Arbitrary Covariance and Distortion Constraints} covers the results under the arbitrary covariance and arbitrary distortion constraints.
Section~\ref{sec:Source Compression with 2-TC Covariance and Arbitrary Distortion Constraints} covers the results under the 2-TC covariance and arbitrary distortion constraints.
Section~\ref{sec:Source Compression with 2-TC Covariance and 2-TD Constraints} covers the results under the 2-TC covariance and 2-TD constraints.
Section~\ref{sec:Source Compression with Isotropic Correlation and Identical Constraints} covers the results under the isotropic correlation and identical constraints.
Section~\ref{sec:Conclusion} concludes the paper.

\subsubsection*{Notation}

Bold uppercase letters denote random vectors. Uppercase and lowercase sans-serif letters denote deterministic matrices and vectors, respectively. Calligraphic letters denote sets.
\(\rank(\av)\) denotes the number of non-zero entries in a vector \(\av\). We use \(\tr(\Am)\), \(\diag(\Am)\), \(\det(\Am)\), and \(\rank(\Am)\) to denote the trace, diagonal, determinant, and rank of a matrix \(\Am\), respectively.
\((\cdot)^{\T}\) denotes the transpose of a vector or matrix.
We denote \([x]^+ = \max\{x, 0\}\). Unless stated otherwise, all logarithms and exponentiations use base 2.
\(\mathbf{1}\) and \(\mathsf{I}\) represent the all-ones vector and the identity matrix, respectively.
For an integer \(k > 0\), let \([k]=\{1, 2, \cdots, k\}\).
We use \(\mathbb{E}[\cdot]\) to denote expectation and \(\mathcal{N}(\cdot, \cdot)\) to denote the Gaussian distribution.
For a real symmetric matrix \(\mathsf{A}\), let \(\mathsf{A}[\mathcal{I}]\) be the principal submatrix indexed by \(\mathcal{I} \subseteq [N]\).
We use \(\cdot \backslash \cdot\) to denote set subtraction.
For functions \(f(x)\) and \(g(x)\), \( f(x) = O(g(x)) \) means \( \limsup_{x \to \infty} \left|f(x)/g(x) \right| < \infty \), and \( f(x) = \Theta(g(x)) \) means \( \limsup_{x \to \infty} \left|f(x)/g(x)\right| = c \) where \( 0 < c < \infty \).
Let \(P_m(x) = \sum_{i=0}^m a_i x^i\) be a polynomial of degree \(m \in \mathbb{N}\) with coefficients \(a_i \in \mathbb{R}\) and \(a_m \neq 0\).

\section{System Model} \label{sec:System Model}

In this section, we first introduce the setup for the Gaussian-quadratic lossy compression problem, and then provide the definitions of the rate-distortion functions under the sum distortion criterion and individual distortion criteria, respectively.

\subsection{Setup} \label{subsec:Setup}

Consider an \(N\)-length vector Gaussian source \(\xv \sim \mathcal{N}(\mathsf{0}, \Rm)\) defined over the alphabet \(\mathcal{X}^N\), where each component of the source \(X_i\) is variance-normalized for \(i \in [N]\).
The eigenvalue decomposition (EVD) of the covariance matrix \(\Rm \succ \zerov\) is given by \(\Rm = \mathsf{U} \mathsf{\Lambda} \mathsf{U}^{\T}\), where \(\mathsf{\Lambda} = \diag(\lambda_1, \lambda_2, \cdots, \lambda_N)\) contains the eigenvalues of \(\Rm\) and \(\mathsf{U}\) is an orthogonal matrix whose columns are the corresponding eigenvectors. 
In Gaussian-quadratic compression, distortion is quantified by the matrix
\begin{align} \label{eq:distortion matrix definition}
   \Dm = \mathbb{E}\!\left[(\xv - \widehat{\xv})(\xv - \widehat{\xv})^{\T}\right],
\end{align}
where \(\widehat{\xv}\) is the reconstruction of \(\xv\).
For each source component \(X_{i}\), the distortion \(d_i\) (i.e., the \(i\)-th diagonal entry of \(\Dm\)) satisfies \(d_i \leq e_i\) with the distortion constraint \(e_i \in (0,1]\).
The vector \(\ev = [e_1, e_2, \cdots, e_N]^{\T}\) represents the ordered sequence of individual distortion constraints imposed on the components.
\begin{remark}[Justification of variance normalization]
   Notice that the component variance normalization does not involve any loss of generality.
Let \(\widetilde{\xv} \sim \mathcal{N}(\mathsf{0}, \mathsf{\Sigma}_{\widetilde{\xv}})\) be any non-degenerate \(N\)-length source with the covariance \(\mathsf{\Sigma}_{\widetilde{\xv}} = \mathsf{V}^{\frac{1}{2}} \Rm \mathsf{V}^{\frac{1}{2}} \succ \zerov\), where \(\mathsf{V} = \diag(\sigma_1^2, \sigma_2^2, \cdots, \sigma_N^2)\) with \(\sigma_i^2 = \mathbb{E}[\widetilde{X}_i^2]\). 
For any \(\widetilde{d}_i > 0\) with \(i \in [N]\), 
consider the distortion constraint \(\mathbb{E}[(\widetilde{X}_i - \widehat{\widetilde{X}}_i)^2] \le \widetilde{d}_i\) on each component, 
where \(\widehat{\widetilde{X}}_i\) is the reconstruction of \(\widetilde{X}_i\). Define the distortion-variance ratio as
\begin{align} \label{eq:e_def}
   e_i \triangleq \frac{\widetilde{d}_i}{\sigma_i^2}.
\end{align}
Then the original distortion constraint is equivalent to the constraint \(\mathbb{E}[(X_i - \widehat{X}_i)^2] \le e_i\), where we define \(X_i \triangleq e_i^{\frac{1}{2}} \widetilde{X}_i/\widetilde{d}_i^{\;\frac{1}{2}}\) and \(\widehat{X}_i \triangleq e_i^{\frac{1}{2}} \widehat{\widetilde{X}}_i/\widetilde{d}_i^{\;\frac{1}{2}}\). Due to \(e_i\) in Eq.~\eqref{eq:e_def}, we have \(X_i \sim \mathcal{N}(0,1)\) for any \(i \in [N]\). 
\end{remark}

\subsection{Definition of RDF} \label{subsec:Definition of RDF}

Let \(\{\xv(t)\}_{t=1}^\infty\) be an independent and identically distributed (i.i.d.) vector-valued Gaussian source, with each \(\xv(t) \sim \mathcal{N}(\zerov, \Rm)\) for all \(t \geq 1\). 
We consider \(k\) i.i.d. realizations of the \(i\)-th source component as \(\xv_i^k \triangleq [X_i(1), \cdots, X_i(k)]^{\T}\) for \(i \in [N]\), and define the concatenated source and reconstruction matrices across all \(N\) components as \(\xv^k \triangleq [\xv_1^k, \cdots, \xv_N^k] \in \mathcal{X}^{k \times N}\) and \(\widehat{\xv}^k \triangleq [\widehat{\xv}_1^k, \cdots, \widehat{\xv}_N^k] \in \widehat{\mathcal{X}}^{k \times N}\), respectively.
An \((N, k, 2^{kR})\) lossy code consists of an encoder \(f: \mathcal{X}^{k \times N} \mapsto \{1,2, \cdots, \lfloor 2^{kR} \rfloor\}\) and a decoder \(g: \{1,2, \cdots, \lfloor 2^{kR} \rfloor\} \mapsto \widehat{\mathcal{X}}^{k \times N}\).
The block-wise average distortion is defined as
\begin{align}
d (\xv_{i}^{k}, \widehat{\xv}_{i}^{k}) \triangleq \frac{1}{k} \sum_{t=1}^k \!\left(X_{i}(t) - \widehat{X}_i(t)\right)^{2}.
\end{align}

\begin{definition}[RDF under sum distortion criterion]
The pair \((R^{\mathrm{sum}}, d_{0})\) is said to be achievable if there exists a sequence of \((N, k, 2^{kR^{\mathrm{sum}}})\) codes such that for all \(\varepsilon^{\mathrm{sum}} > 0\) and all sufficiently large \(k\),
\begin{align} \label{eq:scalar criterion}
   \sum_{i=1}^{N} \mathbb{E}\!\left[ d(\xv_{i}^{k}, \widehat{\xv}_{i}^{k}) \right] \leq d_0 + \varepsilon^{\mathrm{sum}}
\end{align}
holds. The rate-distortion function under sum distortion criterion is defined as
\begin{align} \label{eq:scalar_rdf_definition}
   R^{\mathrm{sum}}([N], d_0) \triangleq \inf\{ R^{\mathrm{sum}} : (R^{\mathrm{sum}}, d_0) \text{ is achievable} \}.
\end{align}
\end{definition}

\begin{definition}[RDF under individual distortion criteria]
The pair \((R^{\mathrm{ind}}, \ev)\), where \(\ev = [e_1, e_2, \cdots, e_N]^{\T}\), is said to be achievable if there exists a sequence of \((N, k, 2^{kR^{\mathrm{ind}}})\) codes such that for all \(\varepsilon_i^{\mathrm{ind}} > 0\) with \(i \in [N]\) and all sufficiently large \(k\),
\begin{align} \label{eq:vector criterion}
\mathbb{E}\!\left[d(\xv_{i}^{k}, \widehat{\xv}_{i}^{k})\right] \leq e_{i} + \varepsilon_i^{\mathrm{ind}}
\end{align}
holds. The rate-distortion function under individual distortion criteria is defined as
\begin{align} \label{eq:vector_rdf_definition}
   R^{\mathrm{ind}}([N], \ev) \triangleq \inf\{ R^{\mathrm{ind}} : (R^{\mathrm{ind}}, \ev) \text{ is achievable} \}.
\end{align}
\end{definition}

We denote the distortion matrix \(\Dm\) in Eq.~\eqref{eq:distortion matrix definition} as \(\Dsm\) when the RDF is achieved.
It is evident that Eq.~\eqref{eq:vector criterion} imposes an individual distortion constraint on each component, thereby reducing the feasible set of the distortion (rate) allocations compared to the sum distortion case in Eq.~\eqref{eq:scalar criterion}, i.e.,
\begin{align} \label{eq:rdf vector scalar}
   \Rx^{\mathrm{ind}}([N], \ev)
   \geq
   \Rx^{\mathrm{sum}}\Big([N], {\textstyle \sum\limits_{i \in [N]}} e_i\Big).
\end{align}

\section{Prior Results} \label{sec:Prior Results}

In this section, we review the existing results on the rate-distortion functions for the vector Gaussian source under both the sum distortion criterion and individual distortion criteria, respectively.

\subsection{RDF under Sum Distortion Criterion} \label{subsec:RDF under Sum Distortion Criterion}


\begin{lemma}[RDF under sum distortion criterion,~\cite{thomas2006elements}] \label{lemma:trace rdf}
Let \(\xv \sim \mathcal{N}(\mathsf{0}, \Rm)\) and \(\Rm\) has eigenvalues \(\lambda_1 \ge \cdots \ge \lambda_N > 0\). \(\Rx([N], d_{0})\) in Eq.~\eqref{eq:scalar_rdf_definition} is
\begin{align} \label{eq:trace rdf}
\Rx^{\mathrm{sum}}([N], d_0) = \frac{1}{2} \sum_{i=1}^N \log \frac{\lambda_i}{d_{i}},
\end{align}
where \(d_i = \min\{\delta, \lambda_i\}\), and \(\delta = (d_0 - \sum_{i=n+1}^N \lambda_i)/n\), with \(n \in [N]\) being the largest index such that \(\lambda_n \ge \delta\).
\end{lemma}
For the parallel source \(\xv^{p} \sim \mathcal{N}(\zerov, \mathsf{\Lambda})\), the optimal distortion allocation \(\Dsm_p = \diag(d_1, \cdots, d_N) \preceq \mathsf{\Lambda}\) follows the classical reverse water-filling principle: equal distortion \(\delta\) is allocated to components with variances above \(\delta\). For the original source \(\xv \sim \mathcal{N}(\zerov, \Rm)\) with \(\Rm = \mathsf{U} \mathsf{\Lambda} \mathsf{U}^{\T}\), the optimal distortion matrix is \(\Dsm = \mathsf{U} \Dsm_p \mathsf{U}^{\T} \preceq \Rm\). Under sum distortion constraint, the RDF of \(\xv\) is obtained from that of \(\xv^p\) via EVD, which converts source correlations into variance differences across parallel components.

\subsection{RDF under Individual Distortion Criteria} \label{subsec:RDF under Individual Distortion Criteria}

The following lemma provides an optimization formulation of the RDF under individual distortion criteria.
\begin{lemma}[Formulation of RDF under individual distortion criteria,~\cite{xiao2005compression}] \label{lemma:maxdet}
For the vector source \(\xv \sim \mathcal{N}(\zerov, \Rm)\) with \(\Rm \succ \zerov\), \(\Rx([N], \ev)\) in Eq.~\eqref{eq:vector_rdf_definition} is given by the solution to\footnote{In subsequent contexts where it is clear, the superscript on the RDF that denotes individual distortion criteria is omitted for simplicity.}
\begin{problem} \label{problem:maxdet}
\begin{align}
\min_{\Dm=\Dm^{\T}} \quad & \frac{1}{2} \log \det (\Dm^{-1} \Rm), \label{eq:maxdet obj} \\
\mathrm{s.t.} \quad & \dv \le \ev, \label{eq:maxdet c1} \\
& \zerov \prec \Dm \preceq \Rm, \label{eq:maxdet c2}
\end{align}
\end{problem}
where \(\dv\) is the vector of diagonal entries of \(\Dm\) in Eq.~\eqref{eq:distortion matrix definition}, and \(\ev = [e_1, \cdots, e_N]^{\T}\) is the normalized distortion constraint vector. \(\dv \le \ev\) is shorthand for \(d_i \le e_i\) for all \(i\in[N]\).
\end{lemma}
\stepcounter{equation}

The matrix constraint in Eq.~\eqref{eq:maxdet c2} ensures the achievability of the RDF.
We denote by \(\dsv\) the vector formed by the diagonal entries of \(\Dsm\), where \(\Dsm\) is the optimal solution to the Max-Det problem in Eq.~\eqref{problem:maxdet}.

\begin{remark}[Sensitivity of individual distortion criteria] \label{remark:sensitivity_individual_constraint}
Compared with a sum distortion constraint, individual distortion constraints impose restrictions on the distortion of each component. This finer granularity renders the resulting RDF more sensitive to source correlations and thus motivates a more explicit characterization of how source correlations affect the compression rate.
Consider two vector sources \(\xv_1\) and \(\xv_2\) with covariance matrices
\begin{align}
   \Rm_1 = \begin{pmatrix}  
   1 & 0.3 & 0.2 \\  
   0.3 & 1 & 0.7 \\  
   0.2 & 0.7 & 1  
   \end{pmatrix}, \quad  
   \Rm_2 = \begin{pmatrix}  
   1.000 & 0.272 & 0.706 \\  
   0.272 & 1.000 & 0.218 \\  
   0.706 & 0.218 & 1.000  
   \end{pmatrix},
\end{align}   
which both have eigenvalues  
\(\lambda \approx \{0.292, 0.860, 1.848\}\).
Under individual distortion constraint \(\Em = \diag(0.60,0.22,0.18)\), the RDFs for the two sources are 2.144 bits/symbol and 2.294 bits/symbol, respectively.  
However, under sum distortion constraint with \(d_0 = \sum_{i} e_i = 1\) (cf. Eq.~\eqref{eq:scalar criterion}), the RDFs for both are 1.832 bits/symbol, despite the two sources having different component-wise correlations.
\end{remark}

For a vector Gaussian source of length \(N = 2\), the Max-Det problem admits an analytical solution~\cite{xiao2005compression,nayak2010successive}. For any \(N \geq 2\), the Hadamard inequality provides a simple lower bound for the RDF in Eqs.~\eqref{eq:hadamard bound} and~\eqref{eq:hadamard bound value}.
When the SDC in Eq.~\eqref{eq:SDC} holds, the exact RDF is obtained and the distortion for individual components is fully allocated, i.e., every component's distortion constraint is active.
However, when the SDC is not satisfied, the Hadamard rate in Eq.~\eqref{eq:hadamard bound value} becomes unachievable. Specifically, there may be the case where \(\det(\Em) > \det(\Rm)\), resulting in a non-positive rate that renders the Hadamard bound in Eq.~\eqref{eq:hadamard bound} meaningless. As we later show, the SDC generally fails to hold, significantly diminishing the utility of the Hadamard bound. 
To our knowledge, existing literature has not provided further results for arbitrary \(N\) when the SDC is not satisfied, including both the closed-form RDF and the optimal distortion allocation. 
In this paper, we provide a comprehensive analysis of this open problem and offer solutions to these unresolved issues.

\setcounter{theorem}{0}

\section{Source Compression with Arbitrary Covariance and Distortion Constraints} \label{sec:Source Compression with Arbitrary Covariance and Distortion Constraints}

In this section, we consider the case where both the covariance matrix \(\Rm\) and the distortion constraint matrix \(\Em\) (resp. vector \(\ev\)) are arbitrary. We obtain the spectral properties of optimal source reconstruction achieving the RDF and establish a stronger version of the SDC expressed as a scalar inequality.

\subsection{Spectral Properties of Optimal Reconstruction} \label{subsec:Spectral Properties of Optimal Reconstruction}

For the Gaussian-quadratic lossy compression under individual distortion criteria,
the source \(\xv\) can be modeled through a backward test channel as~\cite{xiao2005compression,nayak2010successive}
\begin{align} \label{eq:backward test channel}
   \xv = \xvh + \zv,
\end{align}
where the reconstruction \(\xvh \sim \mathcal{N}(\mathsf{0}, \Kxh)\) and the noise \(\zv \sim \mathcal{N}(\mathsf{0}, \Kz)\) are independent, and their covariance matrices satisfy
\begin{align} \label{eq:test channel cov}
\Rm = \Kxh + \Kz.
\end{align}
When the RDF is achieved, we denote \(\xvh = \xvhs\) and \(\zv = \zvs\), with the positive definite (PD) covariance matrix \(\Kzs = \Dsm\). Therefore, the covariance of the optimal reconstruction \(\xvhs\) is \(\Kxs = \Rm - \Dsm\). 

In this work, we take the counts of positive, negative, and zero eigenvalues of \(\Kxs\) as the properties of its spectrum. We use \(n_+\), \(n_0\), and \(n_-\) to represent these counts, including multiplicities. After the EVD, the vector \(\xvhs\) can be regarded as composed of \(n_+(\Kxs)\) independent components together with \(n_0(\Kxs)\) components that are linear combinations of them, which we refer to as trivial components. The reconstructed vector \(\xvhs\) is non-degenerate if \(n_0(\Kxs)=0\), i.e., its covariance matrix is PD.
In fact, we always have \(n_-(\Kxs) = 0\) to ensure the physical existence of \(\xvhs\).
The following theorem relates the SDC in Eq.~\eqref{eq:SDC} to the number of independent components of the optimal reconstruction.
\begin{theorem} \label{thm:det(Rm-Dsm)}
   For an \(N\)-dimensional Gaussian source \(\xv \sim \mathcal{N}(\zerov, \Rm)\) with \(\Rm \succ \zerov\), and a diagonal distortion constraint \(\Em = \diag(\ev)\), the optimal source reconstruction \(\xvhs\) that achieves the rate-distortion function has a covariance matrix \(\Kxs\) that satisfies
\begin{align} \label{eq:det(Rm-Dsm)}
   n_{+}(\Kxs)
   \begin{cases}
   = N, & \Rm \succ \Em, \\
   < N, & \text{otherwise}.
   \end{cases}
\end{align}
\end{theorem}
\begin{IEEEproof}[Proof Sketch]
We first show that strong duality holds for the RDF optimization problem in Lemma~\ref{lemma:maxdet}, and that the Karush-Kuhn-Tucker (KKT) conditions guarantee the optimality of our results.
In the first case of Eq.~\eqref{eq:det(Rm-Dsm)}, we prove that the distortion matrix \(\Dm=\Em\) is optimal.
In the second case, we further distinguish whether the SDC is satisfied or violated.
If the SDC holds on the boundary, we show that \(\Dm=\Em\) remains optimal but leads to a rank-deficient reconstruction covariance.
If the SDC is violated, we establish the same conclusion by contradiction.
See Appendix~\ref{proof:thm:det(Rm-Dsm)} for the complete proof.
\end{IEEEproof}

In this theorem, we show that the SDC is satisfied and inactive, i.e., \(\Rm \succ \Em\) if and only if the optimal source reconstruction is non-degenerate, 
in which case the Hadamard rate in Eq.~\eqref{eq:hadamard bound value} exactly equals the RDF.
Conversely, when the SDC holds on the boundary or not satisfied, 
fewer than \(N\) independent components are sufficient to reconstruct a source with \(N\) independent components.
The following theorem provides more refined upper bounds on \(n_{+}(\Kxs)\).
Before formally stating this theorem, we first introduce a reconstruction vector \(\xvh_{\mathrm{H}}\) with nominal covariance \(\Rm - \Em\), which achieves the Hadamard rate in Eq.~\eqref{eq:hadamard bound value} under the assumed distortion allocation \(\Dm = \Em\). Note that when the SDC is not satisfied, we have \(n_{-}(\Rm - \Em) > 0\).

\begin{theorem} \label{thm:inertia relation}
   For an \(N\)-dimensional Gaussian source \(\xv \sim \mathcal{N}(\zerov, \Rm)\) with \(\Rm \succ \zerov\), and a diagonal distortion constraint \(\Em = \diag(\ev)\), the optimal source reconstruction \(\xvhs\) that achieves the rate-distortion function has a covariance matrix \(\Kxs\), and the achieved distortion vector is \(\dsv\). We have
\begin{align} \label{eq:inertia relation}
   n_+(\Kxs) &\leq \min\{ N - \rank(\ev - \dsv), n_+(\Rm - \Em)\},
\end{align}
where \(\rank(\ev - \dsv)\) denotes the number of non-zero entries in the vector \(\ev - \dsv\).
\end{theorem}
\begin{IEEEproof}[Proof Sketch]
We first derive a fundamental identity for \(\Kxs\) from the KKT conditions.
This identity relates \(n_+(\Kxs)\) to the support size of the dual variables associated with the individual distortion constraints, which establishes the first upper bound.
To prove the second upper bound, we establish a component-wise upper bound on \(\dsv\).
We then construct an elegant subspace whose dimension equals \(n_+(\Kxs)\), and show that \(\Rm-\Em \succ \zerov\) on this subspace using the component-wise upper bound.
Finally, by invoking the variational characterization of the number of positive eigenvalues of a matrix, the second bound is derived.
See Appendix~\ref{proof:thm:inertia relation} for the complete proof.
\end{IEEEproof}

Theorem~\ref{thm:inertia relation} shows that \(n_{+}(\Kxs)\), the number of independent components in the optimal reconstruction, is upper-bounded by two quantities:
\(N-\rank(\ev-\dsv)\) and \(n_+(\Rm-\Em)\). 
\(N-\rank(\ev-\dsv)\) represents the number of components with active distortion constraints in the optimal reconstruction, i.e., the number of constraints in Eq.~\eqref{eq:maxdet c1} that hold with equality.
\(n_+(\Rm-\Em)\) represents the number of physically meaningful independent components in the reconstruction \(\xvh_{\mathrm{H}}\).
Each of the two upper bounds in Eq.~\eqref{eq:inertia relation} is analyzed in turn below.

The first upper bound links the imposed distortion constraints to the independent components in the optimal reconstruction.
Intuitively, an increase in the number of mild constraints (i.e., the number of large components in \(\ev\) increases) allows more distortion to be passively tolerated, thereby reducing the number of independent components in the optimal reconstruction.
   Moreover, in contrast to the sum distortion criterion, the individual distortion criteria exhibit quite intricate behavior regarding the strictness of distortion constraints.
Specifically, under sum distortion criterion, for a non-zero RDF and a bounded distortion constraint, the condition \(\sum_{i=1}^{N} d_{i} \leq \sum_{i=1}^{N} e_{i}\) can be replaced by equality~\cite[Corollary 8.19]{yeung2008information}.
However, under individual distortion criteria, such a replacement is clearly not possible in general for every constraint \(d_i \le e_i\) in Eq.~\eqref{eq:maxdet c1}. Still, the first upper bound implies that at least one distortion constraint for a source component is strict.
The reason is as follows. When the RDF is non-zero, we have \(n_{+}(\Kxs) \ge 1\), and it can be deduced that \(\rank(\ev - \dsv) \leq N-1\) holds, which means that at least one distortion constraint achieves equality, i.e., \(d_i = e_i\), for some \(i \in [N]\).

The second upper bound of \(n_{+}(\Kxs)\) is then given by \(n_+(\Rm-\Em)\).
When the SDC holds, \(\Dsm = \Em\) and \(\dsv = \ev\) follow, and this bound is achieved with equality.
Recall that we set \(\Dm=\Em\) to obtain the reconstruction \(\xvh_{\mathrm{H}}\) with the nominal covariance \(\Rm-\Em\).
When the SDC is not satisfied, i.e., \(\Rm - \Em \nsucceq \zerov\), \(\xvh_{\mathrm{H}}\) becomes physically unrealizable and is no longer the optimal reconstruction \(\xvhs\). However, it is revealed that the number of independent components in \(\xvhs\) does not exceed that of the physically meaningful components in \(\xvh_{\mathrm{H}}\).
It is straightforward to deduce that one way to achieve efficient compression is to find a reconstruction with fewer independent components when the SDC does not hold. These reconstruction components can simultaneously represent multiple source components within certain distortion constraints at a lower compression rate.

\subsection{A Stronger Version of SDC} \label{subsec:A Stronger Version of SDC}

The SDC in Eq.~\eqref{eq:SDC}, expressed using the Loewner partial order, has a key limitation: at the element-wise level, it fails to reflect the relationship between the variance of each component and its corresponding individual distortion constraint. To address this, we reformulate the SDC under a total order.\footnote{A partial order allows comparison between some pairs of elements but not necessarily all, while a total order ensures that every pair of elements is comparable.} A stronger version of the SDC is presented as a scalar inequality, which closely resembles the classical reverse water-filling principle.

\begin{proposition} \label{prop:weyl bound}
   Let the source covariance matrix \(\Rm\) and the diagonal distortion constraint matrix \(\Em\) have eigenvalues \(\{\lambda_i\}_{i=1}^N\) and \(\{e_i\}_{i=1}^N\) arranged in non-increasing order. A stronger version of the SDC \(\Rm - \Em \succeq \zerov\) is
\begin{align}  \label{eq:sufficient cond SDC1}
e_1 \leq \lambda_N.
\end{align}
\end{proposition}
\begin{IEEEproof}
By Weyl's inequality~\cite{horn2012matrix}, for any integers \(i\) and \(j\) satisfying \(1 \le i, j \le N\) and \(N + 1 \le i + j\), the eigenvalues of \(\Rm - \Em\) satisfy \(\lambdav_{i+j-N}(\Rm - \Em) \ge \lambdav_i(\Rm) + \lambdav_j(-\Em)\), where \(\lambdav_i(\Am)\) denotes the \(i\)-th largest eigenvalue of \(\Am\).
We set \(i = N\) and \(j = N\) to obtain the smallest eigenvalues of \(\Rm\) and \(-\Em\), which are \(\lambda_N\) and \(-e_1\), respectively. When Eq.~\eqref{eq:sufficient cond SDC1} is satisfied, it follows that all eigenvalues of \(\Rm - \Em\) are non-negative and the SDC is satisfied.
\end{IEEEproof}

Proposition~\ref{prop:weyl bound} provides a rough interpretation of the SDC: the individual distortion constraints should not be too mild.
Specifically, the mildest distortion constraint across all components is no greater than the smallest eigenvalue of the covariance matrix.
When this condition holds, the optimal distortion allocation is \(\Dsm = \Em\), and the Hadamard rate in Eq.~\eqref{eq:hadamard bound value} becomes the RDF. 
Interestingly, from this perspective, this condition can be seen as a counterpart to the reverse water-filling principle in Lemma~\ref{lemma:trace rdf} for the small distortion regime, where the distortion allocation follows a simple rule.
Concretely, under sum distortion criterion, small distortion means that the equally allocated distortion does not exceed the smallest variance among the parallel sources, i.e., \(\delta \leq \lambda_N\) in Lemma~\ref{lemma:trace rdf}. For small distortions, the optimal allocation is equal distortion under sum distortion criterion, and full distortion per component under individual distortion criteria.

\section{Source Compression with 2-TC Covariance and Arbitrary Distortion Constraints} \label{sec:Source Compression with 2-TC Covariance and Arbitrary Distortion Constraints}

In this section, to make progress on the analytical investigation of the RDF with individual distortion constraints, we consider a particular class of covariance matrices characterized by two types of correlation, referred to as the 2-TC class. Under this class of covariance matrices, when the SDC is satisfied, 
a closed-form RDF explicitly incorporating correlations is derived.
Furthermore, we obtain an equivalent form of the SDC, and examine the SDC region from the perspectives of both distortion constraints and source correlations.

\subsection{2-TC Class for Covariance Modeling} \label{subsec:2-TC Class for Covariance Modeling}

To quantify the component-wise correlations of the source and develop the tractable and closed-form RDF, we consider a length-\(N\) vector Gaussian source with hierarchical \(n\)-type correlations, referred to as the \(n\)-TC class for all \(n \in [N-1]\).
\begin{definition}[\(n\)-TC class of covariance matrices] \label{def:nTC cov}
Let the source be ordered as \(\xv = [X_1, X_2, \cdots, X_N]\). For any \(n \in [N-1]\), define its variance-normalized covariance matrix as
\begin{align} \label{eq:nTC cov}
\Rm_n(i,j) \triangleq 
   \begin{cases}
   1, & \text{if } i = j, \\
   \rho_{\min\{i,j\}}, & \text{if } i \neq j \text{ and } (i < n \text{ or } j < n), \\
   \rho_{n}, & \text{if } i \neq j \text{ and } i,j \geq n,
   \end{cases}
\end{align}
where \(\Rm_n(i,j)\) denotes the \((i,j)\)-th entry of \(\Rm_n\), and \(\rho_1, \rho_2, \cdots, \rho_n \in [-1,1]\) are the component-wise correlations.
The \(n\)-type correlation class of covariance matrices is defined as
\(\mathcal{K}_n \triangleq \left\{ \Rm_n \in \mathbb{R}^{N \times N} : \Rm_n \succ \zerov \text{ has the structure in Eq.~} \eqref{eq:nTC cov}\right\}\).
\end{definition}

We consider a particular class of covariance matrices \(\mathcal{K}_2\), referred to as the \textit{2-TC class}. In this model, \(X_1\) is the \textit{central component}, and the rest are \textit{peripheral components}, with \(\rho_1\) denoting the central-to-peripheral correlation and \(\rho_2\) the correlation among peripheral components.
For the quadratic lossy compression of a vector Gaussian source with a 2-TC covariance \(\Rm_{2} \in \mathcal{K}_2\) under individual distortion criteria, the following lemma is essential for obtaining analytically tractable results.

\begin{remark}
Covariance matrices of the 2-TC form have also appeared extensively in the replica-symmetry (RS) analysis of high-dimensional inference and communication systems in statistical physics. In particular, under the RS assumption, the so-called dominant shell in the large-deviation rate function admits precisely this type of correlation structure due to symmetry considerations. Such covariance models arise, for example, in the asymptotic analysis of randomly spread code-division multiple access (CDMA) systems~\cite[Eqs.~(115a)-(115b)]{guo2005randomly}.
\end{remark}

\begin{lemma}[Parameterized determinant for 2-TC covariance class] \label{lemma:determinant}
   Let \(\Rm_2 \in \mathcal{K}_{2}\) and \(\Gammam = \diag(\gamma_1, \gamma_2, \cdots, \gamma_N)\) with \(\gamma_i \in [0,1]\).
The determinant of \(\Rm_{2} - \Gammam\), written as \(\Delta(\rho_1, \rho_2; \Gammam)\), is given by
\begin{align} \label{eq:Delta determinant}
   \Delta(\rho_1, \rho_2; \Gammam) = & \left( 1 - \gamma_1 + \sum_{i=2}^{N} \frac{\rho_2 - \rho_1^2 - \rho_2 \gamma_1}{1 - \rho_2 - \gamma_i} \right) 
\cdot \prod_{i=2}^{N} (1 - \rho_2 - \gamma_i).
\end{align}
\end{lemma}
\begin{IEEEproof}
   See Appendix~\ref{proof:lemma:determinant}.
\end{IEEEproof}

\(\Delta(\rho_1, \rho_2; \Gammam)\) reduces to \(\det(\Rm_{2})\) when \(\Gammam = \zerov\). If \(\Gammam = \gamma \mathsf{I}\), \(\Delta(\rho_1, \rho_2; \Gammam) = 0\) is the characteristic equation of \(\Rm_{2}\), whose zeros with respect to the variable \(\gamma\) yield the eigenvalues of \(\Rm_{2}\). 
The eigenvalues of \(\Rm_{2}\) are given by \(\lambda_1 = 1 - \rho_2\) with multiplicity \(N-2\), and \(\lambda_{2,3} = 1 + \frac{1}{2}(N-2) \rho_2 \pm \frac{1}{2}\sqrt{(N-2)^2 \rho_2^2 + 4 (N-1) \rho_1^2}\).
Thus, \(\Rm_{2} \succ \zerov\) if and only if \(\nu > 0\) holds, where
\begin{align} \label{eq:nu def}
   \nu = (N-2) \rho_2 + 1 - (N-1)\rho_1^2.
\end{align}
The isotropic correlation class of covariance matrices considered in~\cite{wagner2008rate,floor2015joint,xu2015lossy,saidutta2020vae} is a special case of our 2-TC covariance class.
Theoretical results for the \((N-1)\)-TC covariance class in Eq.~\eqref{eq:nTC cov} are generally intractable and elusive. The techniques developed in this paper yield concise results and offer fundamental insights under the 2-TC covariance class, which, importantly, can be generalized to the \((N-1)\)-TC covariance class.
Specifically, in the 2-TC covariance class, the \(N-1\) peripheral source components can be divided into a new central component and remaining peripheral components.  
This results in two central types and one peripheral type, giving rise to three levels of correlation and thus forming the 3-TC class (see \(n=3\) in Eq.~\eqref{eq:nTC cov}). Recursively applying this process extends the 2-TC class to the \((N-1)\)-TC class with \(N-1\) hierarchical correlation levels, where each higher-level component is equally correlated with all lower-level ones.

We now formally establish the feasibility of this recursive construction process.
With a slight abuse of notation, we use \(\Gammam\) and \(\overline{\Gammam}\) to denote the diagonal matrices associated with \(\Rm_n \in \mathcal{K}_{n}\) and \(\Rm_{n-1} \in \mathcal{K}_{n-1}\), respectively, whose entries \(\{\gamma_i\}_{i=1}^N\) and \(\{\overline{\gamma}_i\}_{i=1}^N\) lie in \([0,1]\).
The covariance matrix \(\Rm_n\) involves \(n\) correlation parameters \(\{\rho_i\}_{i=1}^n\), while \(\Rm_{n-1}\) involves \(n-1\) parameters \(\{\overline{\rho}_i\}_{i=1}^{n-1}\).
This extension fundamentally relies on deriving an explicit expression for \(\det(\Rm_n - \Gammam)\) using the already known explicit form of \(\det(\Rm_{n-1} - \overline{\Gammam})\) for arbitrary \(n \in [N-1]\). 
The following theorem establishes the relation from \(\det(\Rm_{n-1} - \overline{\Gammam})\) to \(\det(\Rm_n - \Gammam)\).

\begin{theorem}[The generalization of \(n\)-TC class] \label{thm:2TC_extension}
For \(n \in [N-1]\), let \(\Rm_{n-1} \in \mathcal{K}_{n-1}\) with \(\{\overline{\rho}_i\}_{i=1}^{n-1}\) and \(\overline{\Gammam} = \diag(\overline{\gamma}_1, \overline{\gamma}_2, \cdots, \overline{\gamma}_N)\) where \(\overline{\gamma}_i \in [0,1]\) for all \(i \in [N]\). Suppose \(\det(\Rm_{n-1} - \overline{\Gammam})\) is explicitly written as \(\Delta_{n-1}(\overline{\rho}_1, \cdots, \allowbreak \overline{\rho}_{n-1}; \overline{a}_1, \cdots, \overline{a}_N)\) with \(\overline{a}_i = 1 - \overline{\gamma}_i \in [0,1]\).
For any \(\Rm_n \in \mathcal{K}_n\) with \(\{\rho_i\}_{i=1}^n\) and \(\Gammam = \diag(\gamma_1, \gamma_2, \cdots, \gamma_N)\) where \(\gamma_i \in [0,1]\), \(\det(\Rm_n - \Gammam)\) can be computed from \(\det(\Rm_{n-1} - \overline{\Gammam})\) as follows.

We first set
\begin{align} 
\overline{\rho}_1 = \rho_1, \quad
\overline{\rho}_2 = \rho_2, \quad
\cdots, \quad
\overline{\rho}_{n-2} = \rho_{n-2}, \label{eq:extension T rho cond1} \\
\overline{a}_1 = a_1, \quad
\overline{a}_2 = a_2, \quad
\cdots, \quad
\overline{a}_{n-2} = a_{n-2}. \label{eq:extension T rho cond2}
\end{align}

The remaining parameters \(\overline{\rho}_{n-1}, \overline{a}_{n-1}, \cdots, \overline{a}_N\) are then determined through the equation
\begin{align} 
   &\sum_{j=1}^{N-n+2} \sum_{i=1}^{N-n+2} \overline{\mathsf{Q}}^{-1}_{ij} \notag \\
   &= \left(\frac{\rho_n}{\rho_{n-1}}\right)^2 \mathsf{Q}^{-1}_{11} + \left(\frac{\rho_n}{\rho_{n-1}}\right) \sum_{i=2}^{N-n+2} \mathsf{Q}^{-1}_{i1} + \left(\frac{\rho_n}{\rho_{n-1}}\right) \sum_{j=2}^{N-n+2} \mathsf{Q}^{-1}_{1j} + \sum_{j=2}^{N-n+2} \sum_{i=2}^{N-n+2} \mathsf{Q}^{-1}_{ij}. \label{eq:trans equival cond entry}
\end{align}
where \(\overline{\mathsf{Q}}^{-1}\) is given by
\begin{align} \label{eq:Q bar inverse entry}
\overline{\mathsf{Q}}^{-1}_{ij} = \delta_{ij} \frac{1}{\overline{a}_{n-2+i} - \overline{\rho}_{n-1}} - \frac{1}{1 + \overline{\rho}_{n-1} \sum_{k=n-1}^N \frac{1}{\overline{a}_k-\overline{\rho}_{n-1}}} \cdot \frac{\overline{\rho}_{n-1}}{(\overline{a}_{n-2+i}-\overline{\rho}_{n-1})(\overline{a}_{n-2+j}-\overline{\rho}_{n-1})},
\end{align}
where \(\delta_{ij}\) denotes the Kronecker delta, which equals \(1\) if \(i=j\) and \(0\) otherwise, with \(1 \leq i,j \leq N-n+2\). \(\mathsf{Q}^{-1}\) is given by
\begin{numcases}{} 
   \mathsf{Q}^{-1}_{11} = \frac{1}{a_{n-1}\left(\frac{\rho_{n}}{\rho_{n-1}}\right)^{2} - \rho_n} - \frac{\rho_n}{G\left[a_{n-1}\left(\frac{\rho_{n}}{\rho_{n-1}}\right)^{2}-\rho_n\right]^{2}}, \label{eq:Q inverse entry1} \\
   \mathsf{Q}^{-1}_{i1} = - \frac{\rho_n}{G\left(a_{n-2+i}-\rho_n\right)\left[a_{n-1}\left(\frac{\rho_{n}}{\rho_{n-1}}\right)^{2}-\rho_n\right]}, \quad \text{ for } 2 \leq i \leq N-n+2, \label{eq:Q inverse entry2} \\
   \mathsf{Q}^{-1}_{1j} = - \frac{\rho_n}{G\left[a_{n-1}\left(\frac{\rho_{n}}{\rho_{n-1}}\right)^{2}-\rho_n\right]\left(a_{n-2+j}-\rho_n\right)}, \quad \text{ for } 2 \leq j \leq N-n+2, \label{eq:Q inverse entry3} \\
   \mathsf{Q}^{-1}_{ij} = \delta_{ij} \frac{1}{a_{n-2+i} - \rho_{n}} - \frac{\rho_n}{G\left(a_{n-2+i}-\rho_n\right)\left(a_{n-2+j}-\rho_n\right)}, \quad \text{ for } 2 \leq i,j \leq N-n+2, \label{eq:Q inverse entry4}
\end{numcases}
where \(G = 1 + \rho_n \Big[ \Big(a_{n-1}\big(\frac{\rho_{n}}{\rho_{n-1}}\big)^{2}-\rho_n\Big)^{-1} + \sum_{k=n}^N \frac{1}{a_k-\rho_n}\Big]\).

Finally, \(\det(\Rm_n - \Gammam)\) is written as \(\Delta_n(\rho_1, \cdots, \rho_n; a_1, \cdots, a_N)\), given by
\begin{align}
   \Delta_n(\rho_1, \cdots, \rho_n; a_1, \cdots, a_N) = \left( \frac{\rho_{n-1}}{\rho_n} \right)^2 \frac{\det(\mathsf{Q})}{\det(\overline{\mathsf{Q}})} \Delta_{n-1}(\overline{\rho}_1, \cdots, \overline{\rho}_{n-1}; \overline{a}_1, \cdots, \overline{a}_N),
\end{align}
where
\begin{align} 
   \det(\overline{\mathsf{Q}}) &= \left( \overline{a}_{n-1} + \overline{\rho}_{n-1} \sum_{i=n}^{N} \frac{\overline{a}_{n-1} - \overline{\rho}_{n-1}}{\overline{a}_i - \overline{\rho}_{n-1}} \right) \cdot \prod_{i=n}^{N} (\overline{a}_i - \overline{\rho}_{n-1}), \label{eq:extension Q bar det} \\
   \det(\mathsf{Q}) &= \left( a_{n-1}\left(\frac{\rho_{n}}{\rho_{n-1}}\right)^{2} + \rho_{n} \sum_{i=n}^{N} \frac{a_{n-1}\left(\frac{\rho_{n}}{\rho_{n-1}}\right)^{2} - \rho_{n}}{a_i - \rho_{n}} \right) \cdot \prod_{i=n}^{N} (a_i - \rho_{n}).
\end{align} 
\end{theorem}
\begin{IEEEproof}[Proof Sketch]
   We utilize the block structures of \(\Rm_{n-1} - \overline{\Gammam}\) and \(\Rm_{n} - \Gammam\) to establish a relationship between their determinants by applying the Sherman-Morrison-Woodbury formula.  
   We then derive a sufficient condition for determining the associated parameters of \(\Rm_{n-1} - \overline{\Gammam}\) based on the given parameters of \(\Rm_{n} - \Gammam\).  
   See Appendix~\ref{proof:thm:2TC_extension} for the complete proof.
\end{IEEEproof}
In view of Theorem~\ref{thm:2TC_extension}, we now only need to focus on the covariance matrix \(\Rm_2\) in the 2-TC covariance class \(\mathcal{K}_2\). When the context is clear, the subscript of \(\Rm_2\) is omitted for simplicity.
All correlations are assumed to be non-negative to avoid tedious mathematical discussions, but the results readily extend to cases with negative correlations~\cite{lapidoth2010sending}.
One advantage of using individual distortion constraints is that different constraints can be imposed on source components according to their practical relevance. 
For the source components \(X_2, X_3, \cdots, X_N\) that share the same statistics, we assume without loss of generality that their distortion constraints satisfy \(e_2 \geq e_3 \geq \cdots \geq e_N\).

\subsection{Characterization of SDC Region} \label{subsec:Characterization of SDC Region}

When the source covariance matrix \(\Rm \in \mathcal{K}_{2}\) and the distortion constraints \(\Em = \diag(e_{1}, e_{2}, \cdots, e_{N})\) are arbitrary, we will delve into the relationship between correlations, distortion constraints, and the rate-distortion function within the SDC region, thereby providing new results and insights into the Hadamard lower bound in Eq.~\eqref{eq:hadamard bound} and the corresponding rate in Eq.~\eqref{eq:hadamard bound value}.
Specifically, the following theorem reformulates the SDC in Eq.~\eqref{eq:SDC} as scalar inequalities and offers a closed-form RDF, explicitly incorporating both the source correlations and the distortion constraints.

\begin{theorem} \label{thm:RE_psd_conditions}
For a 2-TC source covariance matrix \(\Rm \in \mathcal{K}_2\) and an arbitrary distortion constraint matrix \(\Em\), assuming the entries of \(\Em\) satisfy \(e_{2} \geq \cdots \geq e_{N}\), the SDC \(\Rm - \Em \succeq \zerov\) holds if and only if
\begin{align}
   e_3 &\leq 1 - \rho_{2}, \label{eq:e3_condition} \\
   e_2 &\leq 1-\frac{\chi_{3}\rho_{2}^{2}}{1+\chi_{3}\rho_{2}}, \label{eq:e2_condition} \\
   e_1 &\leq 1-\frac{\chi_{2}\rho_{1}^{2}}{1+\chi_{2}\rho_{2}} \label{eq:e1_condition}
\end{align}
hold, where \(\chi_{j}=\sum_{i=j}^{N} \frac{1}{1 - \rho_{2} - e_i}\). The corresponding rate-distortion function is given by
\begin{align}
\Rx([N], \ev) &= \sum_{i=2}^{N} \frac{1}{2} \log \frac{1-\rho_{2}}{e_i} + \frac{1}{2} \log \frac{1 + (N-1)\rho_{\times}}{e_1}, \label{eq:rd sdc1}
\end{align}
where \(\rho_{\times} = \frac{\rho_{2} - \rho_1^2}{1 - \rho_{2}}\). When the correlation is isotropic, i.e., \(\rho=\rho_{1}=\rho_{2}\), the above rate-distortion function simplifies to
\begin{align} \label{eq:rd sdc iso}
\Rx([N], \ev) = \sum_{i=1}^{N} \frac{1}{2} \log \frac{1-\rho}{e_i} + \frac{1}{2} \log \frac{1 + (N-1)\rho}{1-\rho}.
\end{align}
\end{theorem}
\begin{IEEEproof}[Proof Sketch]
   We decompose and analyze all of the principal minors of \(\Rm - \Em\) and invoke Lemma~\ref{lemma:determinant} to obtain an equivalent form of the SDC.
Subsequently, we expand the determinant of \(\Rm\) as a function of the correlations to establish the RDF. See Appendix~\ref{proof:thm:RE_psd_conditions} for the complete proof.
\end{IEEEproof}

In light of Theorem~\ref{thm:RE_psd_conditions}, we quantitatively observe the role of source correlations in reducing the optimal compression rate within the considered 2-TC covariance class.
In particular, stronger inter-component correlations lead to a lower compression rate under the same individual distortion constraints, with the reduction exhibiting a logarithmic-type dependence on the correlation strength.
Moreover, the complicated interactions between all individual distortion constraints \(\{e_i\}_{i=1}^{N}\) and correlations \(\rho_1\) and \(\rho_2\) are captured in Eqs.~\eqref{eq:e3_condition}-\eqref{eq:e1_condition}, which define the SDC region at the element-wise level.
This region can be analyzed from two perspectives: distortion constraints and source correlations.

\subsubsection{A Distortion-Constraint Perspective}
\hspace{0mm}
For a given correlated Gaussian source \(\xv \sim \mathcal{N}(\mathsf{0},\Rm)\) with \(\Rm \in \mathcal{K}_{2}\), we quantify the probability of the event \(\mathcal{A}_{0} = \{\mathbf{E} \preceq \Rm\}\) under the assumption that all diagonal entries of \(\mathbf{E}\) are independently drawn from a uniform distribution.
\begin{theorem} \label{thm:sdc exponential asymp}
For a 2-TC source covariance matrix \(\Rm \in \mathcal{K}_2\), when the individual distortion constraints \(\{e_i\}_{i=1}^{N}\), i.e., \(\mathbf{E}=\diag(e_1,\cdots,e_N)\), are independently and uniformly distributed as \(e_i \sim \mathrm{Unif}[0,1]\), the SDC satisfaction probability takes the form of an \(N\)-fold integral
\begin{align}
   P(\mathcal{A}_0) 
   = (N-1)\underbrace{\int_{0}^{1-\rho_{2}}  \cdots \int_{0}^{1-\rho_{2}}}_{N-2 \text{ times}}  \int_{\max\{e_{3},\cdots,e_{N}\}}^{1-\rho_{2}-\frac{\rho_{1}^{2}-\rho_{2}}{1-\left(\rho_{1}^{2}-\rho_{2}\right)\chi_{3}}}  \int_{0}^{1-\frac{\chi_{2}\rho_{1}^{2}}{1+\chi_{2}\rho_{2}}} \, \dd e_{1} \dd e_{2} \dd e_{3} \cdots \dd e_{N}. \label{eq:SDC prob}
\end{align} 
An explicit analytical upper bound involving a single integral is given by
\begin{align}
   P(\mathcal{A}_0) \leq{}& \left(1-\rho_{2}\right)^{N-2}\biggg[\left(1-\frac{\left(N-2\right)\rho_{1}^{2}}{1+\left(N-3\right)\rho_{2}}\right)\left(1-\rho_{2}-\frac{(N-1)\left(\rho_{1}^{2}-\rho_{2}\right)}{1-\left(\rho_{1}^{2}-\rho_{2}\right)\frac{N-2}{1-\rho_{2}}}\right) \notag \\
   & \qquad \qquad \qquad +\frac{N-1}{\left(1+\rho_{2}\frac{N-2}{1-\rho_{2}}\right)^{2}}\log\!\left(\frac{\rho_{1}^{2}-\rho_{2}}{1-\left(\rho_{1}^{2}-\rho_{2}\right)\frac{N-2}{1-\rho_{2}}}+\frac{\rho_{2}\left(1-\rho_{2}\right)}{1+\left(N-3\right)\rho_{2}}\right)\biggg] \notag \\
   &- \frac{(N-1)(N-2)\rho_{1}^{2}}{\left(1+\rho_{2}\frac{N-2}{1-\rho_{2}}\right)^{2}}\int_{0}^{1-\rho_{2}}  \log\!\left(\left(1-\rho_{2}-e_{n}\right)+\frac{\rho_{2}\left(1-\rho_{2}\right)}{1+\left(N-3\right)\rho_{2}}\right) e_{n}^{N-3}  \, \dd e_{n} \label{eq:SDC prob N-2 upper final}
\end{align} 
and its asymptotic approximation as \(N \to \infty\) is
\begin{align}
   P(\mathcal{A}_0) = e^{N\log\left(1-\rho_{2}\right)}\Bigg[2\left(1-\frac{\rho_{1}^{2}}{\rho_{2}}\right)\left(1-\rho_{2}\right)+O\!\left(\frac{\log N}{N}\right)\Bigg]. \label{eq:SDC prob N-2 upper final Harmo asym}
\end{align} 

\end{theorem}
\begin{IEEEproof}[Proof Sketch]
   The \(N\)-fold integral in Eq.~\eqref{eq:SDC prob} follows from Eqs.~\eqref{eq:e3_condition}-\eqref{eq:e1_condition} in Theorem~\ref{thm:RE_psd_conditions}. By decoupling two key variables, we first upper bound the integrand and then reduce the high-dimensional integral to a single-variable form in Eq.~\eqref{eq:SDC prob N-2 upper final}. Applying the asymptotic approximation of the \(N\)-th harmonic number then yields Eq.~\eqref{eq:SDC prob N-2 upper final Harmo asym}. See Appendix~\ref{proof:thm:sdc exponential asymp} for the complete proof.
\end{IEEEproof}
Under i.i.d. uniform distortion constraints, Theorem~\ref{thm:sdc exponential asymp} provides the exact expression of \(P(\mathcal{A}_0)\) in Eq.~\eqref{eq:SDC prob}. An upper bound is given in Eq.~\eqref{eq:SDC prob N-2 upper final}, and its asymptotic form as \(N \to \infty\) is given in Eq.~\eqref{eq:SDC prob N-2 upper final Harmo asym}.
It follows that the probability of satisfying the SDC decays exponentially with the source length \(N\) at a rate of \(-\log(1-\rho_{2})\).
Fig.~\ref{fig:SDC_prob_N_2} shows the simulation of the SDC satisfaction probability \(P(\mathcal{A}_0)\) versus the source length \(N\), with central correlation \(\rho_1 = 0.45\) fixed.
For each \(N\), we conduct \(10^{6}\) Monte Carlo trials to evaluate \(P(\mathcal{A}_0)\), with all \(e_i \sim \mathrm{Unif}[0,1]\) for \(i \in [N]\).
Our analytical upper bound in Eq.~\eqref{eq:SDC prob N-2 upper final} and the asymptotic approximation in Eq.~\eqref{eq:SDC prob N-2 upper final Harmo asym} yield highly accurate results for the SDC satisfaction probability. Additionally, the approximation is significantly more concise than the \(N\)-fold integral in Eq.~\eqref{eq:SDC prob}. As the peripheral correlation increases from \(\rho_{2}=0.3\) to \(\rho_{2}=0.5\), \(P(\mathcal{A}_0)\) decays exponentially more rapidly with \(N\), dropping to extremely low levels even at relatively short source lengths.
It is evident that the SDC is rarely satisfied under general circumstances, and in such case, the Hadamard compression rate remains merely an unachievable limit. This underlines the relevance of obtaining an exact explicit form for the RDF when the SDC is not satisfied. 
In the next section, we will focus on this issue, derive the RDF and provide valuable insights under certain conditions.
\begin{figure}[t]
   \setlength{\abovecaptionskip}{0pt} 
   \centering
   \includegraphics[width=0.55\textwidth]{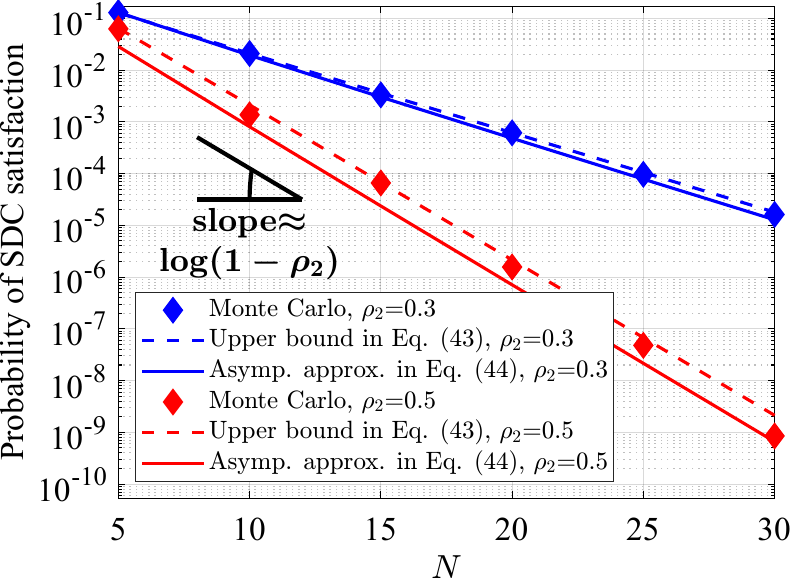}
   \caption{The probability of SDC satisfaction \(P(\mathcal{A}_0)\) versus the source length \(N\) for \(\rho_1=0.45\) and all distortion constraints \(e_i \sim \mathrm{Unif}[0,1]\).}
   \label{fig:SDC_prob_N_2}
\end{figure}

\subsubsection{A Source-Correlation Perspective}

For a given set of distortion constraints, we determine the component-wise correlations that a vector source should have to be compressed at the Hadamard rate in Eq.~\eqref{eq:hadamard bound value}.
In other words, the correlation region that satisfies the SDC is
\begin{align} \label{eq:correlation region}
   \mathcal{C}_{0} = \left\{ (\rho_{2}, \rho_1) \mid 0 \leq \rho_{1} \leq h(\rho_{2}), \, 0 \leq \rho_2 \leq \rho_2^m \right\},
\end{align}
where \(\rho_2^m\) is the zero of the following function
\begin{align} \label{eq:rho2m func}
   h(\rho_{2}) = (1 - e_1)^{\frac{1}{2}} \; \Bigg[\Bigg(\sum_{i=2}^{N} \frac{1}{1 - \rho_{2} - e_i}\Bigg)^{-1} + \rho_{2}\Bigg]^{\frac{1}{2}}.
\end{align} 
The upper bound on \(\rho_1\) in Eq.~\eqref{eq:correlation region} follows from Eq.~\eqref{eq:e1_condition}, and the maximum peripheral correlation \(\rho_2^m\) satisfies \(h(\rho_2^m) = 0\) in light of Eq.~\eqref{eq:e2_condition}.
Solving for \(\rho_2^m\) is intractable, as the equation \(h(\rho_2) = 0\)  can be written as an \((N-1)\)-th degree equation in \(\rho_2\), and such equations generally do not admit closed-form solutions for \(N-1 \geq 5\).
If there exists \(j > 2\) such that \(e_2 = e_3 = \cdots = e_j > e_{j+1} \geq \cdots \geq e_N\), then Eq.~\eqref{eq:e3_condition} directly gives \(\rho_{2}^m = 1 - e_2\).
Otherwise, we provide bounds and approximation for \(\rho_{2}^m\).
\begin{theorem} \label{thm:max rho0}
Given a distortion constraint vector \(\ev = [e_1, e_2, \cdots, e_N]^{\T}\) with \(e_2 > e_3 \geq \cdots \geq e_N\), and a 2-TC source covariance matrix \(\Rm \in \mathcal{K}_2\), the maximum peripheral correlation \(\rho_2^m\) satisfying the SDC is bounded by
\begin{align} \label{eq:max rho0 bounds}
1 - e_2 + \frac{c_l}{N} \leq \rho_{2}^m \leq 1 - e_2 + \frac{c_u}{N-1},
\end{align}
where \(c_l = 2(e_{2}-e_{3})\left(\sqrt{\frac{1-e_{3}}{1-e_{2}}}+1\right)^{-1}\)
and \(c_u = e_2 - \overline{e}_3\),
with \(\overline{e}_3 = \frac{1}{N-2}\sum_{i=3}^{N} e_i\).
For large \(N\), the quantity \(\rho_{2}^m\) satisfies
\begin{align} \label{eq:rho0m theta}
   \rho_{2}^m = 1 - e_2 + \Theta\!\left(\frac{1}{N}\right).
\end{align}
\end{theorem}
\begin{IEEEproof}[Proof Sketch]
The core of the proof lies in splitting \(h(\rho_{2})\) in Eq.~\eqref{eq:rho2m func} into two appropriate functions and bounding these functions to obtain concise and efficient bounds. Combining the bounds yields the asymptotic approximation.
See Appendix~\ref{proof:thm:max rho0} for the complete proof.
\end{IEEEproof}

In Theorem~\ref{thm:max rho0}, we derive the finite source-length bounds as well as the asymptotic approximation for \(\rho_{2}^m\), thereby filling the correlation region \(\mathcal{C}_{0}\) in Eq.~\eqref{eq:correlation region} to characterize the SDC from the source-correlation perspective.
Within the SDC region, there is a trade-off between the source correlation and the distortion constraint.
Specifically, for a vector source with the covariance matrix in the 2-TC class, if it can be compressed at the Hadamard rate in Eq.~\eqref{eq:hadamard bound value} under distortion constraints \(\{e_i\}_{i=1}^N\), then the gap between the maximum peripheral correlation \(\rho_{2}^{m}\) and \(1-e_2\) decreases inversely with the source length \(N\), where \(e_2\) is the mildest distortion constraint on the peripheral components.
If \(e_2\) is sufficiently mild and significantly exceeds \(1 - \rho_{2}\), the source can only be compressed at an optimal rate strictly higher than the Hadamard rate.
In Fig.~\ref{fig:rho0m}, we simulate \(\rho_{2}^m\) over \(10^{4}\) Monte Carlo trials, and the numerical solution is obtained by directly solving an \((N-1)\)-th degree equation, i.e., \(h(\rho_2^m) = 0\).
In Fig.~\ref{fig:rho0m_e2_fixed}, we fix \(e_2 = 0.1\) and randomly generate the remaining distortion constraints. As \(N\) increases, \(\rho_{2}^m\) converges to \(1 - e_2\), and the bounds remain tight.
Additionally, we fix \(N = 10\) to evaluate the bounds for short source lengths. Fig.~\ref{fig:rho0m_N_fixed} shows that our bounds closely match numerical results across the full range of \(e_{2}\).
\begin{figure*}[t]
\centering
\captionsetup[subfloat]{font=scriptsize}
\subfloat[\(e_2=0.1\) is fixed.]
{
   \hspace{-0.8em}
   \includegraphics[width=0.49\textwidth]{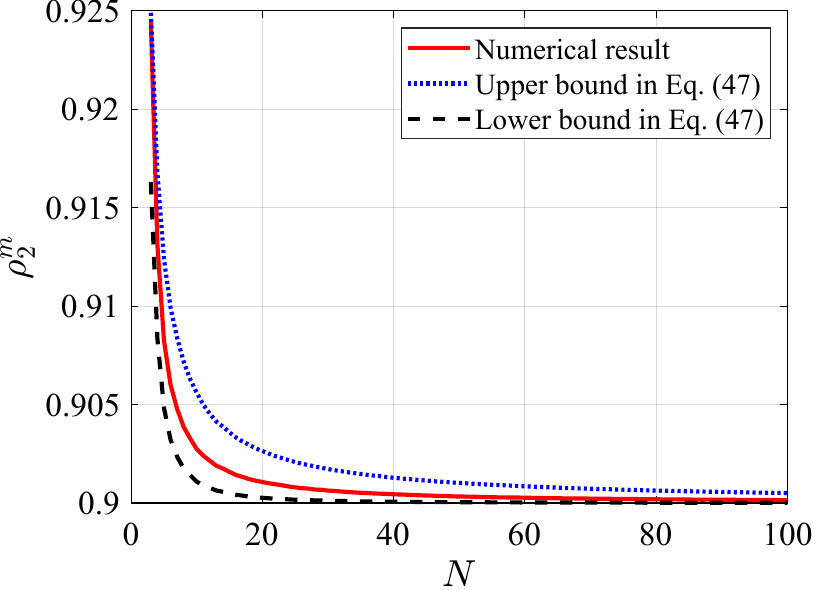}
   \label{fig:rho0m_e2_fixed}
}
\subfloat[\(N=10\) is fixed.]
{
   \hspace{0.1em}
   \includegraphics[width=0.47\textwidth]{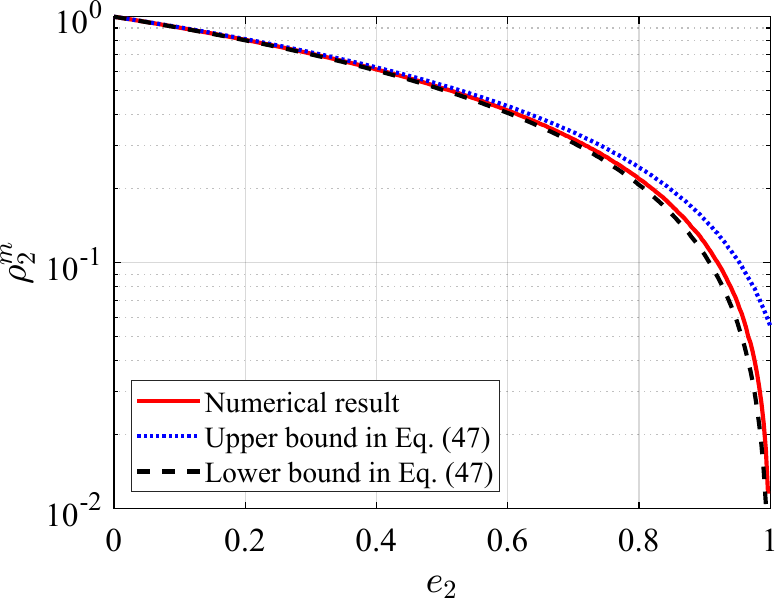}
   \hspace{0.0em} \label{fig:rho0m_N_fixed}
}
\caption{Evaluation of \(\rho_{2}^m\) in Theorem~\ref{thm:max rho0}.}
\label{fig:rho0m}
\end{figure*}

\section{Source Compression with 2-TC Covariance and 2-TD Constraints} \label{sec:Source Compression with 2-TC Covariance and 2-TD Constraints}

In this section, we investigate the Gaussian-quadratic compression with a 2-TC source covariance matrix under a specific distortion constraint.
The structure of the optimal distortion matrix is first characterized. We then present the complete and exact closed-form rate-distortion functions over a two-dimensional distortion plane that is partitioned into seven regions. We further establish the optimality of the distortion allocation in each region and provide simulation results. Subsequently, we offer further clarification and refine part of the results concerning the RDF into a more concise form, which is also proven to be highly accurate. Finally, we analyze the behavior of all regions under extreme correlation scenarios and in the asymptotic regime with respect to source length.

\subsection{Main Results} \label{subsec:Main Results}

Recall that for the covariance matrix \(\Rm\) in the 2-TC class, \(\Rm \succ \zerov\) is equivalent to \(\nu > 0\), with \(\nu\) given in Eq.~\eqref{eq:nu def}.  
Rewriting \(\nu > 0\) yields
\begin{align} \label{eq:2tc pd cond}
   \rho_{1}^{2} < \rho_{2}+\frac{1-\rho_{2}}{N-1},
\end{align}
where \(\rho_1\) denotes the correlation between the central and peripheral components, and \(\rho_{2}\) denotes the correlation among peripheral components.
Herein, we further specify the individual distortion constraints considered in this section.
For a vector Gaussian source \(\xv\), the fidelity criterion for the central component is \(\mathbb{E}[(X_1-\widehat{X}_1)^2] \le e_1\), while for each peripheral component \(X_i\) with any \(i \geq 2\), it is \(\mathbb{E}[(X_i-\widehat{X}_i)^2] \le e_2\). 
This setup, with two types of distortion constraints, is called the 2-TD constraints, and the distortion plane is \(\mathcal{U} = \{(e_1, e_2) \mid e_1, e_2 \in [0, 1]\}\). 
We now present the structure of the optimal distortion matrix in this setting.
\begin{theorem} \label{thm:dsm model}
   For the Gaussian-quadratic lossy compression with a 2-TC covariance matrix under 2-TD constraints, the optimal distortion matrix achieving the RDF has the \(2 \times 2\) block structure
\begin{align} \label{eq:dsm model}
   \Dsm =
\begin{pmatrix}
      \delta_1 & \alpha \mathbf{1}^{\T} \\[6pt]
      \alpha \mathbf{1} & \delta_2 \mathsf{I}_{N-1} + \beta \mathsf{J}_{N-1}
\end{pmatrix} \succ \zerov,
\end{align}
where \(\mathsf{J}_{N-1}=\mathbf{1}\mathbf{1}^{\T}-\mathsf{I}_{N-1}\), \(\mathbf{1}\) is the all-ones \(N-1\) dimensional vector and \(\mathsf{I}_{N-1}\) is the identity matrix with size \(N-1\). Here, the diagonal entries \(\delta_1 \in \mathbb{R}^+\) and \(\delta_2 \in \mathbb{R}^+\) represent the quadratic distortions for the central and peripheral components, respectively. 
The quantity \(\alpha \in \mathbb{R}\) denotes the correlation of distortion between the central and peripheral components and \(\beta \in \mathbb{R}\) denotes the correlation of distortion between different peripheral components.
\end{theorem}
\begin{IEEEproof}[Proof Sketch]
We first introduce the permutation group that fixes the central component while arbitrarily permuting the remaining \(N-1\) peripheral components. We then demonstrate that the group average of any feasible distortion matrix remains feasible and does not increase the RDF. By a contradiction argument, we deduce that the optimal solution must be invariant under this permutation group. This invariance directly implies that the optimal distortion matrix necessarily admits the stated block structure.
See Appendix~\ref{proof:thm:dsm model} for the complete proof.
\end{IEEEproof}

The following Theorem~\ref{thm:seven regions main} provides the complete and closed-form RDF for the Gaussian-quadratic lossy compression problem with a 2-TC covariance matrix under 2-TD constraints.
Specifically, key points are first defined to establish non-trivial boundaries within the unit square \(\mathcal{U}\). Based on these boundaries, \(\mathcal{U}\) is partitioned into seven regions. 
For each region, the corresponding RDF is provided.
The \textit{optimality} of the region-wise results refers to the optimality of the solutions to the underlying rate-distortion optimization problem, which is guaranteed by the KKT conditions (see Appendix~\ref{proof:thm:det(Rm-Dsm)}).
\begin{theorem}[Closed-form RDFs on seven regions]  \label{thm:seven regions main}
Define
\begin{align} \label{eq:eta def}
   \eta = \frac{(N-1)\rho_{1}}{1+(N-2)\rho_{2}}
\end{align}
and key points:
\begin{align}
   P_1 &= \left(\frac{[\rho_{2}-\rho_{1}^{2}]^{+}}{\rho_{2}}, \ 1-\rho_{2}-(N-1)[\rho_{1}^{2}-\rho_{2}]^{+}\right), \\
   P_2 &= \left(1-\rho_{1}\eta, \ 0\right), \\
   P_3 &= \left(0, \ 1 - \rho_{1}^{2}\right), \\
   P_4 &= \left(1-\rho_{2}\eta^{2}, \ 1 - \rho_{2}\right).
\end{align}
We use the notation \(\mathcal{B}_{m,n}^{(j)}\) to denote the boundary shared by the \(m\)-th and \(n\)-th regions, where the superscript \(j\) corresponds to the order of the Lagrange multipliers (see Eq.~\eqref{problem:cvx2}
):\footnote{When the context is clear, we omit the explicit notation for the ordered pairs \((e_1, e_2)\).}

\begingroup
\fontsize{12pt}{10pt}\selectfont
\begin{align*}
&\begin{cases}
   \mathcal{B}_{4,6}^{(1)} = \Biggl\{1-e_{2}-\frac{(N-1)\left(1-e_{1}\right)}{\eta^{2}}+(N-2)\rho_{2} = 0, 
   \
   e_1 \in [x_2, x_4], e_{2} \in [0, y_4]\Biggr\}, \\
   \mathcal{B}_{1,2}^{(1)} = \Biggl\{1-e_{2}-\frac{\left(1-e_{1}\right)}{\eta^{2}} = 0, 
   \ 
   e_1 \in [x_4, 1], e_{2} \in [y_4, 1]\Biggr\};
\end{cases}\\
& \quad \; \mathcal{B}_{1,3}^{(2)} = \{1 - e_{2} - (1 - e_1) \rho_{1}^{2} = 0, 
   \ 
   e_1 \in [0, 1], e_{2} \in [y_3, 1]\}; \\
&\begin{cases}
   \mathcal{B}_{0,5}^{(3)} = \{1 - e_{2} - \rho_{2} = 0, 
   \ 
   e_1 \in [0, x_1]\}, \\
   \mathcal{B}_{2,6}^{(3)} = \{1 - e_{2} - \rho_{2} = 0, 
   \ 
   e_1 \in [x_4, 1]\}, \\
   \mathcal{B}_{1,4}^{(3)} = \Biggl\{e_1 + \frac{N - 1}{1-\rho_{2}} \omega + \rho_{1} \sqrt{\frac{1 - e_1}{1 - e_{2}}} - 1 = 0, 
   \ 
   e_1 \in [x_1, x_4], e_{2} \in [e_{2}^\star, y_4]\Biggr\};
\end{cases}\\
&\begin{cases}
   \mathcal{B}_{0,4}^{(4)} = \Biggl\{1-e_{2}-\frac{\rho_{1}^{2}}{1-e_{1}}+(N-2)\left(\rho_{2}-\frac{\rho_{1}^{2}}{1-e_{1}}\right) = 0, 
   \ 
   e_1 \in [x_1, x_2], e_{2} \in [0, y_1]\Biggr\}, \\
   \mathcal{B}_{1,5}^{(4)} = \Biggl\{1 - e_{2} - \frac{\rho_{1}^{2}}{1 - e_1} = 0, 
   \ 
   e_1 \in [0, x_1], e_{2} \in [y_1, y_3]\Biggr\}.
\end{cases}
\end{align*}
\endgroup
Due to point \(P_1\), the boundaries \(\mathcal{B}_{0,5}^{(3)}\) and \(\mathcal{B}_{1,5}^{(4)}\) exist if and only if
\begin{align} \label{eq:sqrt rho0 cond}
   \rho_1^{2} < \rho_{2}
\end{align}
holds.
For the boundary \(\mathcal{B}_{1,4}^{(3)}\), the parameter \(\omega\) is given by
\begin{align} \label{eq:omega B3}
   \omega = e_{1}(e_{2} + \rho_{2} - 1) - \left(\rho_{1} - \sqrt{(1 - e_{1})(1 - e_2)}\right)^{2},
\end{align}
and \(e_{2}^\star\) is defined as
\begin{align} \label{eq:e0star def}
   e_{2}^\star \triangleq \min\left\{e_{2} : (e_{1}, e_2) \in \mathcal{B}_{1,4}^{(3)} \cap \mathcal{U} \right\}.
\end{align}
We define the top and bottom regions based on the established boundaries above\footnote{Let \(\mathcal{B} = \{(x,y) \in \mathbb{R}^2 \mid f(x,y)=0,\, x\in[x_1,x_2],\, y\in[y_1,y_2]\}\) be the boundary.
The top and bottom regions associated with \(\mathcal{B}\) are defined as \(\mathcal{R}^{\dagger} = \{(x, y + \delta y) \mid \exists \; \delta y \in \mathbb{R}, (x, y) \in \mathcal{B}\}\), where \(\dagger = \mathrm{T}\) if \(\delta y \geq 0\) and \(\dagger = \mathrm{B}\) if \(\delta y \leq 0\).},
and then apply set operations on these regions to partition the whole distortion plane \(\mathcal{U} = \{(e_{1}, e_2) \mid e_{1}, e_2 \in [0, 1]\}\) into seven regions.
Below, the RDF and the optimality condition for each region are presented sequentially.
\begin{itemize}

   \item \(\mathcal{R}_{0}\) (SDC region):
The rate-distortion function with the parameter \(\nu\) in Eq.~\eqref{eq:nu def} is given by
\begin{align} \label{eq:sdc R0 rdf}
\Rx([N], \ev) = \frac{1}{2} \log \frac{\left(1-\rho_{2}\right)^{N-2}\nu}{e_{1}e_{2}^{N-1}}.
\end{align}
Optimality is achieved if and only if \((e_{1}, e_2)\) lies in
\begin{align} \label{eq:sdc inequality v1}
   \mathcal{R}_{0} = \left(\mathcal{R}_{0,5}^{(3), \mathrm{B}} \cup \mathcal{R}_{0,4}^{(4), \mathrm{B}}\right) \cap \mathcal{U}.
\end{align}

   \item \(\mathcal{R}_{1}\):
The rate-distortion function with the parameter \(\omega\) in Eq.~\eqref{eq:omega B3} is given by
\begin{align} \label{eq:region1 rdf}
   \Rx([N], \ev) = \frac{1}{2} \log \frac{\nu}{(N-1)\omega+(1-\rho_{2})e_{1}}.
\end{align}
Optimality is achieved if and only if \((e_{1}, e_2)\) lies in
\begin{align} \label{eq:region1 cond}
\mathcal{R}_{1} = \left(\mathcal{R}_{1,2}^{(1), \mathrm{T}} \cup \mathcal{R}_{1,4}^{(3), \mathrm{T}} \cup \mathcal{R}_{1,5}^{(4), \mathrm{T}}\right) \cap \mathcal{R}_{1,3}^{(2), \mathrm{B}} .
\end{align}

   \item \(\mathcal{R}_{2}\):
The rate-distortion function is given by
\begin{align} \label{eq:region2 rdf}
\Rx([N], \ev) = \frac{1}{2} \log \frac{1 + (N-2)\rho_{2}}{e_{2} + (N-2)\left(e_{2} + \rho_{2} - 1\right)}.
\end{align}
   Optimality is achieved if and only if \((e_{1}, e_2)\) lies in
\begin{align} \label{eq:region2 cond}
\mathcal{R}_{2} = \mathcal{R}_{1,2}^{(1), \mathrm{B}} \cap \mathcal{R}_{2,6}^{(3), \mathrm{T}}.
\end{align}

   \item \(\mathcal{R}_{3}\):
The rate-distortion function is given by
\begin{align} \label{eq:region3 rdf}
\Rx([N], \ev) = \frac{1}{2} \log \frac{1}{e_{1}}.
\end{align}
   Optimality is achieved if and only if \((e_{1}, e_2)\) lies in
\begin{align} \label{eq:region3 cond}
\mathcal{R}_{3} = \mathcal{R}_{1,3}^{(2), \mathrm{T}} \cap \mathcal{U}.
\end{align}

   \item \(\mathcal{R}_{4}\):
   The rate-distortion function is given by
\begin{align} \label{eq:region4 rdf}
\mathbb{R}_{\xv}([N], \ev) ={}& \frac{1}{2} \log \frac{ \nu }{ e_1 [e_{2}  + (N-2)\beta_\star] - (N-1)\alpha_\star^2 } 
+ \frac{N-2}{2} \log \frac{1 - \rho_{2}}{e_{2} - \beta_\star}.
\end{align}
Optimality is achieved if and only if \((e_{1}, e_2)\) lies in
\begin{align} \label{eq:region4 cond}
\mathcal{R}_{4} = \left(\mathcal{R}_{4,6}^{(1), \mathrm{T}} \cup \mathcal{R}_{0,4}^{(4), \mathrm{T}}\right) \cap \mathcal{R}_{1,4}^{(3), \mathrm{B}}.
\end{align}
The quantities \(\alpha_\star\) and \(\beta_{\star}\) are given by
\begin{align} 
      \alpha_{\star} &= \frac{t}{3\left(N-1\right)}-\frac{t^{2}-3p}{3\left(N-1\right)^{2}w}-\frac{w}{3},  \label{eq:alpha star analytical} \\
   \beta_\star &= \rho_{2}+\frac{1-e_{2}}{N-2}-\frac{\left(N-1\right)\left(\alpha_\star-\rho_{1}\right)^{2}}{\left(N-2\right)\left(1-e_{1}\right)}, \label{eq:beta star analytical}
\end{align}
respectively, where \(\alpha_\star\) is the unique root of the cubic equation
\begin{align} \label{eq:region4 cubic func}
   f(\alpha)=\left(N - 1\right)^{2}\alpha^{3} - \left(N - 1\right) t \alpha^{2} + p\alpha + \left(N-1\right)q
\end{align}
with the corresponding parameters
\begin{align} \label{eq:region4 cubic param}
   p&=\left(N-1\right)\left[\left(3N-4\right)e_{1}+1\right]\rho_{1}^{2}+\left(N-1\right)e_{2}\left(1-e_{1}\right) \notag \\
   & \quad -\left[\left(N-2\right)\rho_{2}+1\right]\left[\left(N-2\right)e_{1}+1\right]\left(1-e_{1}\right), \notag \\
   q&=\left\{\left[\left(N-2\right)\rho_{2}+1-e_{2}\right]\left(1-e_{1}\right)-\left(N-1\right)\rho_{1}^{2}\right\}e_{1}\rho_{1}, \notag \\
   t&=\left[\left(2N-3\right)e_{1}+N\right]\rho_{1}, \notag \\
   s&=\frac{27q}{2\left(N-1\right)}+\frac{9tp}{2\left(N-1\right)^{3}}-\frac{t^{3}}{\left(N-1\right)^{3}}, \notag \\
   w&=\Bigg[\Bigg(s^{2}-\frac{\left(t^{2}-3p\right)^{3}}{\left(N-1\right)^{6}}\Bigg)^{\frac{1}{2}}+s\Bigg]^{\frac{1}{3}}.
\end{align}

   \item \(\mathcal{R}_{5}\):
The rate-distortion function is given by
\begin{align} \label{eq:region5 rdf}
\Rx([N], \ev) = \frac{1}{2} \log \frac{\nu}{e_{1}\left[1-\rho_{2}+\left(N-1\right)\left(e_{2}+\rho_{2}-1\right)\right]}.
\end{align}
Optimality is achieved if and only if \((e_{1}, e_2)\) lies in
\begin{align} \label{eq:region5 cond}
\mathcal{R}_{5} = \mathcal{R}_{0,5}^{(3), \mathrm{T}} \cap \mathcal{R}_{1,5}^{(4), \mathrm{B}}.
\end{align}

   \item \(\mathcal{R}_{6}\):
The rate-distortion function is given by
\begin{align} \label{eq:region6 rdf}
\Rx([N], \ev) = \frac{1}{2} \log \frac{(1-\rho_{2})^{N-2} [1+(N-2)\rho_{2}]}{e_{2}^{N-1}}.
\end{align}
Optimality is achieved if and only if \((e_{1}, e_2)\) lies in
\begin{align} \label{eq:region6 cond}
\mathcal{R}_{6} = \left(\mathcal{R}_{4,6}^{(1), \mathrm{B}} \cup \mathcal{R}_{2,6}^{(3), \mathrm{B}}\right) \cap \mathcal{U}.
\end{align}
\end{itemize}
\end{theorem}

\begin{IEEEproof}[Proof Sketch]
   Under the 2-TC covariance matrix and 2-TD constraints, we reformulate the Max-Det problem in Lemma~\ref{lemma:maxdet}. By applying Lemma~\ref{lemma:determinant}, we have
\begin{align} \label{eq:dsm det}
   \det(\Dm) = (\delta_2 - \beta)^{N-2} \Xi
\end{align}
for \(\Dm\) in Eq.~\eqref{eq:dsm model}, where \(\Xi = \delta_1\left[\delta_2 + (N-2)\beta\right] - (N-1)\alpha^2\).
Similarly, we have
\begin{align} \label{eq:Rdsm det}
   \det(\Rm-\Dm) = \left[(1 - \delta_2) - (\rho_{2} - \beta)\right]^{N-2} \Upsilon,
\end{align}
where \(\Upsilon = (1-\delta_{1})\left[(1-\delta_2)+(N-2)(\rho_{2}-\beta)\right] -(N-1)(\rho_{1}-\alpha)^{2}\).
The optimal distortion matrix \(\Dsm\) in Eq.~\eqref{eq:dsm model} is determined by solving the following problem, i.e.,
\begin{problem}  \label{problem:cvx2}
   \begin{align}
      \max_{\delta_1, \alpha, \delta_2, \beta} \quad&  \log \left[(\delta_2-\beta)^{N-2} \Xi\right], \label{eq:opt obj} \\
\text{s.t.} \;\;\; \quad & \delta_1 \leq e_1, \label{eq:opt c1} \\
   & \delta_2 \leq e_{2}, \label{eq:opt c2} \\
   & \delta_2 - \beta \leq 1-\rho_{2}, \label{eq:opt c3} \\
   & 0 \leq \Upsilon. \label{eq:opt c4}
   \end{align}
\end{problem}
For the constraints in Eqs.~\eqref{eq:opt c1}-\eqref{eq:opt c4}, we introduce the multipliers \(\lambda_1,\lambda_2,\lambda_3,\lambda_4\) respectively to form the Lagrangian.
We first derive explicit expressions for all multipliers through the stationarity condition, then analyze and exclude mutually exclusive sign combinations. The remaining sign combinations of the multipliers define distinct rate-distortion regions.
In each region, we first identify the sign combinations of the multipliers, including those that are zero and positive. We then set up polynomial equations for the zero multipliers and equality constraints for the positive ones, based on complementary slackness. Solving these equations yields the distortion parameters \(\delta_1\), \(\alpha\), \(\delta_2\), and \(\beta\) in Eq.~\eqref{eq:dsm model}.
Furthermore, using Lemma~\ref{lemma:determinant}, we derive the exact closed-form RDF, explicitly incorporating distortion constraints and source correlations. Finally, we verify primal feasibility by substituting the distortion parameters into all constraints, and dual feasibility by checking all multipliers.
Together, these establish the optimality conditions for each case, i.e., the boundaries of the corresponding regions.

Among all regions, \(\mathcal{R}_4\) is notable because the cubic equation \(f(\alpha) = 0\) in Eq.~\eqref{eq:region4 cubic func} may have multiple solutions, causing ambiguity in determining \(\alpha_\star\). To resolve this, we analyze all possible bounds on \(\alpha_\star\) based on all primal and dual feasibility conditions. While \(\alpha_\star\) can be obtained from Eq.~\eqref{eq:alpha star analytical}, its complicated form and tedious substitutions make it impractical for deriving boundary conditions directly. Instead, we derive the conditions that all bounds for \(\alpha_{\star}\) should satisfy and thereby establish the complete boundaries of \(\mathcal{R}_4\).
See Appendix~\ref{proof:thm:seven regions main} for the complete proof.
\stepcounter{equation}
\end{IEEEproof}

In Theorem~\ref{thm:seven regions main}, we first present all the boundaries arranged in the order of Lagrange multipliers (superscripts). The range of each boundary is determined by the coordinates of key points. These boundaries partition the two-dimensional distortion plane \(\mathcal{U}\) into seven regions. For each region, we sequentially present the complete and exact closed-form RDF. 
In Fig.~\ref{fig:all_regions_2D}, we plot all seven regions on the plane \(\mathcal{U}\) for \(N=8\), \(\rho_1=0.45\), and \(\rho_{2}=0.3\), with key points and boundaries clearly marked.
For all regions, Table~\ref{tab:region_rank_distortion} summarizes the signs of the Lagrange multipliers, \(\rank(\ev - \dsv)\), \(n_{+}(\Kxs)\), and the optimal distortion allocations.

\begin{figure}[t]
   \setlength{\abovecaptionskip}{0pt} 
   \centering
   \includegraphics[width=0.45\textwidth]{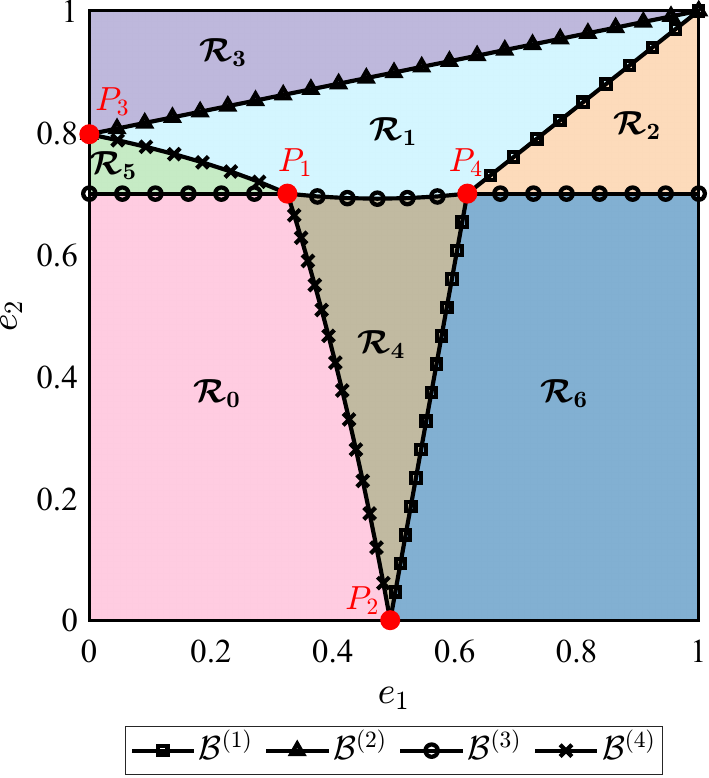}
   \caption{Seven regions on the distortion plane for \(N=8\), \(\rho_1=0.45\), and \(\rho_{2}=0.3\). The superscript on the boundaries corresponds to different Lagrange multipliers, while the subscript is omitted since the boundary-region relationships are clear.}
   \label{fig:all_regions_2D}
\end{figure}

\begin{table*}[!t]
   \footnotesize
   \caption{Properties of the Seven Regions and Optimal Distortion Allocations}
   \label{tab:region_rank_distortion}
   \renewcommand{\arraystretch}{1.3}
   \begin{threeparttable}
   \resizebox{\textwidth}{!}{
   \begin{tabular}{c|c|c|c|c|c|c|c|c|c|c}
      \hline
      Region & \(\lambda_1\)\tnote{1} & \(\lambda_2\) & \(\lambda_3\) & \(\lambda_4\) & \(\rank(\ev - \dsv)\)\tnote{4} & \(n_{+}(\Kxs)\)\tnote{5} & \(\delta_1\) & \(\alpha\) & \(\delta_2\) & \(\beta\) \\
      \hline
      \hline
      \(\mathcal{R}_{0}\)  & \(+\) & \(+\) & \(0\) & \(0\) & \(0\) & \(N\) & \(e_1\) & \(0\) & \(e_{2}\) & \(0\) \\
      \hline
      \(\mathcal{R}_1\)  & \(+\) & \(+\) & \(+\) & \(+\) & \(0\) & \(1\) & \(e_1\) & \(\rho_1 - \sqrt{(1 - e_{2})(1 - e_1)}\) & \(e_{2}\) & \(e_{2} + \rho_{2} - 1\) \\
      \hline
      \(\mathcal{R}_2\)\tnote{2} & \(0\) & \(+\) & \(+\) & \(+\) & \(1\) & \(1\) & \(1-\left(1-e_{2}\right)\eta^{2}\) & \(\rho_{1}-(1-e_{2})\eta\) & \(e_{2}\) & \(e_{2} + \rho_{2} - 1\) \\
      \hline
      \(\mathcal{R}_3\) & \(+\) & \(0\) & \(+\) & \(+\) & \(N-1\) & \(1\) & \(e_1\) & \(e_1\rho_1\) & \(1 - (1-e_1)\rho_1^{2}\) & \(\rho_{2} - (1-e_1)\rho_1^{2}\) \\
      \hline
      \(\mathcal{R}_4\)\tnote{3} & \(+\) & \(+\) & \(0\) & \(+\) & \(0\) & \(N-1\) & \(e_1\) & \(\alpha_{\star} > 0\) & \(e_{2}\) & \(\beta_{\star} < 0\) \\
      \hline
      \(\mathcal{R}_5\) & \(+\) & \(+\) & \(+\) & \(0\) & \(0\) & \(2\) & \(e_1\) & \(0\) & \(e_{2}\) & \(e_{2} + \rho_{2} - 1\) \\
      \hline
      \(\mathcal{R}_6\)\tnote{2} & \(0\) & \(+\) & \(0\) & \(+\) & \(1\) & \(N-1\) & \(1-\rho_{1}\eta+\frac{e_{2}\eta^{2}}{(N-1)}\) & \(\frac{e_{2}\eta}{\left(N-1\right)}\) & \(e_{2}\) & \(0\) \\
      \hline
   \end{tabular}
   }
      \begin{tablenotes}
      \footnotesize
      \item[1] {The ``\(+\)'' and ``\(0\)'' indicate that the Lagrange multiplier is positive and zero, respectively, in that distortion region.}
      \item[2] {The quantity \(\eta\) is given in Eq.~\eqref{eq:eta def}.}
      \item[3] {The quantities \(\alpha_{\star}\) and \(\beta_{\star}\) are given in Eqs.~\eqref{eq:alpha star analytical} and~\eqref{eq:beta star analytical}, respectively. Their bounds are detailed in Proposition~\ref{prop:Dsm bound}.}
      \item[4] {\(\rank(\ev - \dsv)\) denotes the number of non-zero entries in the vector \(\ev - \dsv\).}
      \item[5] {\(n_{+}(\Kxs)\) denotes the number of positive eigenvalues of the matrix \(\Kxs\), including multiplicities.}
      \end{tablenotes}
   \end{threeparttable}
\end{table*}

\begin{remark}
The four multipliers in Table~\ref{tab:region_rank_distortion} and the boundaries in Fig.~\ref{fig:all_regions_2D} are closely related.
For example, consider the transition between \(\mathcal{R}_0\) and \(\mathcal{R}_1\): a path that crosses the fewest boundaries from an arbitrary point in \(\mathcal{R}_0\) to an arbitrary point in \(\mathcal{R}_1\) involves either crossing the pair \(\mathcal{B}_{0,5}^{(3)}\) and \(\mathcal{B}_{1,5}^{(4)}\) or the pair \(\mathcal{B}_{0,4}^{(4)}\) and \(\mathcal{B}_{1,4}^{(3)}\). In both cases, the path crosses boundaries associated with \(\lambda_3\) and \(\lambda_4\), which exactly correspond to the sign differences between \(\mathcal{R}_0\) and \(\mathcal{R}_1\). This correspondence holds for all region transitions.
\end{remark}

Unlike in the case of sum distortion criterion, it is highly challenging to derive a unified expression for the optimal distortion allocation or RDF under individual distortion constraints. This is because both the correlations among components and the imposed distortion constraints significantly influence the optimal distortion allocation.
Nevertheless, our analytical framework advances the resolution of the open problem of determining the optimal distortion allocation and RDF, and the comprehensive results provide some fundamental new theoretical insights.

Compared to the distortion allocations under sum distortion criterion, more distortions need to be allocated under individual distortion criteria. These include not only the diagonal entries of the distortion matrix but also the off-diagonal entries, the latter of which are not considered under sum distortion criterion. These off-diagonal entries represent the correlations among the different distortions and can be viewed as the co-distortions between different components.
Under sum distortion criterion, the RDF is achieved by equally distributing distortions below the variance level for the parallel sources. However, under individual distortion criteria, equal allocation is inherently infeasible. The key to obtaining the RDF lies in carefully analyzing the correlations among the optimal distortions. This aspect is not only excluded under sum distortion criterion but also overlooked by the Hadamard bound, which assumes uncorrelated distortions across components.
Under individual distortion criteria, while certain pairs of regions exhibit closely related optimal distortion allocations, many other pairs display more intricate and less straightforward relationships.
For instance, the optimal distortion allocations in regions \(\mathcal{R}_1\) and \(\mathcal{R}_2\) differ only in the distortion \(\delta_1\) assigned to the central component \(X_{1}\). In \(\mathcal{R}_1\), the constraint is active and \(\delta_1 = e_1\). In \(\mathcal{R}_2\), the constraint on \(\delta_1\) is too mild, leaving a distortion margin such that \(\delta_1 = 1 - (1 - e_{2}) \eta^2\) with \(\eta\) in Eq.~\eqref{eq:eta def}. All other distortion allocations remain unchanged except for replacing \(\delta_1\) accordingly.
However, this simple principle does not hold in general. This is because a difference in the sign of a single multiplier can trigger constraints on multiple components. Specifically, although only the sign of the multiplier \(\lambda_3\) changes between \(\mathcal{R}_1\) and \(\mathcal{R}_4\), both \(\alpha\) and \(\beta\) in the distortion matrix are altered simultaneously.
\begin{remark}
An interesting observation is that, although all source component-wise correlations are non-negative, the correlations between the reconstruction distortions (i.e., the differences between the source and its reconstruction) are not necessarily so.
For example, in \(\mathcal{R}_4\), the optimal off-diagonal entry \(\beta_{\star} = \mathbb{E}[(X_i-\hat{X}_i^{\star})(X_j-\hat{X}_j^{\star})] < 0\) for \(2 \leq i, j \leq N\) with \(i \neq j\), indicates that the reconstruction distortions of the peripheral components are negatively correlated.
Statistically, this implies that the deviations of the reconstructions from the true source values tend to have opposite signs across different components.
For instance, an overestimate in one component is likely to be accompanied by an underestimate in another.
\end{remark}

In Fig.~\ref{fig:all_regions_rate}, we compare the closed-form RDF in Theorem~\ref{thm:seven regions main} with numerical results obtained using the interior-point method, as well as with the Hadamard compression rate in Eq.~\eqref{eq:hadamard bound value}.
The number of variables in the optimization problem grows as \(O(N^2)\), while the number of non-linear constraints grows as \(O(N)\). As a result, for longer source lengths, the numerical method becomes computationally impractical.
Regarding the Hadamard lower rate, there is a significant gap compared to our exact analytical RDF over a wide range, and the Hadamard rate even yields negative values.
In Fig.~\ref{fig:all_regions_3D_R_2D_e1}, the rate appears relatively flat because the distortion constraint pair \((e_{1}, e_2)\) lies in \(\mathcal{R}_2\) or \(\mathcal{R}_6\), where our analytical RDF in Eq.~\eqref{eq:region2 rdf} or Eq.~\eqref{eq:region6 rdf} shows that the rate is independent of \(e_1\). 
In Fig.~\ref{fig:all_regions_3D_R_2D_e0}, for larger \(e_{2}\), the analytical curves for \(\rho_{2}=0.4\) and \(\rho_{2}=0.8\) coincide because the region is \(\mathcal{R}_3\). The RDF in Eq.~\eqref{eq:region3 rdf} depends solely on \(e_1\). Only the central component is reconstructed, so the peripheral distortion constraints do not affect the optimal compression rate.

\begin{figure*}[t]
   \centering
   \captionsetup[subfloat]{font=scriptsize}
   \subfloat[\(\rho_1=0.6\) and \(\rho_{2}=0.3\) are fixed.]
   {
      \hspace{-0.0em}
      \includegraphics[width=0.46\textwidth]{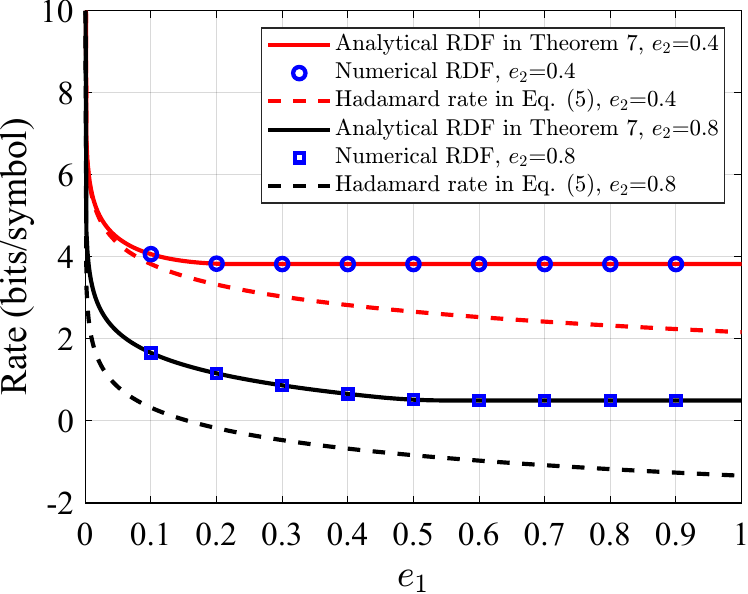}
      \hspace{0.3em} \label{fig:all_regions_3D_R_2D_e1}
   }
   \subfloat[\(\rho_1=0.6\) and \(e_1=0.4\) are fixed.]
   {
      \hspace{0.7em}
      \includegraphics[width=0.46\textwidth]{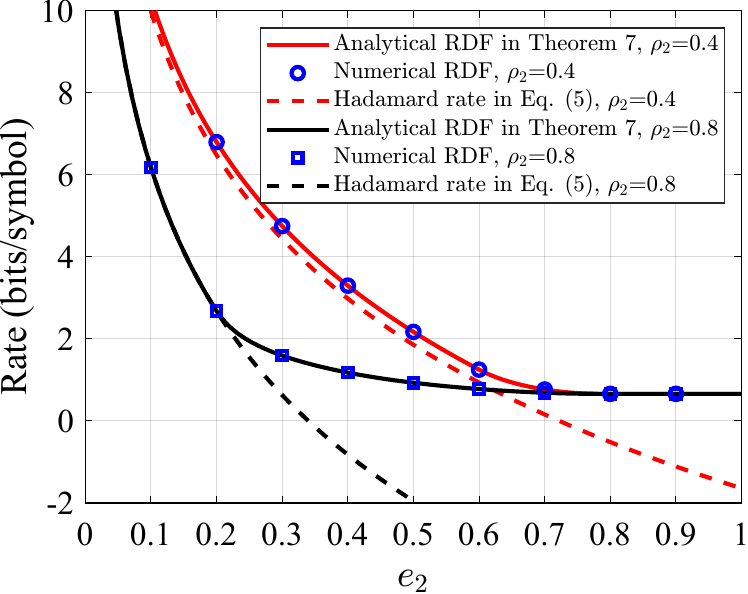}
      \hspace{1.3em} \label{fig:all_regions_3D_R_2D_e0}
   }
   \caption{Comparison of the analytical result, numerical result, and the Hadamard bound result for \(N=8\).}
   \label{fig:all_regions_rate}
\end{figure*}

Subsequently, we provide further clarification on the boundary \(\mathcal{B}_{1,4}^{(3)}\) in Theorem~\ref{thm:seven regions main}. Since \(\mathcal{B}_{1,4}^{(3)}\) is a cubic curve on the distortion plane \(\mathcal{U}\), its geometric characterization is not transparent, rendering the optimality conditions for distortion allocation analytically intractable.
To analyze it explicitly, we defined the quantity \(e_{2}^\star\) in Eq.~\eqref{eq:e0star def}, and the results are as follows.

\begin{proposition} \label{prop:region one B3 e0star}
   Recall the condition for \(\Rm \succ \zerov\) in Eq.~\eqref{eq:2tc pd cond}, and \(e_{2}^\star\) in Eq.~\eqref{eq:e0star def} is given by
\begin{numcases}{e_{2}^\star =}
   \scalebox{1.1}{$1 - \rho_{2} - \frac{\mu - \rho_{2} \sqrt{(N-1)\nu}}{2\sqrt{(N-1)\nu}}$}, \! & if \(\rho_1^{2} \leq \rho_{2} + \frac{1 - \rho_{2}}{2(N - 1)}\), \label{eq:e0star analytical}
   \\
   1 - \rho_1^2, \! & \(\text{otherwise}\),  \label{eq:e0star analytical otherwise}
\end{numcases}
with \(\mu = \sqrt{\left[(N-2)\rho_{2} + 1\right]\left[\rho_{2} \nu - \left(\rho_{2} - \rho_1^2\right)(1-\rho_{2})\right]}\) and \(\nu\) in Eq.~\eqref{eq:nu def}. In the asymptotic regime, we have
\begin{numcases}{1-e_{2}^\star-\rho_{2} =}
   \scalebox{1.1}{$\Theta\!\left(\frac{1}{N^2}\right)$}, \! & if \(\rho_1^{2} \leq \rho_{2} + \frac{1 - \rho_{2}}{2(N - 1)}\), \label{eq:e0star asymptotic}
   \\
   \scalebox{1.1}{$\Theta\!\left(\frac{1}{N}\right)$}, \! & \(\text{otherwise}\).  \label{eq:e0star asymptotic otherwise}
\end{numcases}
\end{proposition}
\begin{IEEEproof}[Proof Sketch]
   We fix \( e_{2} \) in the bivariate set \( \mathcal{B}_{1,4}^{(3)} \) and solve the quadratic discriminant with respect to \( e_1 \) to derive the necessary condition for \( e_{2}^\star \). A detailed analysis then leads to the precise form of \( e_{2}^\star \) under different conditions, and subsequent algebraic manipulations yield the asymptotics. See Appendix~\ref{proof:prop:region one B3 e0star} for the complete proof.
\end{IEEEproof}
Refer to the boundaries in Theorem~\ref{thm:seven regions main} and their graphical representation in Fig.~\ref{fig:all_regions_2D}. Since \( \mathcal{B}_{0,5}^{(3)} \) and \( \mathcal{B}_{2,6}^{(3)} \) are part of the straight line \( e_{2} = 1 - \rho_{2} \), Proposition~\ref{prop:region one B3 e0star} essentially determines the gap between \( \mathcal{B}_{1,4}^{(3)} \) and the line.
Although \(\mathcal{B}_{1,4}^{(3)}\) in Theorem~\ref{thm:seven regions main} is rather cumbersome, we find that \(\mathcal{B}_{1,4}^{(3)}\) eventually converges to a straight line.  
A comparison between Eq.~\eqref{eq:e0star asymptotic} and Eq.~\eqref{eq:e0star asymptotic otherwise} reveals that a smaller \( \rho_1 \) leads to a faster convergence rate.

Next, we further refine our results in \(\mathcal{R}_4\). 
Although Theorem~\ref{thm:seven regions main} already provides the exact RDF in Eq.~\eqref{eq:region4 rdf} and the optimal distortion allocations \(\alpha_\star\) and \(\beta_\star\) in Eqs.~\eqref{eq:alpha star analytical} and \eqref{eq:beta star analytical}, their dependence on many parameters in Eq.~\eqref{eq:region4 cubic param} complicates both computation and interpretation. Therefore, in Proposition~\ref{prop:Dsm bound}, we provide relatively concise bounds for \(\alpha_\star\) and \(\beta_\star\), and also establish an achievable upper bound for the RDF. Subsequently, in Theorem~\ref{thm:iso tighter}, we demonstrate that the bounds are also highly accurate.
\begin{proposition} \label{prop:Dsm bound}
The bounds for \(\alpha_{\star}\) in Eq.~\eqref{eq:alpha star analytical} and \(\beta_{\star}\) in Eq.~\eqref{eq:beta star analytical} are given by
\begin{align}
   0 &< \alpha_{\star} < \alpha^{u}_{\star},  \label{eq:beta star bounds}\\
   \beta_{\star} &< 0,  \label{eq:alpha star bounds}
\end{align}
where
\begin{align} \label{eq:beta star upper}
   \alpha^{u}_{\star} = \rho_{1} - \sqrt{\frac{\left(1-e_{1}\right)\left[(N-2)\rho_{2}+1-e_{2}\right]}{N-1}}.
\end{align}
The upper bound of the optimal rate-distortion function is given by
\begingroup
\fontsize{12pt}{12pt}\selectfont
\begin{align} \label{eq:region4 rdf upper}
   &R^{u}_{\xv}([N], \ev) = \frac{N-2}{2} \log \frac{1 - \rho_{2}}{e_{2}} + \frac{1}{2} \log \frac{ \nu }{ e_1 e_{2} - (N-1)(\alpha^{u}_{\star})^2 },
   \end{align}
\endgroup
when
\begin{align} \label{eq:D tilde sdc}
   \alpha^{u}_{\star}<\sqrt{\frac{e_{1}e_{2}}{N-1}}.
\end{align}
\end{proposition}
\begin{IEEEproof}
   See Appendix~\ref{proof:prop:Dsm bound}.
\end{IEEEproof}

Compared with \(\alpha_{\star}\) in Eq.~\eqref{eq:alpha star analytical}, the upper bound \(\alpha^{u}_{\star}\) in Eq.~\eqref{eq:beta star upper} is remarkably concise, leading to a more compact upper bound for the RDF in Eq.~\eqref{eq:region4 rdf upper}. This compactness allows us to disregard minor contributions to the RDF arising from the parameter dependence in Eq.~\eqref{eq:region4 cubic param}, thereby focusing on the dominant terms.
More importantly, this upper bound is proven to be a tighter approximation than the Hadamard bound, demonstrating its effectiveness in capturing the essential behavior of the RDF.
To establish this, a sufficient condition for Eq.~\eqref{eq:D tilde sdc} to hold in \(\mathcal{R}_4\) is to assume an isotropic covariance matrix with \(\rho = \rho_{1} = \rho_2\).  
When substituting \(\alpha^{u}_{\star}\) into Eq.~\eqref{eq:D tilde sdc}, we obtain an elliptical region in the \((e_{1}, e_2)\) plane. Since this region is convex, it suffices to consider the vertices of \(\mathcal{R}_4\) in Eq.~\eqref{eq:region4 cond}. In particular, for \(\sqrt{\rho_{2}} \leq \rho_{1} \leq \sqrt{\frac{(N-2)\rho_{2} + 1}{N-1}}\), \(P_3\) is one of the vertices of \(\mathcal{R}_4\)
and may fail to satisfy Eq.~\eqref{eq:D tilde sdc} due to variations in correlations.
Under the isotropic correlation, Eq.~\eqref{eq:D tilde sdc} is always satisfied in \(\mathcal{R}_4\), and Theorem~\ref{thm:iso tighter} demonstrates that \(R^{u}_{\xv}([N], \ev)\) provides a tighter bound on the RDF than the Hadamard bound.

\begin{theorem} \label{thm:iso tighter}
   In \(\mathcal{R}_{4}\), when \(\rho = \rho_{1} = \rho_2\) is satisfied, the following holds for arbitrary \(N \geq 3\),
   \begin{align} \label{eq:iso tighter}
      \Rx([N], \ev) \geq \frac{1}{2}\left(R^{u}_{\xv}([N], \ev) + R^{\ell}_{\xv}([N], \ev)\right),
   \end{align}
   where \(\Rx([N], \ev)\) is given in Eq.~\eqref{eq:region4 rdf}, \(R^{u}_{\xv}([N], \ev)\) in Eq.~\eqref{eq:region4 rdf upper}, and \(R^{\ell}_{\xv}([N], \ev)\) in Eq.~\eqref{eq:hadamard bound value}.
\end{theorem}
\begin{IEEEproof}
   See Appendix~\ref{proof:thm:iso tighter}.
\end{IEEEproof}

The first fundamental difficulty in proving Theorem~\ref{thm:iso tighter} stems from the inability to determine the exact RDF \(\Rx([N], \ev)\) directly.
While Theorem~\ref{thm:seven regions main} provides an expression for \(\Rx([N], \ev)\), its dependence on numerous parameters from Eq.~\eqref{eq:region4 cubic param} renders analytical proof excessively tedious.
To circumvent this complexity, we adopt a functional perspective by treating \(\Rx([N], \ev)\) as unknown rather than building upon the results in Theorem~\ref{thm:seven regions main}.
Specifically, we establish that in \(\mathcal{R}_4\), the rate corresponding to any distortion matrix \(\Dm\) defined in Eq.~\eqref{eq:dsm model} invariably satisfies Eq.~\eqref{eq:iso tighter}.
However, the determinant of the distortion matrix contains terms with imbalanced powers, making the analysis intractable.
To overcome this, we strategically separate the high-order and low-order terms and subsequently process them on a logarithmic scale.
A deeper challenge lies in the multi-layered non-linearity involving variables, which makes analytical treatment particularly difficult. This necessitates judicious linearization at carefully selected operating points to simplify the analysis while maintaining the validity of the proof.
Finally, a subtle but important detail is the implicit boundary \(\mathcal{B}_{1,4}^{(3)}\) of \(\mathcal{R}_4\). We replace this boundary with the line segment \(\overline{P_1P_4}\) defined by \(e_{2}=1-\rho_{2}\) (see Fig.~\ref{fig:all_regions_2D}). While Theorem~\ref{thm:iso tighter} is established over a broader feasible parameter space, Proposition~\ref{prop:region one B3 e0star} shows that \(\mathcal{B}_{1,4}^{(3)}\) approaches this line segment as the source length grows large.
For an arbitrary \(N \geq 3\), we split the proof into three cases: \(N=3\), \(N=4\), and \(N \geq 5\). We extensively employ mathematical tools such as the arithmetic-geometric mean inequality, Taylor series expansion for linearization, and derivative properties to handle complicated non-linear inequalities involving multiple variables and numerous parameters.

In Fig.~\ref{fig:region4_tight_fixed_e0}, we plot the rate as a function of \(e_1\) for fixed \(e_2 = e_2^\star\), with \(N = 8\) and isotropic correlation \(\rho = 0.45\). Since the Hadamard rate provided in Eq.~\eqref{eq:hadamard bound value} assumes a diagonal distortion matrix \(\Dm = \Em\), the rate \(R^{\ell}_{\xv}([N], \ev)\) exhibits a simple logarithmic relationship with \(e_1\). However, because the SDC does not hold, the optimal rate differs significantly from the Hadamard rate. The upper bound \(R^{u}_{\xv}([N], \ev)\) in Eq.~\eqref{eq:region4 rdf upper} effectively captures this behavior in a compact form while remaining nearly identical to the true value, due to the succinct and accurate bounds on \(\alpha_\star\) and \(\beta_\star\) derived in Proposition~\ref{prop:Dsm bound}.

\begin{figure}[t]
   \setlength{\abovecaptionskip}{0pt} 
   \centering
   \includegraphics[width=0.50\textwidth]{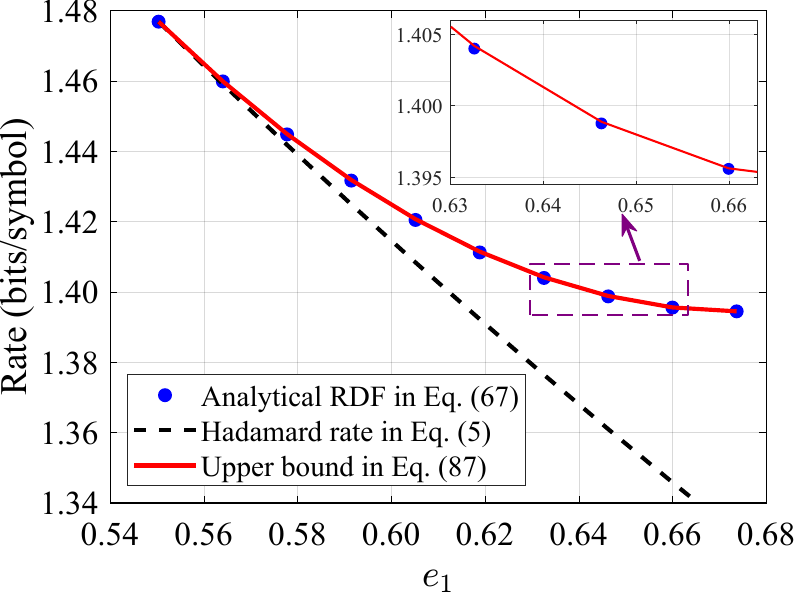}
   \caption{Relationship between the rate-distortion function in \(\mathcal{R}_4\) and the distortion constraint \(e_1\), under the parameters \(N=8\), \(\rho_{1} = \rho_2 = 0.45\), and \(e_{2} = e_{2}^\star\) in Eq.~\eqref{eq:e0star analytical}.}
   \label{fig:region4_tight_fixed_e0}
\end{figure}

\subsection{Analysis under Extreme Correlation and Asymptotics} \label{subsec:Analysis under Extreme Correlation and Asymptotics}

Under the 2-TC covariance matrix and 2-TD constraints, we examine two extreme correlation scenarios in the compression of a vector-valued Gaussian source, where some components are completely independent or dependent. Additionally, we analyze this compression problem in the asymptotic regime as \(N \to \infty\).
Before the extreme correlation analysis, we clarify the valid range of \((\rho_1, \rho_2) \in [0,1]^2\) that ensures \(\Rm \succ \zerov\) in Eq.~\eqref{eq:2tc pd cond}. For any \(\rho_2 \in [0,1]\), we have \(\rho_{1,\min} = 0\) and \(\rho_{1,\max} = \sqrt{\rho_2 + \frac{1 - \rho_2}{N - 1}}\), with \(\rho_{1,\max} = 1\) if and only if \(\rho_2 = 1\).
For any \(\rho_1 \in [0, \frac{1}{\sqrt{N - 1}}]\), we have \(\rho_{2,\min} = 0\), and for any \(\rho_1 \in [0,1]\), we have \(\rho_{2,\max} = 1\).

Let us begin with the extreme case of \(\rho_{1}\).
First, when \(\rho_1 = 0\), only \(\mathcal{R}_0\) and \(\mathcal{R}_5\) exist. Since the central component and the peripheral components are completely independent, the optimality is to compress them separately (\(\alpha = 0\)). Meanwhile, the peripheral components are isotropically correlated, and \(\beta = [e_{2} + \rho_{2} - 1]^+\) implies that their distortions exhibit non-negative correlations.
Second, when \(\rho_1 = \sqrt{\rho_{2}+\frac{1-\rho_{2}}{N-1}}\), \(\mathcal{R}_0\) and \(\mathcal{R}_5\) vanish because optimality cannot be achieved through separate compression. Furthermore, this strong correlation prevents all constraints from being active simultaneously, leading to the disappearance of \(\mathcal{R}_1\).

Next, consider the extreme case of \(\rho_{2}\).
First, when \(\rho_{2} = 0\), the vector source consists of \(N-1\) i.i.d. peripheral components with a commonly correlated central component.
In this case, \(\mathcal{R}_2\) and \(\mathcal{R}_5\) no longer exist.
The absence of \(\mathcal{R}_2\) arises from a contradiction. Recall we have \(\rank(\ev - \dsv)=1\) and \(n_{+}(\Kxs)=1\) in \(\mathcal{R}_2\) (see Table~\ref{tab:region_rank_distortion}). When all peripheral components are i.i.d., \(\rank(\ev - \dsv)=1\) implies that the distortion constraint on the unique central component is inactive, so the reconstruction can be zero almost surely. In contrast, \(n_{+}(\Kxs)=1\) shows that the central component is recovered in a non-trivial manner. This contradiction proves that \(\mathcal{R}_2\) cannot exist.
Additionally, \( \mathcal{R}_5 \) does not exist, because this scenario is impossible: peripheral component distortions are correlated (\( \beta \neq 0 \)) but independent of the central component distortion (\( \alpha = 0 \)).
Second, when \(\rho_{2} = 1\), all peripheral components are fully correlated with each other. In this case, the compression problem for an arbitrarily long vector source degenerates into a two-component source in~\cite[Fig. 1]{nayak2010successive}. The corresponding distortion region comprises \(\mathcal{R}_1\), \(\mathcal{R}_2\), \(\mathcal{R}_3\), and \(\mathcal{R}_5\), and it is symmetric with respect to the line \(e_1 = e_{2}\) (see Fig.~\ref{fig:all_regions_2D}).
Furthermore, when \(\rho_1 = \rho_2 = 1\), the source reduces to a scalar Gaussian-quadratic case. In this degenerate setting, only \(\mathcal{R}_2\) and \(\mathcal{R}_3\) are relevant, and the RDF simplifies to \(\Rx([N], \ev) = -\frac{1}{2} \log \left( \min\{e_1, e_2\}\right)\).

Finally, we analyze the asymptotics of the distortion regions. We find that \(\mathcal{R}_4\) vanishes with increasing source length, as its optimality conditions become mathematically unattainable. In contrast to the two-component source compression problem, this case is particularly interesting as the peripheral component is no longer a single entity but an asymptotically large number of identical components—a case that has not been considered in existing work.
Consequently, the distortion regions lose symmetry.
The regions \(\mathcal{R}_1\), \(\mathcal{R}_2\), \(\mathcal{R}_3\), and \(\mathcal{R}_5\) partition the part of \(\mathcal{U}\) where \(e_{2} > 1 - \rho_{2}\), while the bottom part is divided by \(\mathcal{R}_0\) and \(\mathcal{R}_6\) with the boundary \(e_1 = 1 - \rho_1^2/\rho_{2}\).

\section{Source Compression with Isotropic Correlation and Identical Constraints} \label{sec:Source Compression with Isotropic Correlation and Identical Constraints}

In this section, we consider the case of an isotropically correlated source with identical distortion constraints. Based on Theorem~\ref{thm:seven regions main}, we present the rate-distortion function and the optimal distortion allocation in the following corollary.

\begin{corollary} \label{coro:alliso rd}
   For the isotropic correlation \(\rho = \rho_{1} = \rho_2\) and identical distortion constraints \(e = e_{1} = e_2\), the rate distortion functions under both individual distortion criteria and sum distortion criterion are identical, which are given by
\begin{align} \label{eq:alliso rd v2}
   \Rx([N], \ev)={}& \frac{N}{2} \left[\log \frac{1-\rho}{e}\right]^+ +\frac{1}{2} \log \frac{1 + (N-1)\rho}{1-\rho} 
+\frac{1}{2} \log \frac{e-[e + \rho - 1]^+}{e + (N-1) [e + \rho - 1]^+}.
\end{align}
The optimal distortion allocation is
\begin{align} \label{eq:alliso Dsm}
\Dsm = \Em + [e + \rho - 1]^+ \mathsf{J},
\end{align}
where \(\Em = e \mathsf{I}\) and \(\mathsf{J}=\mathbf{1}\mathbf{1}^{\T}-\mathsf{I}\). Specifically, the average compression rate per component \(\overline{R}_{\xv}([N], \ev) = \frac{1}{N}\Rx([N], \ev)\) asymptotically behaves as  
\begin{align}  \label{eq:alliso rate per asym} 
   \overline{R}_{\xv}([N], \ev) = \begin{cases}  
\frac{1}{2} \log \frac{1-\rho}{e} + \Theta\!\left(\frac{\log N}{N}\right), & \text{if } e + \rho \leq 1, \\  
\Theta\!\left(\frac{1}{N}\right), & \text{if } e + \rho > 1.  
\end{cases}  
\end{align}
\end{corollary}

\begin{figure*}[t]
   \centering
   \captionsetup[subfloat]{font=scriptsize}
   
   \begin{minipage}[b]{0.48\textwidth}
      \centering
      \subfloat[Solid dots mark the condition \(e = 1 - \rho\).]{
         \includegraphics[height=6.06cm]{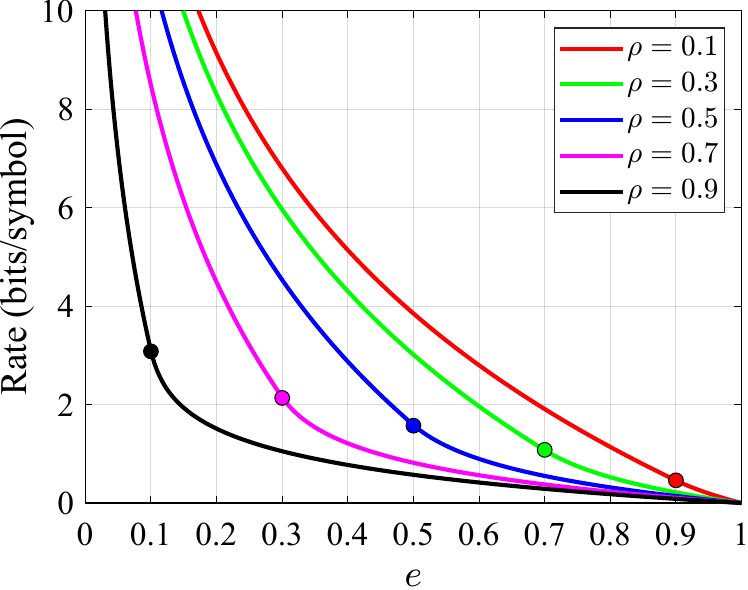}
         \label{fig:iso_uniform_2D_rho}
      }
   \end{minipage}
   \hfill
   \begin{minipage}[b]{0.48\textwidth}
      \centering
      \subfloat[Solid dots mark the condition \(\rho = 1 - e\).]{
         \includegraphics[height=6.1cm]{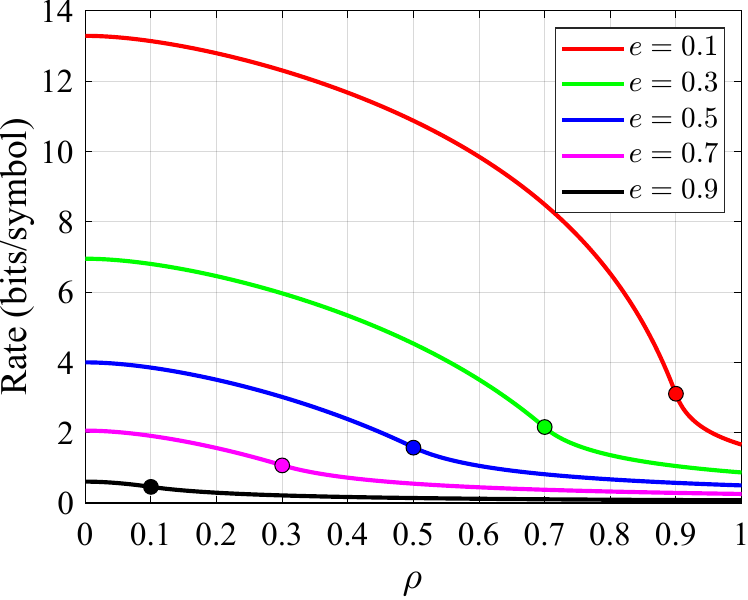}
         \label{fig:iso_uniform_2D_e}
      }
   \end{minipage}
   
   \caption{Relationship between the rate in Eq.~\eqref{eq:alliso rd v2} and distortion constraint and source correlation with \(N=8\) fixed.}
   \label{fig:iso_uniform_2D}
\end{figure*}

\begin{figure}[t]
   \setlength{\abovecaptionskip}{0pt} 
   \centering
   \includegraphics[width=0.50\textwidth]{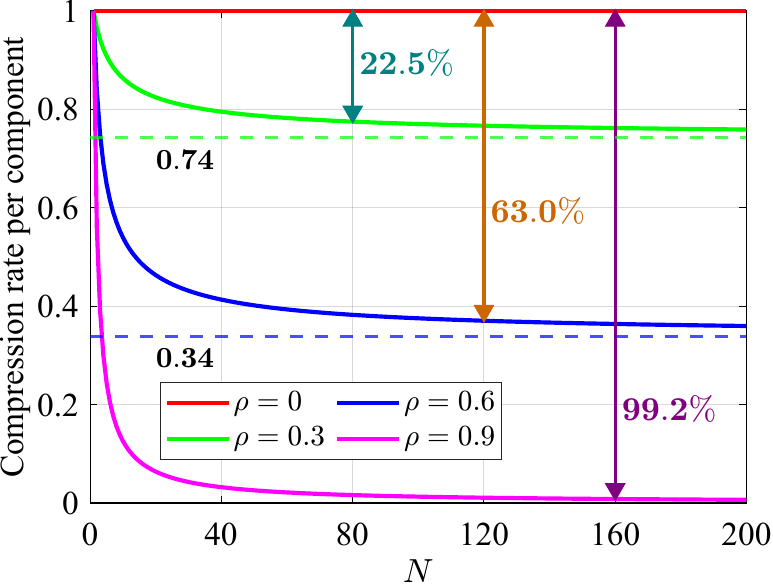}
   \caption{Compression rate per component versus source length under identical distortion constraints \(e=0.25\).}
   \label{fig:iso_uniform_2D_N}
\end{figure}

As mentioned in Eq.~\eqref{eq:rdf vector scalar} and the surrounding context, for the same source and distortion budget, the RDF under individual distortion criteria cannot be lower than that under sum distortion criterion. Corollary~\ref{coro:alliso rd} shows the case where the RDF under individual distortion constraints equals that under sum distortion criterion.
However, the distortion allocation that achieves the RDF differs significantly.
Specifically, under sum distortion criterion, the optimal distortion allocation follows the classical reverse water-filling principle applied to the parallel sources after EVD.
In contrast, under individual distortion constraints, the optimal allocation is performed directly on the original vector source with inter-component correlations. Here, it is necessary not only to meet each component's distortion constraint but also to consider the correlations between distortions across different components. We find that these correlations equal \([e + \rho - 1]^+\), elegantly linking distortion constraints and source correlations.
The SDC holds when \(e + \rho - 1 \leq 0\), and the per-component compression rate \(\overline{R}_{\xv}([N], \ev)\) in Eq.~\eqref{eq:alliso rate per asym} tends to \(\frac{1}{2} \log \frac{1 - \rho}{e}\) with a deviation of \(\Theta\!\left(\frac{\log N}{N}\right)\). If \(e + \rho - 1 > 0\), the optimal reconstruction becomes degenerate with the dimension collapsing to \(n_{+}(\Kxs) = 1\), and the rate tends to zero with a smaller deviation of \(\Theta\!\left(\frac{1}{N}\right)\).

In Fig.~\ref{fig:iso_uniform_2D}, the optimal compression rate's dependence on correlation and distortion constraint is presented for a fixed source length \(N = 8\). 
In Figs.~\ref{fig:iso_uniform_2D_rho} and \ref{fig:iso_uniform_2D_e}, the solid points on each curve correspond to the condition \(e + \rho = 1\).
Generally, for the left segments of the curves (relative to the solid points), the optimal distortion allocation is \(\Dsm = \Em\). For the right segments, the optimal allocation becomes \(\Dsm = \Em + (e + \rho - 1)\mathsf{J}\). 
Specifically, in Fig.~\ref{fig:iso_uniform_2D_rho}, the classical rate-distortion trade-off is observed: a higher tolerable distortion allows for a lower compression rate. Additionally, with overly mild distortion constraints (right segments of the curves), the reduction in compression rate becomes less effective, as the optimal reconstruction degrades.
In Fig.~\ref{fig:iso_uniform_2D_e}, we observe the gains in compression efficiency obtained by fully leveraging correlations.
If the source components exhibit a higher correlation, we can further reduce the compression rate without compromising the reconstruction quality. This is achieved by using the same reconstruction to represent multiple components while meeting their individual distortion constraints.
In Fig.~\ref{fig:iso_uniform_2D_N}, we plot the average compression rate per component versus the source length, with \(e = 0.25\) fixed. 
Compared to the independent case (red line), the compression cost per component is significantly reduced due to correlations. 
For example, under insufficiently strong correlations, i.e., \(\rho \leq 1-e\), \(\overline{R}_{\xv}([N], \ev)\) decreases by 63.0\% when \(N = 120\) and \(\rho = 0.6\), and converges to 0.34 bits per component.
Under stronger correlations, i.e., \(\rho > 1-e\), \(\overline{R}_{\xv}([N], \ev)\) decreases by 99.2\% when \(N = 160\) and \(\rho = 0.9\), which results in a massive rate savings. 
In summary, effectively leveraging correlations in data is a crucial strategy for reducing storage and processing overhead. Furthermore, we have quantitatively observed that the extent of this reduction varies depending on whether the SDC is satisfied or not.

\section{Conclusion} \label{sec:Conclusion}

In this paper, we have investigated the Gaussian-quadratic lossy compression problem under individual distortion criteria.
Theoretical results were provided based on different source covariances and imposed distortion constraints.
First, we have presented the spectral properties of optimal source reconstruction under arbitrary covariance and distortion constraints. 
A stronger version of the SDC was then derived in the form of a scalar inequality.
Next, we proposed the 2-TC class of covariance matrices and prove it generalizes to the broader \((N-1)\)-TC class. We then established the relationship between correlations, distortion constraints, and the optimal compression rate when the SDC holds.
Furthermore, we have characterized the SDC region from two perspectives: distortion constraints and source correlations. 
Subsequently, under the 2-TC covariance and 2-TD constraints, we have provided the complete and exact closed-form RDFs and established the optimality of the distortion allocation in each region. After that, we have refined our results to enhance simplicity while maintaining comparable precision.
An important insight is that the essence of pursuing the RDF under individual distortion criteria lies in thoroughly analyzing the correlations among the optimal distortions across different components.
Under isotropic correlation and identical constraints, we have examined the theoretical limits of the optimal compression rate per component. It is revealed that when the SDC is not satisfied, the system achieves substantial cost reduction at optimal rates by fully exploiting source correlations in practice.

\appendices

\section{
   Proof of Theorem~\ref{thm:det(Rm-Dsm)} \label{proof:thm:det(Rm-Dsm)}
}

To prove Theorem~\ref{thm:det(Rm-Dsm)}, we analyze the optimization problem in Lemma~\ref{lemma:maxdet}. The problem is convex, as the objective \(\log\det(\Dm)\) in Eq.~\eqref{eq:maxdet obj} is strictly concave over the PD cone, and the constraints in Eqs.~\eqref{eq:maxdet c1}-\eqref{eq:maxdet c2} involve only linear inequalities and linear matrix inequalities, which form a convex set. We first exclude the trivial case where the RDF is zero. By definition, \(\Rx([N],\ev)=0\) holds if and only if \(\ev = \mathbf{1}\) (assuming normalized variances), which implies \(\Dsm = \Em = \Rm\). Unless otherwise stated, we focus on the non-trivial case in the sequel.
For any bounded and non-zero \(\Rx([N],\ev)\), there exists a strictly feasible distortion matrix in the interior of the constraint set.
For instance, for sufficiently small \(\epsilon>0\), the matrix \(\Dm=\epsilon\mathsf{I}\) satisfies
\(\zerov\prec\Dm\prec\Rm\) and \(\dv<\ev\).
Hence, Slater's condition holds, which implies strong duality and a zero duality gap~\cite{boyd2004convex}.
Consequently, the KKT conditions are both necessary and sufficient for global optimality.

Let \(\mathsf{P}\succeq \zerov\) and diagonal \(\mathsf{Q}\succeq \zerov\) be the dual slack matrices associated with the constraints \(\Dm\preceq \Rm\) and \(\dv\le\ev\), respectively.
Ignoring the constant \(\frac{1}{2}\log\det(\Rm)\) and the factor \(\frac{1}{2}\), the Lagrangian is given by
\begin{align}
\mathcal{L}(\Dm;\mathsf{P},\mathsf{Q})= -\log\det(\Dm) - \tr\big(\mathsf{P}(\Rm-\Dm)\big) - \tr\big(\mathsf{Q}(\Em-\Dm)\big),
\label{eq:app_Lagrangian}
\end{align}
where \(\Em=\diag(e_1,\cdots,e_N)\).
Let \(\Dsm\) denote an optimal primal solution and \((\mathsf{P}^\star,\mathsf{Q}^\star)\) denote an optimal dual pair.
Under Slater's condition, \(\Dsm\) is optimal if and only if there exist dual variables \((\mathsf{P}^\star,\mathsf{Q}^\star)\) such that the following KKT conditions are satisfied.
\begin{enumerate}
\item Primal feasibility:
The optimal distortion matrix \(\Dsm\) satisfies the primal constraints in Eqs.~\eqref{eq:maxdet c1}-\eqref{eq:maxdet c2}.
\item Dual feasibility: 
The optimal dual variables satisfy the non-negativity constraints
\begin{align}
&\mathsf{P}^\star \succeq \zerov, \\
&\mathsf{Q}^\star=\diag(q_1^\star,\cdots,q_N^\star)\succeq \zerov.
\label{eq:app_KKT_dual}
\end{align}
\item Stationarity:
The stationarity condition states that the gradient of the Lagrangian with respect to \(\Dm\) vanishes at \(\Dsm\). We have \(\nabla_{\Dm}\mathcal{L}(\Dm;\mathsf{P}^\star,\mathsf{Q}^\star) = -\Dm^{-1}+(\mathsf{P}^\star+\mathsf{Q}^\star)\).
Therefore, the optimal distortion matrix satisfies 
\begin{align}
\Dsm=(\mathsf{P}^\star+\mathsf{Q}^\star)^{-1}.
\label{eq:kkt_stat}
\end{align}
\item Complementary slackness:
The vanishing duality gap implies \(\tr\big(\mathsf{P}^\star(\Rm-\Dsm)\big)=0\) and \(\tr\big(\mathsf{Q}^\star(\Em-\Dsm)\big)=0\).
Since all primal and dual slack variables are PSD, the complementary slackness conditions can be expressed compactly. In particular, we have
\begin{align}
\mathsf{P}^\star(\Rm-\Dsm)&=\zerov, \label{eq:kkt_cs1} \\
\mathsf{Q}^\star (\ev-\mathsf{d}^{\star})&=\zerov,
\label{eq:kkt_cs2}
\end{align}
where \(\ev-\mathsf{d}^{\star}\) denotes the vector formed by the diagonal entries of \(\Em-\Dsm\).
\end{enumerate}

We are now ready to establish the two cases in Eq.~\eqref{eq:det(Rm-Dsm)}.

\noindent\textbf{Case I:}
In this case, we have \(\Rm-\Em \succ \zerov\).
Then, the point \(\Dm=\Em\) is strictly feasible for the semidefinite constraint in Eq.~\eqref{eq:maxdet c2}, and satisfies the diagonal constraints in Eq.~\eqref{eq:maxdet c1} with equality.
We next verify that \(\Dm=\Em\) satisfies the KKT conditions derived above to demonstrate its optimality. To this end, let 
\begin{align}
\mathsf{P}^\star&=\zerov, \\
\mathsf{Q}^\star&=\Em^{-1}.
\label{eq:caseI_dual_choice}
\end{align}
Then dual feasibility holds since \(\mathsf{P}^\star\succeq\zerov\) and \(\mathsf{Q}^\star\succeq\zerov\) (note that \(\Em\succ\zerov\) under any bounded non-zero RDF). Moreover, the stationarity condition in Eq.~\eqref{eq:kkt_stat} is satisfied by construction. Finally, the complementary slackness conditions hold trivially: Eq.~\eqref{eq:kkt_cs1} is satisfied because \(\mathsf{P}^\star=\zerov\), and Eq.~\eqref{eq:kkt_cs2} is satisfied because the distortion constraints are active, i.e., \(\ev - \dsv = \zerov\).
Therefore, the triple \((\Dm,\mathsf{P}^\star,\mathsf{Q}^\star)=(\Em,\zerov,\Em^{-1})\) satisfies all KKT conditions.
Since Slater's condition holds, this confirms \(\Dsm=\Em\) as the unique global optimum.
Recall the optimal reconstruction covariance matrix is given by \(\Kxs = \Rm - \Dsm \succ\zerov\), \(\Kxs\) is strictly full-rank, which establishes the first case stated in Eq.~\eqref{eq:det(Rm-Dsm)}.

\noindent\textbf{Case II:} This encompasses two subcases: 
(i) \(\Rm-\Em\succeq\zerov\) and \(\det(\Rm-\Em)=0\); and
(ii) the SDC is violated, i.e., \(\Rm-\Em\nsucceq\zerov\).

\smallskip
\noindent (i) \(\Rm-\Em\succeq\zerov\) and \(\det(\Rm-\Em)=0\).
In this subcase, \(\Dm=\Em\) remains primal feasible.
The optimality proof from Case~I extends directly here: the triple \((\Dm,\mathsf{P}^\star,\mathsf{Q}^\star)=(\Em,\zerov,\Em^{-1})\) satisfies all KKT conditions.
Combining \(\Dsm=\Em\) with the condition \(\det(\Rm-\Em)=0\), we obtain \(\det(\Rm-\Dsm) = 0\). 
This implies that the resulting reconstruction covariance matrix \(\Kxs = \Rm-\Dsm\) is rank-deficient.

\noindent (ii) \(\Rm-\Em\nsucceq\zerov\).
In this subcase, \(\Dm=\Em\) is not primal feasible, hence \(\Dsm\neq \Em\).
We proceed to prove that \(\Rm-\Dsm\) must be singular. To this end, assume to the contrary that
\begin{align}
\Rm-\Dsm\succ\zerov.
\label{eq:caseII_contra_strict}
\end{align}
This implies that the semidefinite constraint in Eq.~\eqref{eq:maxdet c2} is strictly inactive at \(\Dsm\), and the complementary slackness condition in Eq.~\eqref{eq:kkt_cs1} dictates that \(\mathsf{P}^\star=\zerov\).
Substituting \(\mathsf{P}^\star=\zerov\) into the stationarity condition in Eq.~\eqref{eq:kkt_stat}, the optimal distortion matrix simplifies to the diagonal form \(\Dsm=(\mathsf{Q}^\star)^{-1}\), which implies that \(\mathsf{Q}^\star\) is strictly PD (and diagonal) since \(\Dsm \succ \zerov\).
To satisfy the component-wise complementary slackness condition in Eq.~\eqref{eq:kkt_cs2}, we have \(\ev - \mathsf{d}^\star = \mathsf{0}\).
Combined with the diagonal structure, this implies \(\Dsm=\Em\).
However, substituting \(\Dsm=\Em\) back into the initial assumption in Eq.~\eqref{eq:caseII_contra_strict} would imply \(\Rm-\Em \succ \zerov\), which contradicts the hypothesis of this subcase (\(\Rm-\Em \nsucceq \zerov\)).
This contradiction shows that the assumption in Eq.~\eqref{eq:caseII_contra_strict} is false.
Given primal feasibility \(\Rm-\Dsm \succeq \zerov\), we conclude that \(\Kxs = \Rm-\Dsm\) is rank-deficient.

\section{
   Proof of Theorem~\ref{thm:inertia relation} \label{proof:thm:inertia relation}
}

To prove Theorem~\ref{thm:inertia relation}, we simplify the notation by defining \(\mathsf{A} = \Rm-\Dsm\) and \(\mathsf{B} = \Rm-\Em\).
We observe that since \(\mathsf{A} \succeq \zerov\), its rank is identically equal to the number of its strictly positive eigenvalues.
Proceeding from the KKT conditions established in Appendix~\ref{proof:thm:det(Rm-Dsm)}, we substitute the stationarity condition in Eq.~\eqref{eq:kkt_stat} into the complementary slackness equation in Eq.~\eqref{eq:kkt_cs1}. This eliminates \(\mathsf{P}^\star\) and yields \(\bigl((\Dsm)^{-1}-\mathsf{Q}^\star\bigr)\mathsf{A}=\zerov\).
Left-multiplying by \(\Dsm\) establishes the fundamental identity
\begin{align}\label{eq:A_DQA}
\mathsf{A} = \Dsm \mathsf{Q}^\star \mathsf{A}.
\end{align} 

We are now positioned to prove the first part of the inequality in Eq.~\eqref{eq:inertia relation}.
Since \(\Dsm\) is strictly PD, multiplication by its inverse preserves rank. Using the identity in Eq.~\eqref{eq:A_DQA}, we have \(\rank(\mathsf{A}) = \rank\bigl((\Dsm)^{-1}\mathsf{A}\bigr) = \rank(\mathsf{Q}^\star \mathsf{A})\).
Applying the rank inequality \(\rank(\mathsf{XY}) \le \min\{\rank(\mathsf{X}),\rank(\mathsf{Y})\}\), we obtain
\begin{align}
\rank(\mathsf{A}) \le \rank(\mathsf{Q}^\star).
\end{align}
Furthermore, the complementary slackness condition regarding the distortion constraints in Eq.~\eqref{eq:kkt_cs2} implies
\begin{align}
\rank(\mathsf{Q}^\star) + \rank(\ev-\mathsf{d}^{\star}) \le N,
\end{align}
where \(\rank(\ev-\dv^\star)\) denotes the cardinality of the support (number of non-zero entries) of the vector \(\ev-\dv^\star\).
This leads to the following chain of inequalities:
\begin{align}
\rank(\mathsf{A}) \le \rank(\mathsf{Q}^\star) \le N - \rank(\ev-\mathsf{d}^{\star}).
\end{align}

To prove the second part of Eq.~\eqref{eq:inertia relation}, i.e., \(n_+(\mathsf{B}) \ge n_+(\mathsf{A})\), we first establish a fundamental inequality on the optimal reconstruction distortion.
From the dual feasibility condition \((\Dsm)^{-1}-\mathsf{Q}^\star = \mathsf{P}^\star \succeq \zerov\), the definition of positive semidefiniteness implies that for any vector \(\mathsf{x}\), we have \(\mathsf{x}^{\T} \bigl((\Dsm)^{-1}-\mathsf{Q}^\star\bigr) \mathsf{x} \ge 0\).
Specifically, we choose \(\mathsf{x} = \Dsm \mathsf{b}_i\), where \(\mathsf{b}_i = [0, \cdots, 1, \cdots, 0]^{\T}\) is the unit vector with the \(i\)-th entry being \(1\) and all others \(0\).
We then obtain
\begin{align}
(\Dsm \mathsf{b}_{i})^{\T} (\Dsm)^{-1} (\Dsm \mathsf{b}_{i}) - (\Dsm \mathsf{b}_{i})^{\T} \mathsf{Q}^\star (\Dsm \mathsf{b}_{i}) \ge 0.
\end{align}
The first term simplifies to \(\mathsf{b}_{i}^{\T} \Dsm \mathsf{b}_{i} = \Dsm_{ii}\), i.e., the \(i\)-th diagonal element of \(\Dsm\).
For the second term, recall that \(\mathsf{Q}^\star = \diag(q_1^\star, \cdots, q_N^\star)\) with non-negative entries. Thus, we have
\begin{align}
(\Dsm \mathsf{b}_{i})^{\T} \mathsf{Q}^\star (\Dsm \mathsf{b}_{i})
&= \sum_{j=1}^N q_j^\star |\Dsm_{ji}|^2 \\
&\ge q_i^\star |\Dsm_{ii}|^2.
\end{align}
Combining this with the first term \(\Dsm_{ii}\) yields \(\Dsm_{ii} \ge q_i^\star |\Dsm_{ii}|^2\).
Given that \(\Dsm \succ \zerov\) implies \(\Dsm_{ii} > 0\), we divide both sides by \(\Dsm_{ii}\) to establish the following key inequality for all \(i \in [N]\):
\begin{align} \label{eq:qi_bound}
q_i^\star \Dsm_{ii} \le 1.
\end{align}

Define the subspace \(\mathcal W \triangleq \mathrm{Range}(\mathsf{Q}^\star \mathsf{A})=\{\mathsf{Q}^\star \mathsf{A} \mathsf{x}:\forall\, \mathsf{x}\in\mathbb{C}^N\}\).
Invoking the identity \(\mathsf{Q}^\star \mathsf{A}=(\Dsm)^{-1}\mathsf{A}\) from Eq.~\eqref{eq:A_DQA}, and noting the invertibility of \((\Dsm)^{-1}\), we obtain \(\dim(\mathcal W)=\rank(\mathsf{A})\).
We now proceed to show that \(\mathsf{w}^{\T}(\mathsf{B}-\mathsf{A})\mathsf{w}=\mathsf{w}^{\T}(\mathsf{D}^\star-\Em)\mathsf{w}\ge 0\) for all \(\mathsf{w}\in\mathcal W\).
Fix \(\mathsf{w}\in\mathcal W\). Then there exists a vector \(\mathsf{u}\) such that \(\mathsf{w}=\mathsf{Q}^\star \mathsf{A} \mathsf{u}\). Let \(\mathsf{v}\triangleq \mathsf{A}\mathsf{u}\), we have \(\mathsf{w}=\mathsf{Q}^\star \mathsf{v}\).
Using Eq.~\eqref{eq:A_DQA} gives
\begin{align}\label{eq:v_equals_Dw}
\mathsf{v}=\mathsf{D}^\star \mathsf{w}.
\end{align} 
Recall \(\mathsf{Q}^\star=\diag(q_1^\star,\cdots,q_N^\star)\) and let the active set be \(\mathcal{S}\triangleq\{i: q_i^\star>0\}\). Since \(\mathsf{w}_i=q_i^\star \mathsf{v}_i\), we have \(\mathsf{w}_i=0\) for all \(i\notin \mathcal{S}\).
Moreover, from Eq.~\eqref{eq:kkt_cs2}, for each \(i\in \mathcal{S}\), we have \(\mathsf{D}^\star_{ii}=e_i\).
Combining this with the key inequality in Eq.~\eqref{eq:qi_bound} yields 
\begin{align}
   \frac{1}{q_i^\star}-e_i \geq 0
\end{align} 
for all \(i\in \mathcal{S}\).
Now compute
\begin{align}
\mathsf{w}^{\T}\mathsf{D}^\star \mathsf{w} &= \mathsf{v}^{\T}(\mathsf{D}^\star)^{-1}\mathsf{v} \label{eq:calc_v_inv} \\
&= \mathsf{v}^{\T}\mathsf{Q}^\star \mathsf{v} \label{eq:calc_P_zero} \\
&= \sum_{i\in \mathcal{S}} q_i^\star |\mathsf{v}_i|^2 \label{eq:calc_sum_S} \\
&=\sum_{i\in \mathcal{S}}\frac{|\mathsf{w}_i|^2}{q_i^\star}, \label{eq:calc_sub_w}
\end{align}
where Eq.~\eqref{eq:calc_v_inv} follows from the identity in Eq.~\eqref{eq:v_equals_Dw}; Eq.~\eqref{eq:calc_P_zero} follows from the identity in Eq.~\eqref{eq:A_DQA} and \(\mathsf{v} \in \text{Range}(\mathsf{A})\);
Eq.~\eqref{eq:calc_sum_S} restricts the summation to the support set \(\mathcal{S}\);
and Eq.~\eqref{eq:calc_sub_w} utilizes the relation \(\mathsf{w}_i = q_i^\star \mathsf{v}_i\) for \(i \in \mathcal{S}\).
On the other hand, since \(\Em\) is diagonal and \(\mathsf{w}\) is supported on \(\mathcal{S}\), we get \(\mathsf{w}^{\T}\Em\mathsf{w}=\sum_{i\in \mathcal{S}} e_i |\mathsf{w}_i|^2\).
Consequently, we establish the following relation for all \(\mathsf{w} \in \mathcal{W}\):
\begin{align}\label{eq:wDEw_nonneg}
\mathsf{w}^{\T}(\mathsf{B}-\mathsf{A})\mathsf{w}&=\mathsf{w}^{\T}(\mathsf{D}^\star-\Em)\mathsf{w}\\
&=\sum_{i\in \mathcal{S}}|\mathsf{w}_i|^2\Bigl(\frac{1}{q_i^\star}-e_i\Bigr)\\
&\ge 0.
\end{align} 
We then prove that \(\mathsf{w}^{\T}\mathsf{A}\mathsf{w}>0\) for all non-zero \(\mathsf{w}\in\mathcal W\).
Since \(\mathsf{A}\succeq \zerov\), we have \(\mathsf{w}^{\T}\mathsf{A}\mathsf{w}=0\) if and only if \(\mathsf{A}\mathsf{w}=\zerov\).
We claim \(\mathcal W\cap \ker(\mathsf{A})=\{\mathsf{0}\}\). Indeed, for \(\mathsf{w}\in\mathcal W\), we write \(\mathsf{w}=\mathsf{Q}^\star \mathsf{A} \mathsf{u}=(\mathsf{D}^\star)^{-1}\mathsf{A}\mathsf{u}\) in view of Eq.~\eqref{eq:A_DQA}.
If \(\mathsf{A}\mathsf{w}=\zerov\), then \(\mathsf{A}(\mathsf{D}^\star)^{-1}\mathsf{A}\mathsf{u}=\zerov\). Left-multiplying by \(\mathsf{u}^{\T}\) yields \((\mathsf{A}\mathsf{u})^{\T}(\mathsf{D}^\star)^{-1}(\mathsf{A}\mathsf{u})=\zerov\).
Because \((\mathsf{D}^\star)^{-1}\succ \zerov\), this implies \(\mathsf{A}\mathsf{u}=\zerov\), and hence \(\mathsf{w}=(\mathsf{D}^\star)^{-1}\mathsf{A}\mathsf{u}=\zerov\).
Thus we conclude \(\mathcal W\cap \ker(\mathsf{A})=\{\mathsf{0}\}\), so \(\mathsf{w}^{\T}\mathsf{A}\mathsf{w}>0\) for all non-zero \(\mathsf{w}\in\mathcal W\).

We now establish the positivity of \(\mathsf{B}\) on \(\mathcal{W}\) and deduce the conclusion via a variational characterization.
Observing the decomposition \(\mathsf{B} = \mathsf{A} + (\Dsm-\Em)\), and considering any non-zero vector \(w \in \mathcal{W}\), we combine the result in Eq.~\eqref{eq:wDEw_nonneg} with the strict positivity \(\mathsf{w}^{\T}\mathsf{A}\mathsf{w} > 0\) to obtain
\begin{align}
\mathsf{w}^{\T}\mathsf{B}\mathsf{w} > 0.
\end{align}
This implies that \(\mathsf{B}\) is PD on the subspace \(\mathcal{W}\).
Finally, we invoke the standard variational characterization of the number of positive eigenvalues of a Hermitian matrix \(\mathsf{M}\), i.e., \(n_+(\mathsf{M}) \triangleq \max\bigl\{\dim(\mathcal X):\mathsf{x}^{\T}\mathsf{M}\mathsf{x}>0, \forall \, \mathsf{x}\in\mathcal X\setminus\{\mathsf{0}\}\bigr\}\).
Since we have identified a subspace \(\mathcal{W}\) with \(\dim(\mathcal{W}) = n_+(\mathsf{A})\) on which \(\mathsf{B}\) is strictly PD, we conclude that
\begin{align}
n_+(\mathsf{B}) \ge \dim(\mathcal{W}) = n_+(\mathsf{A}).
\end{align}
This establishes the second part of the inequality in Eq.~\eqref{eq:inertia relation}, thereby completing the proof of Theorem~\ref{thm:inertia relation}.

\section{
   Proof of Lemma~\ref{lemma:determinant} \label{proof:lemma:determinant}
}

We first observe that the matrix \(\Am \triangleq \Rm_2 - \Gammam \) can be decomposed as
\(\Am = \Bm + \rho_{2} \mathbf{1}\mathbf{1}^{\T}\), where
\begin{align}
\Bm = \begin{pmatrix}
      1 - \rho_{2} - \gamma_1 & \vv^{\T} \\
      \vv & \Cm
\end{pmatrix},
\end{align}
\(\vv = [\rho_1 - \rho_{2}, \cdots, \rho_1 - \rho_{2}]^{\T}\), and \(\Cm = \diag(1 - \rho_{2} - \gamma_2, \cdots, 1 - \rho_{2} - \gamma_N) \in \mathbb{R}^{(N-1) \times (N-1)}\). The matrix \(\Bm\) is an arrowhead matrix, and its determinant is given by
\begin{align} \label{eq:determinant Bm}
\det(\Bm) = \xi \prod_{i=2}^{N} (1 - \rho_{2} - \gamma_i),
\end{align}
where \(\xi = 1 - \rho_{2} - \gamma_1 - \sum_{i=2}^{N} \frac{(\rho_1 - \rho_{2})^2}{1 - \rho_{2} - \gamma_i}\). The inverse of \(\Bm\) can be expressed as
\begin{align} \label{eq:inverse Bm}
\Bm^{-1} =
\begin{pmatrix}
      0 & \\
      & \Cm^{-1}
\end{pmatrix} + \frac{1}{\xi} \uv \uv^{\T},
\end{align}
where \(\uv = \big[-1, \frac{\rho_1 - \rho_{2}}{1 - \rho_{2} - \gamma_2}, \cdots, \frac{\rho_1 - \rho_{2}}{1 - \rho_{2} - \gamma_N}\big]^{\T}\). Applying the matrix determinant lemma~\cite{harville1998matrix}, we have
\begin{align} \label{eq:Delta determinant origin}
\Delta(\rho_{2}, \rho_{1}, \Gammam) = \det(\Bm) \cdot \Big(1 + \rho_{2} \sum_{i,j} \left(\Bm^{-1}\right)_{i,j}\Big).
\end{align}
Substituting Eqs.~\eqref{eq:determinant Bm} and~\eqref{eq:inverse Bm} into Eq.~\eqref{eq:Delta determinant origin} and simplifying yields Eq.~\eqref{eq:Delta determinant}.

\section{
   Proof of Theorem~\ref{thm:2TC_extension} \label{proof:thm:2TC_extension}
}

For simplicity, we first define \(\mathsf{A}_n \triangleq \Rm_n - \Gammam\).  
We then multiply the \((n-1)\)-th row and the \((n-1)\)-th column of \(\mathsf{A}_n\) by the factor \(\frac{\rho_n}{\rho_{n-1}}\), and obtain
\(\det(\mathsf{A}_{n})=\left(\frac{\rho_{n-1}}{\rho_{n}}\right)^2\det(\mathsf{B}_{n})\),
where the resulting matrix admits the following \(2 \times 2\) block structure:
\begin{align} \label{eq:block matrix Bn}
\mathsf{B}_{n}=
\begin{pmatrix}
\mathsf{T} & \mathsf{S} \\
\mathsf{S}^{\T} & \mathsf{Q}
\end{pmatrix}.
\end{align}
Specifically, we have \(\mathsf{T} \in \mathbb{R}^{(n-2)\times(n-2)}\) and for \(i,j \in [n-2]\),
\begin{align} \label{eq:T matrix}
\mathsf{T}(i,j) =
\begin{cases}
a_i, & \text{if } i = j, \\[-6pt]
\rho_{\min\{i,j\}}, & \text{if } i \neq j,
\end{cases}
\end{align}
\begin{align} \label{eq:S matrix}
\mathsf{S} = \boldsymbol{\rho} \mathsf{v}^{\T} \in \mathbb{R}^{(n-2)\times(N-n+2)},
\end{align}
where \(\boldsymbol{\rho} = [\rho_1, \rho_2, \cdots, \rho_{n-2}]^{\T}\), \(\mathsf{v} = \big[\frac{\rho_n}{\rho_{n-1}}, 1, 1, \cdots, 1\big]^{\T}\) is of length \(N-n+2\), and for \(i,j \in [N-n+2]\),
\begin{align} \label{eq:Q matrix}
\mathsf{Q}(i,j) =
\begin{cases}
a_{n-1} \left( \dfrac{\rho_n}{\rho_{n-1}} \right)^2, & \text{if } i = j = 1, \\[-6pt]
a_{n-2+i}, & \text{if } i = j \neq 1, \\[-6pt]
\rho_n, & \text{if } i \neq j.
\end{cases}
\end{align}
When \(\mathsf{Q}\) is invertible, using the block matrix \(\mathsf{B}_n\) in Eq.~\eqref{eq:block matrix Bn}, we have
\begin{align} 
\det(\mathsf{B}_{n}) &= \det(\mathsf{Q}) \cdot \det\!\left(\mathsf{T} - \mathsf{S} \mathsf{Q}^{-1} \mathsf{S}^{T}\right) \label{eq:det B_n} \\
&= \det(\mathsf{Q}) \cdot \det\!\left(\mathsf{T} - (\mathsf{v}^{\T} \mathsf{Q}^{-1} \mathsf{v}) \boldsymbol{\rho} \boldsymbol{\rho}^{\T}\right), \label{eq:det B_n v1}
\end{align}
where Eq.~\eqref{eq:det B_n v1} is derived using the form of \(\mathsf{S}\) in Eq.~\eqref{eq:S matrix}.
Similarly, we have
\begin{align}
\det(\mathsf{A}_{n-1}) 
&= \det(\overline{\mathsf{Q}}) \cdot \det\!\left(\overline{\mathsf{T}} - (\overline{\mathsf{v}}^{\T} \overline{\mathsf{Q}}^{-1} \overline{\mathsf{v}}) \mathop{\overline{\boldsymbol{\rho}}} \mathop{\overline{\boldsymbol{\rho}}^{\T}}\right),  \label{eq:det A_n-1 v1}
\end{align}
where \(\overline{\boldsymbol{\rho}} = [\overline{\rho}_1, \overline{\rho}_2, \cdots, \overline{\rho}_{n-2}]^{\T}\), \(\overline{\mathsf{v}} = \left[1, 1, \cdots, 1\right]^{\T}\) is of length \(N-n+2\), and \(\overline{\mathsf{T}}\) is obtained by replacing all \(a_i\) and \(\rho_i\) in \(\mathsf{T}\) in Eq~\eqref{eq:T matrix} with \(\overline{a}_i\) and \(\overline{\rho}_i\), respectively. The structure of \(\overline{\mathsf{Q}}\) is the same as that of \(\mathsf{Q}\) in Eq~\eqref{eq:Q matrix}, with diagonal entries ranging from \(\overline{a}_{n-1}\) to \(\overline{a}_N\) and off-diagonal entries \(\overline{\rho}_{n-1}\).
Comparing Eqs.~\eqref{eq:det B_n v1} and~\eqref{eq:det A_n-1 v1}, we aim to achieve  
\begin{align} \label{eq:equivalence cond}
   \det\!\left(\overline{\mathsf{T}} - (\overline{\mathsf{v}}^{\T} \overline{\mathsf{Q}}^{-1} \overline{\mathsf{v}}) \mathop{\overline{\boldsymbol{\rho}}} \mathop{\overline{\boldsymbol{\rho}}^{\T}}\right)=\det\!\left(\mathsf{T} - (\mathsf{v}^{\T} \mathsf{Q}^{-1} \mathsf{v}) \boldsymbol{\rho} \boldsymbol{\rho}^{\T}\right),
\end{align} 
so that the desired result
\begin{align}
\det(\mathsf{A}_n) = \left( \frac{\rho_{n-1}}{\rho_n} \right)^2 \frac{\det(\mathsf{Q})}{\det(\overline{\mathsf{Q}})} \det(\mathsf{A}_{n-1}) \label{eq:A_n det final}
\end{align}
can be established.
A sufficient condition for Eq.~\eqref{eq:equivalence cond} to hold is  
\begin{numcases}{}
   \overline{\mathsf{T}} = \mathsf{T}, \label{eq:equivalence T n-2} \\
   \overline{\boldsymbol{\rho}} = \boldsymbol{\rho}, \label{eq:equivalence rho n-2} \\
   \overline{\mathsf{v}}^{\mathsf{T}} \overline{\mathsf{Q}}^{-1} \overline{\mathsf{v}} = \mathsf{v}^{\mathsf{T}} \mathsf{Q}^{-1} \mathsf{v}. \label{eq:trans equival cond}
\end{numcases}
Eqs.~\eqref{eq:equivalence T n-2} and~\eqref{eq:equivalence rho n-2} are equivalent to Eqs.~\eqref{eq:extension T rho cond1} and~\eqref{eq:extension T rho cond2}.
Note that \(\overline{\mathsf{Q}}\) and \(\mathsf{Q}\) have identical off-diagonals. Taking \(\overline{\mathsf{Q}} = \overline{\mathsf{D}} + \overline{\rho}_{n-1} \mathop{\overline{\mathsf{v}}} \mathop{\overline{\mathsf{v}}^{\T}}\) as an example with \(\overline{\mathsf{D}} = \diag(\overline{a}_{n-1}, \overline{a}_n, \cdots, \overline{a}_N) - \overline{\rho}_{n-1} \mathsf{I}\), where \(\mathsf{I}\) is the identity matrix, we can directly derive its determinant \(\det(\overline{\mathsf{Q}})\) in Eq.~\eqref{eq:extension Q bar det}
from Lemma~\ref{lemma:determinant}.
Using the Sherman-Morrison-Woodbury formula, \(\overline{\mathsf{Q}}^{-1}\) is computed as  
\begin{align}
   \overline{\mathsf{Q}}^{-1} 
   = \overline{\mathsf{D}}^{-1} - \frac{\overline{\rho}_{n-1} \overline{\mathsf{D}}^{-1} \mathop{\overline{\mathsf{v}}} \mathop{\overline{\mathsf{v}}^{\T}} \overline{\mathsf{D}}^{-1}}{1 + \overline{\rho}_{n-1} \overline{\mathsf{v}}^{\T} \overline{\mathsf{D}}^{-1} \overline{\mathsf{v}}},
\end{align} 
where
\begin{align}
   \overline{\mathsf{v}}^{\T} \overline{\mathsf{D}}^{-1} \overline{\mathsf{v}} &= \sum_{k=n-1}^N \frac{1}{\overline{a}_k-\overline{\rho}_{n-1}}, \\
(\overline{\mathsf{D}}^{-1} \mathop{\overline{\mathsf{v}}} \mathop{\overline{\mathsf{v}}^{\T}} \overline{\mathsf{D}}^{-1})_{ij} &= \frac{1}{\overline{a}_{n-2+i}-\overline{\rho}_{n-1}}\frac{1}{\overline{a}_{n-2+j}-\overline{\rho}_{n-1}}.
\end{align} 
Thus, \(\overline{\mathsf{Q}}^{-1}\) and \(\mathsf{Q}^{-1}\) are given in Eq.~\eqref{eq:Q bar inverse entry} and Eqs.~\eqref{eq:Q inverse entry1}-\eqref{eq:Q inverse entry4}, respectively.
Eq.~\eqref{eq:trans equival cond} can be written in entry-wise form as Eq.~\eqref{eq:trans equival cond entry}.
In summary, for given \(\rho_{n-1}\), \(\rho_n\), \(a_{n-1}, \cdots, a_N\), we determine the right-hand side of Eq.~\eqref{eq:trans equival cond entry} through \(\mathsf{Q}^{-1}\), and then determine \(\overline{\rho}_{n-1}\), \(\overline{a}_{n-1}, \cdots, \overline{a}_N\) through \(\overline{\mathsf{Q}}^{-1}\).
Strictly speaking, solving for the \(N - n + 3\) quantities from Eq.~\eqref{eq:trans equival cond entry} should yield infinitely many solutions, but since Eq.~\eqref{eq:A_n det final} also involves \(\det(\overline{\mathsf{Q}})\), we obtain a unique determinant of \(\mathsf{A}_n\).

\section{
   Proof of Theorem~\ref{thm:RE_psd_conditions} \label{proof:thm:RE_psd_conditions}
}

The necessary and sufficient condition for \(\Rm - \Em \succeq \zerov\) is that all principal minors of \(\Rm - \Em\) are non-negative, as stated by Sylvester's criterion~\cite{horn2012matrix}.
Define two index sets: \(\mathcal{I}_{\backslash \{1\}} \subset [N]\), which represents the indices of principal minors of \(\Rm - \Em\) that exclude the central component, and \(\mathcal{I}_{\cup \{1\}} \subseteq [N]\), which represents the indices of principal minors that include the central component. We analyze these two cases separately.

First, consider the case excluding the central component. By Lemma~\ref{lemma:determinant}, the determinant of \((\Rm - \Em)[\mathcal{I}_{\backslash \{1\}}]\) is given by
\begin{align}
\Delta(\rho_{2},\ev,\mathcal{I}_{\backslash \{1\}}) =& \Bigg( 1 + \sum_{i \in \mathcal{I}_{\backslash \{1\}}} \frac{\rho_{2}}{1 - \rho_{2} - e_i} \Bigg)
\cdot \prod_{i \in \mathcal{I}_{\backslash \{1\}}} (1 - \rho_{2} - e_i).
\end{align}
If \(e_2 + \rho_{2} \leq 1\), it follows that \(e_i + \rho_{2} \leq 1\) for all \(i \in [N] \setminus \{1\}\). Under this condition, \((\Rm - \Em)[\mathcal{I}_{\backslash \{1\}}]\) is PSD for any \(\mathcal{I}_{\backslash \{1\}}\). Conversely, if there exists some \(j\) such that \(e_j + \rho_{2} > 1\), the condition
\begin{align} \label{eq:RE_psd_uniqueness}
\frac{\rho_{2}}{\rho_{2} + e_j - 1} \geq 1 + \sum_{i \in \mathcal{I}_{\backslash \{1,j\}}} \frac{\rho_{2}}{1 - \rho_{2} - e_i}
\end{align}
should hold. To prove the uniqueness of \(j\), suppose there exists another \(k \neq j\) such that both \(j\) and \(k\) satisfy Eq.~\eqref{eq:RE_psd_uniqueness}. The \(2 \times 2\) principal minor corresponding to the subset \(\{j, k\}\) is
\begin{align} \label{eq:RE_psd_2minor}
\det\big((\Rm - \Em)[\{j, k\}]\big) = (1 - e_j)(1 - e_k) - \rho_{2}^2.
\end{align}
Since \(e_j + \rho_{2} > 1\) and \(e_k + \rho_{2} > 1\), it follows that \(1 - e_j < \rho_{2}\) and \(1 - e_k < \rho_{2}\). As a result, the minor in Eq.~\eqref{eq:RE_psd_2minor} is non-positive. This contradicts the positive semidefiniteness of \(\Rm - \Em\). Therefore, \(j\) should be unique.
Moreover, since \(\{e_i\}_{i=2}^{N}\) are in non-increasing order, at most only \(e_2 > 1-\rho_{2}\) satisfies the condition in Eq.~\eqref{eq:RE_psd_uniqueness}.
For any \(\mathcal{I}_{\backslash \{1,j\}} \subset [N]\), this condition holds, and rewriting it leads to Eqs.~\eqref{eq:e3_condition} and~\eqref{eq:e2_condition}.

Next, consider the case including the central component. Based on Eq.~\eqref{eq:e3_condition}, the positive semidefiniteness of \((\Rm - \Em)[\mathcal{I}_{\cup \{1\}}]\) is equivalently expressed as
\begin{align} \label{eq:RE_psd_exclude}
   \!\!\! \Bigg( 1 - e_1 + \sum_{i \in \mathcal{I}_{\backslash \{1\}}} \frac{\rho_{2} - \rho_1^2 - \rho_{2} e_1}{1 - \rho_{2} - e_i} \Bigg) (1 - \rho_{2} - e_2) \geq 0.
\end{align}
By replacing the set \(\mathcal{I}_{\backslash \{1\}}\) with \([N]\) and performing simplifications, Eq.~\eqref{eq:RE_psd_exclude} reduces to Eq.~\eqref{eq:e1_condition}.

Finally, due to the Hadamard compression rate in Eq.~\eqref{eq:hadamard bound value} and by applying Lemma~\ref{lemma:determinant} to derive \(\det(\Rm)\), we obtain the explicit RDF in Eq.~\eqref{eq:rd sdc1}.

\section{
   Proof of Theorem~\ref{thm:sdc exponential asymp} \label{proof:thm:sdc exponential asymp}
}

The \(N\)-fold integral in Eq.~\eqref{eq:SDC prob} is obtained from Eqs.~\eqref{eq:e3_condition}-\eqref{eq:e1_condition} in Theorem~\ref{thm:RE_psd_conditions}, where the upper bound of \(e_2\) in Eq.~\eqref{eq:e2_condition} is improved by Eq.~\eqref{eq:e1_condition} to prevent the divergence of the integral at boundary points.   
The term \(N-1\) arises from the ordering of the constraints \(\{e_{i}\}_{i=2}^{N}\).

   We first analyze the inner double integral and then proceed to evaluate the outer \((N-2)\)-fold integral. 
      After computing the inner double integral, we obtain
   \begin{align}
      &P(\mathcal{A}_0) \notag \\
      ={}& (N-1)\underbrace{\int_{0}^{1-\rho_{2}}  \cdots \int_{0}^{1-\rho_{2}}}_{N-2 \text{ times}}  \left(1-\frac{\rho_{1}^{2}}{\rho_{2}}+\frac{\rho_{1}^{2}}{\rho_{2}\left(1+\rho_{2}\chi_{3}\right)}\right)\left(1-\rho_{2}-\frac{\rho_{1}^{2}-\rho_{2}}{1-\left(\rho_{1}^{2}-\rho_{2}\right)\chi_{3}}-e_{n}\right)  \, \dd e_{3} \cdots \dd e_{N} \notag \\
      &+ (N-1)\underbrace{\int_{0}^{1-\rho_{2}}  \cdots \int_{0}^{1-\rho_{2}}}_{N-2 \text{ times}}  \frac{\rho_{1}^{2}}{\left(1+\rho_{2}\chi_{3}\right)^{2}}\log\!\Bigg(\frac{\left(1+\rho_{2}\chi_{3}\right)\Big(\frac{\rho_{1}^{2}-\rho_{2}}{1-\left(\rho_{1}^{2}-\rho_{2}\right)\chi_{3}}\Big)+\rho_{2}}{\left(1+\rho_{2}\chi_{3}\right)\left(1-\rho_{2}-e_{n}\right)+\rho_{2}}\Bigg)  \, \dd e_{3} \cdots \dd e_{N}. \label{eq:SDC prob N-2} \raisetag{16pt}
   \end{align}
   Herein we have two decoupled variables, \(\chi_3\) and \(e_n = \max\{e_3, \cdots, e_N\}\), in the sense that for any given \(e_n\) with \(3 \leq n \leq N\), there exists a set \(\{e_i\}_{i=3}^N\) such that \(\chi_3 = \sum_{i=3}^N \frac{1}{1 - \rho_{2} - e_i}\) is satisfied.
   A choice could be \(e_i=1 - \rho_{2} - a_0 q^i\) for each \(i \neq n\), where \(a_0 > 0\) and \(q \geq 1\). Thus, we obtain  
   \(e_{n} = 1 - \rho_{2} - \left(\chi_{3} - \frac{1}{a_{0}} \left( \frac{q^{-2} - q^{-N}}{q - 1} - \frac{1}{q^{n}} \right)\right)^{-1}\).
To satisfy \(e_n = \max\{e_3, \cdots, e_N\}\), i.e., \(e_n + a_0 q^3 \leq 1 - \rho_{2}\), we observe that \(e_n + a_0 q^3\) admits the following upper bound
\( 
   e_n + a_0 q^3 \leq 1-\rho_{2}-\left(\chi_{3}-\frac{1}{a_{0}\left(q-1\right)}\right)^{-1}+a_{0}q 
\). 
Therefore, it suffices to have  
\(q \geq \frac{1}{2} \Big(\frac{a_0^2 \chi_3^2 + 4}{a_0^2 \chi_3^2}\Big)^{\frac{1}{2}} + \frac{a_0 \chi_3 + 2}{2 a_0 \chi_3}\), 
and \(\{e_i\}_{i=3}^N\) is constructed such that for any \(e_{n}\), \(\chi_3 = \sum_{i=3}^N \frac{1}{1 - \rho_{2} - e_i}\) holds. 

We are now ready to handle the \((N-2)\)-fold integral in Eq.~\eqref{eq:SDC prob N-2}. Since \(\chi_3 = \sum_{i=3}^{N} \frac{1}{1 - \rho_{2} - e_i}\) and \(e_n = \max\{e_3, \cdots, e_N\}\) are decoupled, we can separately analyze their impact on \(P(\mathcal{A}_0)\). First, we observe that the integrand in Eq.~\eqref{eq:SDC prob N-2} with respect to \((e_3, \cdots, e_N)\) is monotonically decreasing in \(\chi_3 \geq \frac{N-2}{1 - \rho_{2}}\).
Thus, \(P(\mathcal{A}_0)\) can be upper-bounded as
\begin{align}
   P(\mathcal{A}_0) \leq & \underbrace{\int_{0}^{1-\rho_{2}} \cdots \int_{0}^{1-\rho_{2}}}_{N-2 \text{ times}} M(e_n) \, \dd e_3 \cdots \dd e_N \\  
   ={}& (1-\rho_{2})^{N-2} \mathbb{E}[M(e_n)] \label{eq:SDC prob N-2 Ev1} \\  
   ={}& (N-2)\int_{0}^{1-\rho_{2}} M(e_n) e_{n}^{N-3} \, \dd e_n, \label{eq:SDC prob N-2 Ev2}  
\end{align} 
where the integrand is
\begin{align}
M(e_n) = {}& (N-1)  \left(1-\frac{\left(N-2\right)\rho_{1}^{2}}{1+\left(N-3\right)\rho_{2}}\right)\left(1-\rho_{2}-\frac{\rho_{1}^{2}-\rho_{2}}{1-\left(\rho_{1}^{2}-\rho_{2}\right)\frac{N-2}{1-\rho_{2}}}-e_{n}\right) \notag \\
      &+ (N-1)  \frac{\rho_{1}^{2}}{\left(1+\rho_{2}\frac{N-2}{1-\rho_{2}}\right)^{2}}\log\!\Biggg(\frac{\frac{\rho_{1}^{2}-\rho_{2}}{1-\left(\rho_{1}^{2}-\rho_{2}\right)\frac{N-2}{1-\rho_{2}}}+\frac{\rho_{2}}{1+\rho_{2}\frac{N-2}{1-\rho_{2}}}}{\left(1-\rho_{2}-e_{n}\right)+\frac{\rho_{2}}{1+\rho_{2}\frac{N-2}{1-\rho_{2}}}}\Biggg). \label{eq:SDC prob N-2 upper integrand}
\end{align} 
The step from Eq.~\eqref{eq:SDC prob N-2 Ev1} to Eq.~\eqref{eq:SDC prob N-2 Ev2} follows because \(P(e_n < t) = \prod_{i=3}^N P(e_i < t) = \left( \frac{t}{1-\rho_{2}} \right)^{N-2}\) and the probability density function (PDF) of \(e_n\) is \(p_{e_n} = \frac{(N-2)e_{n}^{N-3}}{(1-\rho_{2})^{N-2}}\).
   The single-variable integral in Eq.~\eqref{eq:SDC prob N-2 Ev2} corresponds to the upper bound in Eq.~\eqref{eq:SDC prob N-2 upper final} with \(M(e_n)\) in Eq.~\eqref{eq:SDC prob N-2 upper integrand}.
Proceeding from Eq.~\eqref{eq:SDC prob N-2 upper final}, we note that
\begin{align}
   &\int_{0}^{1-\rho_{2}}\log\left(1-\rho_{2}-e_{n}+b(1-\rho_{2})\right)e_{n}^{N-3} \, \dd e_{n} \notag \\
   ={}&\left(1-\rho_{2}\right)^{N-2}\Bigg[\frac{\log\left(1-\rho_{2}\right)}{N-2} 
+\int_{0}^{1}\log\left(1-x+b\right)x^{N-3} \, \dd x \Bigg] \label{eq:SDC prob N-2 Harmo} \\
   >{}& \frac{\left(1-\rho_{2}\right)^{N-2}}{N-2}\left(\log\left(1-\rho_{2}\right)-H_{N-2}\right), \label{eq:SDC prob N-2 Harmo v1}
\end{align} 
where \(b = \frac{\rho_{2}}{1 + (N - 3)\rho_{2}} > 0\). Eq.~\eqref{eq:SDC prob N-2 Harmo} is obtained through the variable substitution \(e_{n} = (1 - \rho_{2})x\). 
Eq.~\eqref{eq:SDC prob N-2 Harmo} can be lower-bounded by Eq.~\eqref{eq:SDC prob N-2 Harmo v1} due to the integral representation of the \(N\)-th harmonic number \(H_{N}=\sum_{i=1}^{N}1/i\)~\cite{vualean2019almost}.
From Eq.~\eqref{eq:SDC prob N-2 Harmo v1}, an alternative explicit analytical upper bound for \(P(\mathcal{A}_0)\) involving \(H_{N}\) is given by
   \begin{align}
   P(\mathcal{A}_0)<&\left(1-\rho_{2}\right)^{N-2}\Bigg[\left(1-\frac{\left(N-2\right)\rho_{1}^{2}}{1+\left(N-3\right)\rho_{2}}\right)\left(1-\rho_{2}-\frac{(N-1)\left(\rho_{1}^{2}-\rho_{2}\right)}{1-\left(\rho_{1}^{2}-\rho_{2}\right)\frac{N-2}{1-\rho_{2}}}\right) +\frac{\left(N-1\right)\rho_{1}^{2}}{\left(1+\rho_{2}\frac{N-2}{1-\rho_{2}}\right)^{2}} \notag \\
   & \cdot \left(\log\!\left(\frac{\rho_{1}^{2}-\rho_{2}}{1-\left(\rho_{1}^{2}-\rho_{2}\right)\frac{N-2}{1-\rho_{2}}}+\frac{\rho_{2}\left(1-\rho_{2}\right)}{1+\left(N-3\right)\rho_{2}}\right)+H_{N-2}-\log\left(1-\rho_{2}\right)\right)\Bigg]. \label{eq:SDC prob N-2 upper final Harmo}
   \end{align}
The harmonic number is asymptotically given by \(H_N = \log N + \gamma + O(1/N)\), where \(\gamma\) is the Euler-Mascheroni constant~\cite{havil2009gamma}. 
Thus, \(P(\mathcal{A}_0)\) can be asymptotically approximated in Eq.~\eqref{eq:SDC prob N-2 upper final Harmo asym}.

\section{
   Proof of Theorem~\ref{thm:max rho0} \label{proof:thm:max rho0}
}

We define two functions \(f(\rho_{2}) = \frac{\rho_{2}}{\rho_{2} + e_2 - 1}\) and \(g(\rho_{2}) = 1 + \sum_{i=3}^{N} \frac{\rho_{2}}{1 - \rho_{2} - e_i}\). The function \(f(\rho_{2})\) is strictly decreasing, while \(g(\rho_{2})\) is piecewise increasing in the intervals \((1 - e_i, 1 - e_{i+1})\) for \(2 \leq i \leq N-1\).  
We note that \(\rho_{2}^m\) is the solution to \(f(\rho_{2}) = g(\rho_{2})\) when \(\rho_{2} \in (1 - e_2, 1 - e_3)\).

To establish a lower bound, we define a strictly increasing function \(g_{1}(\rho_{2}) = 1 + \frac{(N-2)\rho_{2}}{1 - \rho_{2} - e_3} \geq g(\rho_{2})\) for \(\rho_{2} \in (0, 1 - e_3)\), with equality if and only if \(e_i = e_3\) for all \(i > 3\).
Due to the continuity of \(f(\rho_{2})\), there exists a lower bound \(\rho_{2}^{\ell} \in (1 - e_2, 1 - e_3)\) for \(\rho_{2}^m\) such that \(f(\rho_{2}^{\ell}) \geq f(\rho_{2}^m)\). Since \(\rho_{2}^{\ell}\) is the unique positive root of the equation \(g_{1}(\rho_{2}) = f(\rho_{2})\), we have
\begin{align}
\rho_{2}^{\ell} = (1 - e_2) \frac{(N-2) + \sqrt{(N-2)^2 + 4(N-1)\frac{1-e_3}{1-e_2}}}{2(N-1)}.
\end{align}
Through further derivation and simplifications, \(\rho_{2}^{\ell}\) can be more concisely lower-bounded by the left-hand side of Eq.~\eqref{eq:max rho0 bounds}, with equality if and only if \(N = 2\).

For the upper bound of \(\rho_{2}\), we first derive a lower bound for \(g(\rho_{2})\) using Jensen's inequality, which is given by 
\(g_2(\rho_{2}) = 1 + (N-2)\rho_{2}\Big(1 - \rho_{2} - \frac{\sum_{i=3}^{N} e_i}{N-2}\Big)^{-1}\) 
for \(\rho_{2} \in (0, 1 - e_3)\). To determine \(\rho_{2}^u\), we solve the equation \(g_2(\rho_{2}) = f(\rho_{2})\). This yields
\begin{align}
\rho_{2}^u = (1 - e_2) \frac{(N-3) + \sqrt{(N-3)^2 + 4(N-2)\frac{1 - \overline{e}_3}{1 - e_2}}}{2(N-2)},
\end{align}
where \(1 - e_{2} \leq \rho_{2}^{m} \leq \rho_{2}^{u} \leq 1 - e_{3}\) and \(\overline{e}_3 = \frac{1}{N-2}\sum_{i=3}^{N} e_i\). Furthermore, since 
\(N-1 \leq \sqrt{(N-3)^2 + 4(N-2)\frac{1 - \overline{e}_3}{1 - e_2}}\) 
holds for \(N \geq 2\), we have a more streamlined upper bound in the right-hand side of Eq.~\eqref{eq:max rho0 bounds}.
It is straightforward to derive the asymptotic approximation in Eq.~\eqref{eq:rho0m theta} from the lower and upper bounds of \(\rho_{2}^m\) in Eq.~\eqref{eq:max rho0 bounds}.

\section{
   Proof of Theorem~\ref{thm:dsm model} \label{proof:thm:dsm model}
}

Fix \(N\ge 3\) and consider a 2-TC covariance matrix \(\Rm\in\mathcal K_2\) under 2-TD distortion constraints, i.e.,
\(\ev=[e_1,e_2,\cdots,e_2]^{\T}\) of length \(N\).
To establish the structure of the optimal distortion matrix \(\Dsm\), we start from an arbitrary distortion matrix \(\Dm\) in the feasible set of the optimization problem in Lemma~\ref{lemma:maxdet}.
Without loss of generality, we describe \(\Dm\) entry-wise as
\begin{align}
\Dm_{ij} \triangleq
\begin{cases}
d_i, & i=j, \\[-6pt]
\alpha_{1j}, & i=1,\ j\ge 2, \\[-6pt]
\alpha_{1i}, & j=1,\ i\ge 2, \\[-6pt]
\beta_{ij}, & 2\le i<j, \\[-6pt]
\beta_{ji}, & 2\le j<i,
\end{cases}
\label{eq:general_D}
\end{align}
with \(d_i\in (0,e_i]\).
To facilitate the proof, we first introduce permutation matrices.
A matrix \(\Pim\in\{0,1\}^{N\times N}\) is a permutation matrix if each row and column contains exactly one entry equal to \(1\).
Equivalently, there exists a permutation \(\pi\) of \(\{1,\cdots,N\}\) such that \(\Pim_{i,\pi(i)}=1\) for \(i\in[N]\) and \(\Pim_{ij}=0\) otherwise.
Every permutation matrix is orthogonal, i.e., \(\Pim^{\T}\Pim=\Pim\Pim^{\T}=\mathsf{I}\) and \(\Pim^{-1}=\Pim^{\T}\).

To capture the exchangeability of the \(N-1\) peripheral components, let \(\mathsf{b}_1=[1,0,\cdots,0]^{\T}\) and define \(\mathcal{P}\triangleq\bigl\{\Pim\in\{0,1\}^{N\times N}:\Pim \text{ is a permutation matrix and } \Pim \mathsf{b}_1=\mathsf{b}_1\bigr\}\).
That is, for any vector \(\mathsf{x}\in\mathbb{R}^N\) and any \(\mathsf{P}\in\mathcal{P}\), we have \([\mathsf{P}\mathsf{x}]_1 = x_1\), while \([\mathsf{P}\mathsf{x}]_i = \mathsf{x}_{\pi(i)}\) for \(i\ge 2\), where \(\pi\) is a permutation of \(\{2,\cdots,N\}\).
For any permutation matrix \(\mathsf{P}\in\mathcal{P}\), define
\begin{align}
\Dm_{\mathsf{P}} \triangleq \mathsf{P} \Dm \mathsf{P}^{\T}.
\end{align}
This operation corresponds to simultaneously permuting the rows and columns of \(\Dm\) according to \(\mathsf{P}\), i.e., relabeling the peripheral components while keeping the central component fixed.
We next show that \(\Dm_{\mathsf{P}}\) remains feasible for all \(\mathsf{P}\in\mathcal{P}\), i.e., it satisfies the constraints in Eqs.~\eqref{eq:maxdet c1}-\eqref{eq:maxdet c2}.
Since \(\mathsf{P}\) is an orthogonal matrix, we have \(\Dm\succ\zerov\) if and only if \(\Dm_{\mathsf{P}}\succ\zerov\).
Moreover, the constraint \(\Dm\preceq\Rm\) implies \(\Dm_{\mathsf{P}} \preceq \mathsf{P} \Rm \mathsf{P}^{\T}\), because congruence transformations by an invertible matrix preserve positive semidefiniteness.
Under the 2-TC covariance model, \(\Rm\in\mathcal{K}_2\) is invariant under any permutation of the peripheral indices.
That is, for all \(\mathsf{P}\in\mathcal{P}\), we have \(\mathsf{P} \Rm \mathsf{P}^{\T} = \Rm\).
Consequently, since \(\Dm \preceq \Rm\), it follows that \(\Dm_{\mathsf{P}} \preceq \Rm\).
Moreover, recall that the distortion constraint is given by \(\Dm_{11}\le e_1\) and \(\max_{i\ge 2}\Dm_{ii}\le e_2\).
Since any permutation matrix \(\mathsf{P}\in\mathcal{P}\) only permutes the remaining \(N-1\) peripheral coordinates, the diagonal entries of \(\Dm_{\mathsf{P}}\) satisfy the same individual distortion constraints.
Next, we consider the objective value in Eq.~\eqref{eq:maxdet obj} under permutation.
We have \(\det(\Dm_{\mathsf{P}}) = \det(\mathsf{P})^2 \det(\Dm) = \det(\Dm)\) since any permutation matrix \(\mathsf{P}\) is orthogonal with \(\det(\mathsf{P})=\pm 1\).
Now define the averaged distortion matrix
\begin{align}
\overline{\Dm}
\triangleq \frac{1}{|\mathcal{P}|}\sum_{\mathsf{P}\in\mathcal{P}} \Dm_{\mathsf{P}}.
\label{eq:D_bar_def}
\end{align}
Moreover, the averaged matrix in Eq.~\eqref{eq:D_bar_def} is invariant under permutation \(\mathsf{P}\in\mathcal{P}\), i.e., \(\mathsf{P}\overline{\Dm}\mathsf{P}^{\T}=\overline{\Dm}\).
Specifically, for any \(\mathsf{P}_0\in\mathcal{P}\), we have
\begin{align}
\mathsf{P}_0\overline{\Dm}\mathsf{P}_0^{\T}
&= \frac{1}{|\mathcal{P}|}\sum_{\mathsf{P}\in\mathcal{P}}
\mathsf{P}_0 \mathsf{P}\Dm \mathsf{P}^{\T}\mathsf{P}_0^{\T} \\
&= \frac{1}{|\mathcal{P}|}\sum_{\mathsf{P}'\in\mathcal{P}}
\mathsf{P}'\Dm (\mathsf{P}')^{\T} \\
&= \overline{\Dm},
\end{align}
where we used the fact that \(\mathcal{P}\) is a group under matrix multiplication, i.e., \(\{\mathsf{P}_0\mathsf{P}:\mathsf{P}\in\mathcal{P}\}=\mathcal{P}\).
From the feasibility of each \(\Dm_{\mathsf{P}}\) and the convexity of the constraint set,
it follows that \(\overline{\Dm}\) is also feasible.
More importantly, since \(\log\det(\cdot)\) is concave over the PD cone,
Jensen's inequality yields
\begin{align}
\log\det(\overline{\Dm})
&\ge \frac{1}{|\mathcal{P}|}\sum_{\mathsf{P}\in\mathcal{P}} \log\det(\Dm_{\mathsf{P}}) \\
&= \log\det(\Dm),
\end{align}
where the equality holds if and only if \(\Dm_{\mathsf{P}}=\Dm\) for all \(\mathsf{P}\in\mathcal{P}\).
Consequently, \(\overline{\Dm}\) achieves a strictly larger \(\log\det(\cdot)\) value than \(\Dm\) unless \(\Dm\) is already \(\mathcal{P}\)-invariant.

We now prove that an optimal solution must be \(\mathcal{P}\)-invariant.
Let \(\Dsm\) be an optimal solution of the problem in Lemma~\ref{lemma:maxdet}, and consider its group average \(\overline{\Dsm}\triangleq \frac{1}{|\mathcal{P}|}\sum_{\mathsf{P}\in\mathcal{P}} (\Dsm)_{\mathsf{P}}\). 
As shown above, \(\overline{\Dsm}\) is feasible.
Moreover, by the strict concavity of the log-determinant function, if \(\Dsm\) is not invariant under \(\mathcal{P}\), then \(\log\det(\overline{\Dsm})>\log\det(\Dsm)\), which contradicts the optimality of \(\Dsm\).
Therefore, there exists an optimal distortion matrix satisfying \(\mathsf{P}\Dsm\mathsf{P}^{\T}=\Dsm\) for all \(\mathsf{P}\in\mathcal{P}\), which directly implies that \(\Dsm\) admits the block structure in Eq.~\eqref{eq:dsm model}.




\section{
   Proof of Theorem~\ref{thm:seven regions main} \label{proof:thm:seven regions main}
}

For the optimization problem in Eq.~\eqref{problem:cvx2}, we construct the Lagrangian
\begin{align}
   \mathcal{L} ={}& \log\!\left[(\delta_2 - \beta)^{N-2} \Xi\right]
+ \widetilde{\lambda}_1 (e_1 - \delta_1)
+ \widetilde{\lambda}_2 (e_{2} - \delta_2) 
+ \widetilde{\lambda}_3 [(1 - \rho_{2}) - (\delta_2 - \beta)]
+ \widetilde{\lambda}_4 \Upsilon, \label{eq:problem:cvx2 Lagrange}
\end{align}
where the Lagrange multipliers \(\widetilde{\lambda}_1, \widetilde{\lambda}_2, \widetilde{\lambda}_3, \widetilde{\lambda}_4 \geq 0\).
The stationarity conditions are obtained by differentiating the Lagrangian with respect to each variable and setting the derivatives to zero.
By substituting \(\lambda_i = \Xi \widetilde{\lambda}_i\) with \(\Xi > 0\), we obtain the simplified forms of the Lagrange multipliers,
\begin{align}
   \lambda_1 &= \frac{\partial \Xi}{\partial \delta_1} - \frac{\partial \Upsilon}{\partial \delta_1} \frac{\frac{\partial \Xi}{\partial \alpha}}{\frac{\partial \Upsilon}{\partial \alpha}}, \label{eq:lambda1} \\
   \lambda_2 &= \frac{\partial \Xi}{\partial \beta} + \frac{\partial \Xi}{\partial \delta_2} -\left(\frac{\partial \Upsilon}{\partial \beta}+\frac{\partial \Upsilon}{\partial \delta_2}\right)\frac{ \frac{\partial \Xi}{\partial \alpha}}{\frac{\partial \Upsilon}{\partial \alpha}}, \label{eq:lambda2} \\
   \lambda_3 &= \frac{\partial \Upsilon}{\partial \beta}\frac{ \frac{\partial \Xi}{\partial \alpha}}{\frac{\partial \Upsilon}{\partial \alpha}} + (N-2) \frac{\Xi}{\delta_2 - \beta} - \frac{\partial \Xi}{\partial \beta}, \label{eq:lambda3} \\
   \lambda_4 &= -\frac{\frac{\partial \Xi}{\partial \alpha}}{\frac{\partial \Upsilon}{\partial \alpha}}. \label{eq:lambda4}
\end{align}
The following proposition provides the conditions that these multipliers should satisfy.
\begin{proposition} \label{prop:lambda24 nonzero simul}
In the non-SDC region, 
the Lagrange multipliers satisfy \(\lambda_i^2 + \lambda_j^2 > 0\) 
for every \(i \in \{2,4\}\) and any \(1 \leq j \leq 4\) with \(j \neq i\).
\end{proposition}
\begin{IEEEproof}
   We separately analyze the cases \(\lambda_4 = 0\) and \(\lambda_2 = 0\). We then prove that the other multipliers cannot be zero simultaneously with one of these two multipliers.
When \(\lambda_4 = 0\), i.e., \(\frac{\partial \Xi}{\partial \alpha} = 0\), it follows that \(\alpha = 0\). Under this condition,
\begin{enumerate}
   \item For \(\lambda_1\): Since
\(
   \lambda_1 = \frac{\partial \Xi}{\partial \delta_1} = \frac{\Xi}{\delta_1}
\)
   and \(\Dsm \succ \zerov\) hold, it follows that \(\lambda_1 > 0\).
   \item For \(\lambda_2\): Since \(\delta_1 > 0\) holds,
   it follows that
\(
   \lambda_2 = \frac{\partial \Xi}{\partial \beta} + \frac{\partial \Xi}{\partial \delta_2} = (N-1) \delta_1 > 0
\).
   \item For \(\lambda_3\): We have
\begin{align}
\!\!\!\!\!\! \lambda_3 &= (N-2) \frac{\Xi}{\delta_2 - \beta} - \frac{\partial \Xi}{\partial \beta} = (N-2)\delta_1 \frac{(N-1)\beta}{\delta_2 - \beta}.
\end{align}
If \(\lambda_3 = 0\) holds, it implies \(\beta = 0\). Combined with \(\lambda_1 > 0\) and \(\lambda_2 > 0\), this corresponds to \(\Dsm = \Em\), which is the SDC region.
\end{enumerate}
When \(\lambda_2 = 0\), it implies
\(
\alpha = \delta_{1} \rho_{1}
\).
Additionally, \(\lambda_4 > 0\) leads to \(\Upsilon = 0\), which results in
\begin{align} \label{eq:upsilon sub}
   \left(N - 2\right) \beta + \delta_2 = \left(N - 2\right)\rho_{2} + \left(N - 1\right)\rho_{1}^{2} \left(\delta_{1} - 1\right) + 1.
\end{align}
\begin{enumerate}
   \item For \(\lambda_1\): Since
\(
\lambda_1 = \frac{(N-2) \delta_{1} \rho_{2} + \delta_{1} - (N-2) \beta - \delta_2}{\delta_{1} - 1}
\) holds,
and we invoke Eq.~\eqref{eq:upsilon sub} to eliminate \(\beta\), it follows that
\(
   \lambda_1 = \nu > 0
\) with \(\nu\) in Eq.~\eqref{eq:nu def}.

   \item For \(\lambda_3\): Since
\(
   \lambda_3 = (N-2)\delta_{1} \frac{(N-1)\delta_{1} \rho_{1}^{2} - (N-2) \beta - \delta_2}{\beta - \delta_2}
\) holds,
\(\lambda_3 > 0\) follows after eliminating \(\beta\).
\end{enumerate}
By combining these two cases, we complete the proof of Proposition~\ref{prop:lambda24 nonzero simul}.
\end{IEEEproof}
Accordingly, the configurations of the remaining Lagrange multiplier signs correspond to the other six non-SDC regions.\footnote{Based on the results of this paper, a Lagrange multiplier being zero is a necessary and sufficient condition for the corresponding constraint to be inactive. However, this equivalence does not generally hold mathematically within the complementary slackness condition.}
\begin{itemize}
   \item \(\mathcal{R}_{1}\):
   All constraints are active, that is, \(\prod_{i=1}^{4} \lambda_i > 0\).
   \item \(\mathcal{R}_{2}-\mathcal{R}_{5}\):
   Exactly one constraint is inactive, that is, for some \(i \in [4]\), \(\lambda_i = 0\) and \(\prod_{j \neq i} \lambda_j > 0\).
   \item \(\mathcal{R}_{6}\):
   Two constraints are inactive, that is, \(\lambda_1 = 0\) and \(\lambda_3 = 0\) while \(\lambda_2 > 0\) and \(\lambda_4 > 0\).
\end{itemize}

Next, we group the seven distortion regions into four parts to derive the optimal distortion allocations and corresponding region boundaries. Once the optimal distortion allocations are determined, Lemma~\ref{lemma:determinant} can be used to obtain the exact closed-form RDF.

\textit{1) Region Zero:} Under the 2-TC covariance matrix and 2-TD constraints, the SDC region is \(\mathcal{R}_0\). Theorem~\ref{thm:RE_psd_conditions} yields the RDF in Eq.~\eqref{eq:sdc R0 rdf} and the optimality condition in Eq.~\eqref{eq:sdc inequality v1}.

\textit{2) Region One:}
Under the complementary slackness conditions, all constraints are active in \(\mathcal{R}_1\), and this leads to \(\Dsm\).
Substituting the parameters of \(\Dsm\) into Eqs.~\eqref{eq:lambda1}-\eqref{eq:lambda4} yields the following concrete multipliers:
\begin{align}
   &\lambda_1 = (N-2)\rho_{2} + 1 - (N-1) \rho_{1}\sqrt{\frac{1-e_{2}}{1-e_{1}}}, \\
   &\lambda_2 = (N - 1)\left[e_{1} - \sqrt{\frac{1-e_{1}}{1-e_{2}}}\left(\rho_{1} - \sqrt{(1 - e_{1})(1 - e_2)}\right)\right], \\
   &\lambda_3 = (N - 2)\left(e_{1} + \frac{N - 1}{1-\rho_{2}} \omega + \rho_{1} \sqrt{\frac{1-e_{1}}{1-e_{2}}} - 1\right), \\
   &\lambda_4 = \frac{\rho_{1}}{\sqrt{(1 - e_{1})(1 - e_2)}} - 1.
\end{align}
Then, to ensure optimality, i.e., \(\prod_{i=1}^{4}\lambda_i > 0\), the four boundaries can be sequentially derived.
Since \(\mathcal{R}_{1}\) is part of the non-SDC region, the boundaries of \(\mathcal{R}_{0}\) should also be considered. Given the ranges of the boundaries in \(\mathcal{R}_1\), it suffices to analyze the relationship between \(\mathcal{R}_{0,5}^{(3), \mathrm{B}}\) and \(\mathcal{R}_{1,5}^{(4), \mathrm{T}}\), as well as \(\mathcal{R}_{0,4}^{(4), \mathrm{B}}\) and \(\mathcal{R}_{1,4}^{(3), \mathrm{T}}\).
To the left of point \(P_1\), since \(\mathcal{B}_{0,5}^{(3)}\) is a straight line, it only intersects \(\mathcal{B}_{1,5}^{(4)}\) at \(P_1\) if the condition in Eq.~\eqref{eq:sqrt rho0 cond} holds; otherwise, neither of these boundaries exists on the plane \(\mathcal{U}\), so that \(\mathcal{R}_{0,5}^{(3), \mathrm{B}} \cap \mathcal{R}_{1,5}^{(4), \mathrm{T}} = \emptyset\).
To the right of point \(P_1\), it is straightforward to conclude \(\mathcal{R}_{0,4}^{(4), \mathrm{B}} \cap \mathcal{R}_{1,4}^{(3), \mathrm{T}} = \emptyset\) by analyzing the positions of \(\mathcal{B}_{0,4}^{(4)}\) and \(\mathcal{B}_{1,4}^{(3)}\).
Therefore, all boundaries of \(\mathcal{R}_1\) are exactly the four boundaries determined by the corresponding multipliers.

\textit{3) Regions Two, Three, Five, and Six:} 
We adopt procedures similar to those outlined in \(\mathcal{R}_{1}\). In each region, based on the distinct sign combinations of the multipliers, we first determine \(\Dsm\) using the zero multipliers and the constraint equalities corresponding to the positive multipliers. These parameters are then substituted into the positive multipliers and the constraints corresponding to the zero multipliers to determine the region boundaries.

\textit{4) Region Four:}
When \(\lambda_1\lambda_2 > 0\), we have \(\delta_1 = e_1\) and \(\delta_2 = e_{2}\). From \(\lambda_3 = 0\), we obtain
\begin{align} \label{eq:alpha lambda3}
   \beta = g_{1}(\alpha),
\end{align}
where
\begin{align} \label{eq:g1beta}
   g_{1}(\alpha) = \alpha\frac{\left(N-1\right)\alpha\left(\alpha-\rho_{1}\right)+e_{2}\left(1-e_{1}\right)}{\alpha\left[\left(N-2\right)e_{1}+1\right]-\left(N-1\right)e_{1}\rho_{1}}.
\end{align}
Similarly, since \(\lambda_4 > 0\) implies that \(\Upsilon = 0\), we obtain
\begin{align} \label{eq:alpha upsilon}
   \beta = g_{2}(\alpha),
\end{align}
where
\begin{align} \label{eq:g2beta}
   g_{2}(\alpha) = \frac{\left(N-2\right)\rho_{2}+1-e_{2}}{N-2}-\frac{\left(N-1\right)\left(\alpha-\rho_{1}\right)^{2}}{\left(N-2\right)\left(1-e_{1}\right)}.
\end{align}
Combining Eq.~\eqref{eq:alpha lambda3} and Eq.~\eqref{eq:alpha upsilon} defines a cubic function of \(\alpha\), namely \(f(\alpha)\) in Eq.~\eqref{eq:region4 cubic func}.
A cubic equation inherently has at least one real root, ensuring the existence of \(\alpha_{\star}\) in \(\mathcal{R}_{4}\), but the possibility of multiple roots introduces ambiguity.

Without loss of generality, we consider the case where \( f(\alpha) \) has three real roots ordered as \( \alpha_1 \leq \alpha_2 \leq \alpha_3 \). Thus, the optimal \(\alpha_\star\) satisfies
\(
\alpha_\star \in \{\alpha_i \in \mathbb{R} : f(\alpha_i) = 0, \ i = 1, 2, 3\}
\)
and we have
\begin{align} \label{eq:beta star range}
   \max\{\alpha^{(1)}, \alpha^{(3)}, \alpha^{(4)}\} < \alpha_\star < \alpha^{(2)},
\end{align}
where
\(\alpha^{(1)} = \rho_{1} - \frac{\left(\left(N-2\right)\rho_{2}+1\right)\left(1-e_{1}\right)}{\left(N-1\right)\rho_{1}}\),
\(\alpha^{(2)} = e_{1}\rho_{1}\),
\(\alpha^{(3)} = \rho_{1} - \sqrt{\left(1-e_{2}\right) \left(1-e_{1}\right)}\)
and
\(\alpha^{(4)} = 0\). 
\(\alpha^{(i)}\) is obtained by solving \(\lambda_i > 0\) and the active constraint associated with \(\lambda_i\).
We denote the local minimum point of \(f(\alpha)\) as \(\alpha_{\mathrm{min}}\), satisfying
\begin{align} \label{eq:beta min geq e1rho1}
   \alpha_{\mathrm{min}} > \frac{t}{3\left(N-1\right)} > \alpha^{(2)},
\end{align}
with \(t\) in Eq.~\eqref{eq:region4 cubic param}.
Furthermore, we find
\begin{align}
   f(\alpha = \alpha^{(2)}) = e_{1} \rho_{1} \left(N - 2\right) \left(1-e_{1}\right)^{2} \nu \geq 0. \label{eq:beta e1rho1 geq0}
\end{align}
We conclude that the optimal solution is \(\alpha_\star = \alpha_1 \in \mathbb{R}\), which can be established by contradiction.
When \(\alpha_2, \alpha_3 \in \mathbb{R}\) and suppose \(\alpha_\star = \alpha_2\), we have \(\alpha^{(2)} > \alpha_2\) in light of Eq.~\eqref{eq:beta star range}. However, since Eq.~\eqref{eq:beta e1rho1 geq0} holds, it follows that \(\alpha^{(2)} > \alpha_3\), as \(f(\alpha) < 0\) for \(\alpha \in (\alpha_2, \alpha_3)\). This would contradict Eq.~\eqref{eq:beta min geq e1rho1} due to \(\alpha_{\mathrm{min}} < \alpha_3\).
Likewise, we have \(\alpha_\star \neq \alpha_3\), which can be deduced using the same reasoning.
Therefore, the ambiguity associated with multiple roots is eliminated. The exact \(\alpha_\star\) in Eq.~\eqref{eq:alpha star analytical} and \(\beta_\star\) directly yield the closed-form RDF.
Given the range of \(\alpha_\star\) in Eq.~\eqref{eq:beta star range}, the conditions \(f(\alpha = \alpha^{(i)}) < 0\) for \(i = 1, 3, 4\) sequentially lead to the optimality conditions \(\mathcal{B}^{(i)}\) for \(i = 1, 3, 4\) in \(\mathcal{R}^{(4)}\). It is straightforward to verify that \(\mathcal{B}_{0,5}^{(3)}\) is not one of the boundaries of \(\mathcal{R}^{(4)}\).

\section{
   Proof of Proposition~\ref{prop:region one B3 e0star} \label{proof:prop:region one B3 e0star}
}

To determine \(e_{2}^\star\), we treat \(\mathcal{B}_{1,4}^{(3)}\) as a function of \(e_1\) and apply an appropriate variable substitution. Using \(e_1 = 1 - t^2\), we transform \(\mathcal{B}_{1,4}^{(3)}\) into a quadratic equation \(q(t)\) defined for \(t \in [-1,1]\).
The discriminant of \(q(t)\) is given by \(\frac{4}{\left(1-\rho_{2}\right)^{2}}\pi(e_{2})\), where \(\pi(e_{2})\) is given by
\begin{align}
   \pi(e_{2}) = {} & \psi + \left(e_{2} + \rho_{2} - 1\right) 
\cdot \left[N^{2} \left(\rho_{2} - \rho_{1}^{2}\right) - N \left(3 \rho_{2} - 2 \rho_{1}^{2} - 1\right)\right],
\end{align}
and
\begin{align}
   \psi ={}& \frac{\left(1-\rho_{2}\right)^{2}\left(2e_{2}+\rho_{2}-2\right)^{2}}{4(1-e_{2})} 
\cdot \left[1-\frac{\rho_{2}^{2}-\rho_{1}^{2}}{\left(1-\rho_{2}\right)^{2}}\left(1 + \frac{1-2\rho_{2}}{\left(2e_{2}+\rho_{2}-2\right)^{2}}\right)\right].
\end{align}
For the unbounded \(\mathcal{B}_{1,4}^{(3)}\), 
the extreme value \(\tilde{e}_2^\star \in (-\infty, 1]\) satisfies \(\pi(\tilde{e}_2^\star) = 0\). Solving this yields
\begin{align} \label{eq:e0star tilde analytical}
   \tilde{e}_2^\star = 1 - \rho_{2} - \frac{\mu - \rho_{2} \sqrt{(N-1)\nu}}{2\sqrt{(N-1)\nu}}
\end{align}
with \(\mu = \sqrt{\left[(N-2)\rho_{2} + 1\right]\left[\rho_{2} \nu - \left(\rho_{2} - \rho_1^2\right)(1-\rho_{2})\right]}\) and \(\nu\) in Eq.~\eqref{eq:nu def}. For a given \(N\), \(\rho_{2}\), and \(\rho_1\), \(\tilde{e}_1^\star\) is obtained by mapping \(\tilde{e}_2^\star\) through the unbounded \(\mathcal{B}_{1,4}^{(3)}\),
\begin{align}
   \tilde{e}_1^\star = 1 - \frac{\rho_{1}^{2}}{1 - \tilde{e}_2^\star} \left(\frac{2\left(N-1\right)\left(1 - \tilde{e}_2^\star\right) + \left(1 - \rho_{2}\right)}{2\left(\left(N-2\right)\rho_{2} + 1\right)}\right)^{2}.
\end{align}
The condition \(\tilde{e}_1^\star \geq 0\) is equivalent to the condition in Eq.~\eqref{eq:e0star analytical}, and then
\(
e_{2}^\star = \tilde{e}_2^\star
\) is determined in Eq.~\eqref{eq:e0star tilde analytical}.
Otherwise, when \(\tilde{e}_1^\star < 0\) holds, the point \((\tilde{e}_1^\star, \tilde{e}_2^\star)\) does not lie on the bounded segment \(\mathcal{B}_{1,4}^{(3)}\), and it reduces to
\(
e_{2}^\star = 1 - \rho_1^2
\)
due to point \(P_1\). It should be noted that Eq.~\eqref{eq:e0star analytical otherwise} is obtained under the condition \(\rho_{1}^{2} < \rho_{2}+\frac{1-\rho_{2}}{N-1}\) in Eq.~\eqref{eq:2tc pd cond}.

To analyze the convergence of \(e_{2}^\star\) to \(1-\rho_{2}\), we rewrite Eq.~\eqref{eq:e0star analytical} as
\begin{align}
   1 - e_{2}^\star - \rho_{2} = \frac{\rho_1^2 (1-\rho_{2})^2}{2\sqrt{(N-1)\nu}\left(\rho_{2} \sqrt{(N-1)\nu} + \mu\right)}.
\end{align}
Given the orders of \(\mu\) and \(\nu\), the asymptotics of the denominator lead to Eq.~\eqref{eq:e0star asymptotic}. 
Then, based on the condition in Eq.~\eqref{eq:e0star analytical otherwise} (cf. Eq.~\eqref{eq:2tc pd cond}) and replacing \(\rho_1^{2}\) with \(1 - e_{2}^\star\), we obtain Eq.~\eqref{eq:e0star asymptotic otherwise}.

\section{
   Proof of Proposition~\ref{prop:Dsm bound} \label{proof:prop:Dsm bound}
}

Based on the analysis of \(\alpha_\star\) in \(\mathcal{R}_4\) as detailed in Appendix~\ref{proof:thm:seven regions main}, herein we provide bounds for \(\alpha_\star\), \(\beta_\star\) and the RDF.
All roots of \(g_1(\alpha)=0\) in Eq.~\eqref{eq:g1beta} are given as
\(\alpha_{g1,1} = 0\), 
\(\alpha_{g1,2} = \frac{\rho_{1}}{2}-\sqrt{\frac{\rho_{1}^{2}}{4}-\frac{e_{2}(1-e_{1})}{N-1}}\) 
and
\(\alpha_{g1,3} = \frac{\rho_{1}}{2}+\sqrt{\frac{\rho_{1}^{2}}{4}-\frac{e_{2}(1-e_{1})}{N-1}}\),
where \(\alpha_{g1,2}\) and \(\alpha_{g1,3}\) are real if
\begin{align} \label{eq:g1beta real}
   e_{2} \leq \frac{(N-1)\rho_{1}^{2}}{4(1-e_{1})}.
\end{align}
If Eq.~\eqref{eq:g1beta real} holds, \(g_1(\alpha) < 0\) for \(\alpha \in (\alpha_{g1,1}, \alpha_{g1,2})\); otherwise, \(g_1(\alpha)\) has only one root \(\alpha_{g1,1}\), and satisfies \(g_1(\alpha) < 0\) for \(\alpha \in (\alpha_{g1,1}, \infty)\). Furthermore, the following inequality
\begin{align} \label{g1g2 root relation}
   \alpha_{g2,1} \leq \alpha_{g1,2}
\end{align}
holds, where
\(\alpha_{g2,1} = \rho_{1} - \sqrt{\frac{\left(1-e_{1}\right) \left[(N-2) \rho_{2} + 1 - e_{2}\right]}{N - 1}}\) 
is the smaller real root of \(g_2(\alpha)=0\) in Eq.~\eqref{eq:g2beta}.
Eq.~\eqref{g1g2 root relation} holds because \(\alpha_{g1,2} - \alpha_{g2,1}\) is a decreasing function of \(e_1\) and reaches equality at the boundary \(\mathcal{B}_{4,6}^{(1)}\).
Thus, we have \(g_2(\alpha_{g1,1}) < g_1(\alpha_{g1,1}) = 0\) and \(g_1(\alpha_{g2,1}) < g_2(\alpha_{g2,1}) = 0\). Since \(g_1(\alpha)\) and \(g_2(\alpha)\) are continuous, it follows that \(\alpha_\star \in (0, \alpha_{g2,1})\) and \(\beta_\star < 0\). Let \(\alpha^{u}_{\star} = \alpha_{g2,1}\) in Eq.~\eqref{eq:beta star upper}. Thus, the bounds for \(\alpha_\star\) and \(\beta_\star\) are provided in Eqs.~\eqref{eq:beta star bounds} and \eqref{eq:alpha star bounds}.
Since \(\lambda_1, \lambda_2, \lambda_4 > 0\) hold, all primal feasibility conditions are satisfied, and Eq.~\eqref{eq:alpha lambda3} is violated. Therefore, the solution \(\delta_1 = e_1\), \(\alpha = \alpha_{g2,1}\), \(\delta_2 = e_{2}\), \(\beta = 0\) remains feasible but is strictly suboptimal in \(\mathcal{R}_{4}\). Furthermore, the upper bound of the RDF is given by Eq.~\eqref{eq:region4 rdf upper} when Eq.~\eqref{eq:D tilde sdc} holds.

\section{
   Proof of Theorem~\ref{thm:iso tighter} \label{proof:thm:iso tighter}
}

Due to the objective in Eq.~\eqref{eq:maxdet obj} and disregarding the constant \(\det(\Rm)\), Eq.~\eqref{eq:iso tighter} is equivalent to
\begin{align} \label{eq:Dsm det upper}
   [\det (\Dsm)]^{2} \leq [\det (\widetilde{\Dm})][\det (\Em)],
\end{align}
where \(\Dsm\), \(\widetilde{\Dm}\), and \(\Em\) are the distortion matrices corresponding to the rates \(\Rx([N], \ev)\), \(R^{u}_{\xv}([N], \ev)\), and \(R^{\ell}_{\xv}([N], \ev)\), respectively.
From Eq.~\eqref{eq:dsm det}, we have
\(
\det (\Dsm) = (e_{2} - \beta_\star)^{N-2}\big[ e_1 (e_{2}  + (N-2)\beta_\star) \allowbreak - (N-1)\alpha_\star^2 \big]
\)
where \(\beta_\star = g_2(\alpha_\star)\) with \(g_2(\alpha)\) in Eq.~\eqref{eq:g2beta}.
From Proposition~\ref{prop:Dsm bound}, we have
\(
   \det (\widetilde{\Dm}) = e_{2}^{N-2} \left[ e_1 e_{2} \allowbreak - (N-1)(\alpha_{g2,1}^{\iso})^2 \right]
\),
where \(\alpha_{g2,1}^{\iso}\) is the isotropic correlation version of \(\alpha_{\star}^{u}\) in Eq.~\eqref{eq:beta star upper}.
We will prove that in \(\mathcal{R}_{4}\), any \(\Dm\) in Eq.~\eqref{eq:dsm model} with \(\delta_1 = e_1\), \(\delta_2 = e_{2}\) and \(\beta = g_2(\alpha)\) that 
is feasible
also satisfies Eq.~\eqref{eq:Dsm det upper}. This directly implies that \(\Dsm\) satisfies Eq.~\eqref{eq:Dsm det upper}.
To simplify, let \(\theta = e_{2} - \beta\). Since \(\Dsm\) satisfies \(\beta = g_2(\alpha)\), \(\alpha\) can also be expressed in terms of \(\theta\).
We then construct two functions as follows,
\begin{align}
   h_1(\theta) &= \frac{\sqrt{[\det (\widetilde{\Dm})][\det (\Em)]}}{\theta^{N-2}}, \label{eq:h1}
\end{align}
\begin{align}
   h_2(\theta) &= e_{1}\left[e_{2}+(N-2)(e_{2}-\theta)\right]-(N-1) 
   \Bigg[\rho-\sqrt{\frac{1-e_{1}}{N-1}}\sqrt{(N-2)(\theta+\rho-e_{2})+1-e_{2}}\Bigg]^2. \label{eq:h2}
\end{align}
For any \(\theta \in [0, e_{2}]\), Eq.~\eqref{eq:Dsm det upper} that needs to be proved can be written as
\begin{align} \label{eq:h1 geq h2}
   h_1(\theta) \geq h_2(\theta).
\end{align}
Since \(h_2''(\theta) < 0\), \(h_2(\theta)\) is concave down. A first-order Taylor expansion at \(\theta = e_{2}\) gives \(L_{h2}(\theta)\), which satisfies
\begin{align} \label{eq:h2 taylor}
   L_{h2}(\theta) \geq h_2(\theta)
\end{align}
with 
\begin{align}
   L_{h2}(\theta) &= \left(N-2\right)\left(\phi\sqrt{1-e_{1}}-1\right)\theta +\phi\sqrt{1-e_{1}}\left\{2\left[\left(N-2\right)\rho+1\right]-Ne_{2}\right\} \notag \\
   & \quad +\left(N-1\right)\left(e_{2}-\rho^{2}\right)-\left[\left(N-2\right)\rho+1\right]\left(1-e_{1}\right), \label{eq:Lh2}
\end{align}
where
\(
\phi = \frac{\sqrt{N - 1}\rho}{\sqrt{\left(N - 2\right)\rho + 1 - e_{2}}}
\).
Next, we judiciously adopt different proofs for \(N \geq 3\) to ensure efficiency and simplicity in the derivations.

We first prove Eq.~\eqref{eq:h1 geq h2} for \( N = 3 \). Now we aim to prove that this inequality
\begin{align} \label{eq:N=3 inequality origin}
   h_{1}(\theta) \geq L_{h2}(\theta)
\end{align}
holds, which is explicitly given in 
   \begin{align} \label{eq:N=3 inequality}
&\left[e_{1}e_{2}^{3}\left(e_{1}e_{2}-2(\alpha_{g2,1}^{\iso})^2\right)\right]^{\frac{1}{2}}\theta^{-1}+ \left(1-\bar{\phi}\sqrt{1-e_{1}}\right)\theta \notag \\
\geq {} &\bar{\phi}\sqrt{1-e_{1}}\left[2\left(\rho+1\right)-3e_{2}\right]+2\left(e_{2}-\rho^{2}\right)-\left(\rho+1\right)\left(1-e_{1}\right).
   \end{align}
When Eq.~\eqref{eq:N=3 inequality origin} holds, combined with Eq.~\eqref{eq:h2 taylor}, it directly leads to Eq.~\eqref{eq:h1 geq h2}.
To distinguish it, we denote the value of \(\phi\) when \(N = 3\) by \(\bar{\phi}\), where \(\bar{\phi} = \frac{\sqrt{2}\rho}{\sqrt{1+\rho-e_{2}}}\).
For the left-hand side (LHS) of Eq.~\eqref{eq:N=3 inequality}, we obtain a lower bound by applying the arithmetic-geometric mean (AM-GM) inequality. By substituting \(\alpha_{g2,1}^{\iso}\), letting \(x = \sqrt{1 - e_1}\) for simplicity, and examining the boundaries \(\mathcal{B}_{4,6}^{(1)}\) and \(\mathcal{B}_{0,4}^{(4)}\), the domain of \(x\) for any \(N\) is 
\begin{align} \label{eq:x domain}
   \kappa \leq x \leq \phi,
\end{align}
where \(\kappa=\frac{\left(N-1\right)\rho^{2}}{\phi\left(\left(N-2\right)\rho+1\right)}\). Squaring reformulates Eq.~\eqref{eq:N=3 inequality} after lower-bounding the LHS, and through algebraic simplifications, we obtain 
\begin{align}
& -16e_{2}^{3}\left(\rho+1-e_{2}\right)\left(x-\bar{\phi}\right)^{2}\left(1-\bar{\phi}x\right)^{2}\left(1-x^{2}\right) \notag \\
& +16e_{2}^{4}\left[\left(1-x^{2}\right)^{2}\left(1-\bar{\phi}x\right)^{2}-\left(1-\bar{\phi}^{2}\right)^{4}\right]+16e_{2}^{4}\left(1-\bar{\phi}^{2}\right)^{4} \notag \\
\geq& \left\{2e_{2}\left(1-\bar{\phi}x\right)+\left[-\left(1+\rho\right)x+\left(1+\rho-e_{2}\right)\bar{\phi}\right]\left(x-\bar{\phi}\right)\right\}^{4}. \label{eq:N=3 inequality v1}
\end{align}
Equality is verified to hold at \(x = \bar{\phi}\).
Define the difference between the LHS and right-hand side (RHS) of Eq.~\eqref{eq:N=3 inequality v1} as \(\Delta(x)\), with its derivative
\(\frac{\mathrm{d} \Delta(x)}{\mathrm{d}x} = -4(x - \bar{\phi})P_{6}(x)\).\footnote{In the notation, we have introduced \(P_{m}(x)\) as a polynomial function of degree \(m\).} 
The fifth derivative \(P_{6}^{(5)}(x) > 0\) is a linear function under the condition in Eq.~\eqref{eq:x domain}, which indicates that \(P_{6}^{(4)}(x)\) is increasing. Since the sign of the product \(P_{6}^{(4)}(\phi)P_{6}^{(4)}(\kappa)\) is indeterminate, \(P_{6}^{(4)}(x)\) may have a root, corresponding to a local minimum of \(P_{6}^{(3)}(x)\). As \(P_{6}^{(3)}(\phi) \leq 0\) and \(P_{6}^{(3)}(\kappa) \leq 0\), \(P_{6}^{(2)}(x)\) decreases and may also have a root, corresponding to a local maximum of \(P_{6}'(x)\). Since \(P_{6}'(\phi) \geq 0\) and \(P_{6}'(\kappa) \geq 0\), \(P_{6}(x)\) is an increasing function. With \(P_{6}(\phi) \leq 0\), it follows that \(\Delta'(x) \leq 0\), which implies \(\Delta(x) \geq \Delta(\phi) = 0\). Therefore, Eq.~\eqref{eq:N=3 inequality v1} holds, completing the proof for the case \(N = 3\).

When \(N > 3\), the above method becomes impractical for proving Eq.~\eqref{eq:N=3 inequality origin} for two main reasons. First, the AM-GM inequality is no longer directly applicable because the exponents of \(\theta\) do not match. Second, operations such as raising to integer powers (e.g., squaring) become computationally demanding as \(N\) increases, making subsequent processing overly tedious.
The proof of Eq.~\eqref{eq:N=3 inequality origin} is complicated by the non-linear terms in \(h_1(\theta)\) from Eq.~\eqref{eq:h1}. Therefore, the argument proceeds by judiciously selecting a point for the linear approximation of \(h_1(\theta)\).
For \(\theta \in [0, e_{2}]\), let \(\theta_0\) denote the point satisfying
\(
h_1'(\theta = \theta_0) = m_{h2}
\),
where \(m_{h2}\) is the slope of \(L_{h2}(\theta)\) in Eq.~\eqref{eq:Lh2}. A first-order Taylor expansion of \(h_1(\theta)\) at \(\theta = \theta_0\) yields \(L_{h1}(\theta)\). Since \(h_1''(\theta) > 0\), \(h_1(\theta)\) is concave up, leading to
\begin{align} \label{eq:h1 taylor}
   L_{h1}(\theta) \leq h_1(\theta).
\end{align}
As \(L_{h1}(\theta)\) and \(L_{h2}(\theta)\) are parallel, using Eq.~\eqref{eq:h2 taylor} and Eq.~\eqref{eq:h1 taylor}, a sufficient condition for Eq.~\eqref{eq:h1 geq h2} to hold is that
\begin{align} \label{eq:intercept e1 b}
b_{h1} \geq b_{h2}
\end{align}
is satisfied, where \(b_{h1}\) and \(b_{h2}\) are the intercepts of \(L_{h1}(\theta)\) and \(L_{h2}(\theta)\), respectively. Eq.~\eqref{eq:intercept e1 b} is explicitly given in 
   \begin{align}
      &\left(N-1\right)e_{1}^{\frac{1}{2\left(N-1\right)}}e_{2}^{\frac{\left(2N-3\right)}{2\left(N-1\right)}}\zeta^{\frac{1}{N-1}}\left(1-\phi\sqrt{1-e_{1}}\right)^{\frac{N-2}{N-1}} \notag \\
   \geq& \left(2\phi\sqrt{1-e_{1}}+e_{1}-1\right)\left[\left(N-2\right)\rho+1\right]+e_{2}\left[N\left(1-\phi\sqrt{1-e_{1}}\right)-1\right]-\left(N-1\right)\rho^{2}, \label{eq:intercept e1}
   \end{align}
   with
   \begin{align} \label{eq:zeta def}
      \zeta&=\Big[e_{1}\left(\left(N-1\right)\rho+1\right)\left(1-\rho\right)-\left(\sqrt{\left(N-2\right)\rho+1-e_{2}}-\rho\sqrt{1-e_{1}}\sqrt{N-1}\right)^{2}\Big]^{\frac{1}{2}}.
   \end{align}
Using \(x=\sqrt{1-e_{1}}\), Eqs.~\eqref{eq:intercept e1} and~\eqref{eq:zeta def} can be rewritten by introducing \(x\), \(\tau=\frac{\left(N-2\right)\rho+1-2e_{2}}{\left(N-2\right)\rho+1}\phi\), and \(\upsilon=\frac{\left(N-2\right)\rho+1-Ne_{2}}{\left(N-2\right)\rho+1-2e_{2}}\tau\). Specifically, the RHS of Eq.~\eqref{eq:intercept e1} becomes
\begin{align}
   &-\left[\left(N-2\right)\rho+1\right]\left(x-\phi\right)\left(x-\upsilon\right) + e_{2}\left(N-1\right)\left(1-\phi^{2}\right). \label{eq:rhs simplify}
\end{align}
Similarly, the LHS is further simplified to
\begin{align}
   &\left(N-1\right)e_{2}^{\frac{\left(2N-3\right)}{2\left(N-1\right)}}(1-x^{2})^{\frac{1}{2\left(N-1\right)}}\left(1-\phi x\right)^{\frac{N-2}{N-1}}\zeta^{\frac{1}{N-1}}. \label{eq:lhs simplify}
\end{align}
For \(\zeta\) in Eq.~\eqref{eq:zeta def}, it can be rewritten as
\begin{align}
   \zeta=&\left[e_{2}\left(1-\phi^{2}\right)-\left(\left(N-2\right)\rho+1\right)\left(x-\phi\right)\left(x-\tau\right)\right]^{\frac{1}{2}}.
   \label{eq:a trans}
\end{align}
Thus, by inserting Eq.~\eqref{eq:a trans} into Eq.~\eqref{eq:lhs simplify} and combining Eq.~\eqref{eq:rhs simplify}, Eq.~\eqref{eq:intercept e1} can be reformulated on a logarithmic scale in\footnote{For simplicity, we use the natural logarithm.}.
   \begin{align}
      &\log\left(N-1\right)+\frac{\left(2N-3\right)}{2\left(N-1\right)}\log\left(e_{2}\right)+\frac{\log\left(1-x^{2}\right)}{2\left(N-1\right)}+\frac{N-2}{N-1}\log\left(1-\phi x\right)+\frac{\log\left(\zeta\right)}{N-1} \notag \\
      \geq& \log\!\left[-\left(\left(N-2\right)\rho+1\right)\left(x-\phi\right)\left(x-\upsilon\right) + e_{2}\left(N-1\right)\left(1-\phi^{2}\right)\right] \label{eq:intercept x v1}
   \end{align}
with equality satisfied at \(x = \phi\).
Let \(\Delta(x) \triangleq \text{LHS} - \text{RHS}\) of the inequality in Eq.~\eqref{eq:intercept x v1}.
After eliminating the common factor \(\frac{\left(N-2\right)\rho+1}{\left(N-1\right)\left(1-\phi^{2}\right)e_{2}} \geq 0\) and removing the root \(x = \phi\) by factoring out \((x - \phi)\), we obtain \(\bar{\Delta}(x)\), a simplified and sign-reversed version of \(\Delta(x)\) within \(\kappa \leq x \leq \phi\).
Its derivative is
\begin{align} \label{eq:derivative inequality}
   \bar{\Delta}'(x)=q_{1}(x) - q_{2}(x) - q_{3}(x) - q_{4}(x),
\end{align}
where
\begin{align}
   q_{1}(x) &= \frac{C\left(\phi-\upsilon\right)\left(x-\upsilon\right)+2\left(N-1\right)\left(1-\phi^{2}\right)e_{2}}{C\left(\phi-x\right)\left(x-\upsilon\right)+\left(N-1\right)\left(1-\phi^{2}\right)e_{2}}, \\
   q_{2}(x) &= \frac{C\left(\tau-\phi\right)\left(x-\tau\right)-2\left(1-\phi^{2}\right)e_{2}}{2C\left(x-\phi\right)\left(x-\tau\right)-2\left(1-\phi^{2}\right)e_{2}}, \\
   q_{3}(x) &= \frac{e_{2}}{C}\frac{\left(N-2\right)\phi^{2}}{1-\phi x}, \\
   q_{4}(x) &= \frac{e_{2}}{C}\frac{1+\phi x}{1-x^{2}}
\end{align}
and \(C=\left(N-2\right)\rho+1\).
To prove Eq.~\eqref{eq:intercept e1 b}, we establish its equivalent form Eq.~\eqref{eq:intercept x v1} by showing the difference \(\Delta(x) \geq 0\) or \(\bar{\Delta}(x) \leq 0\). Noting that \(\bar{\Delta}(\phi) \leq 0\) and \(\kappa \leq x \leq \phi\), a sufficient condition is to establish that its first derivative satisfies \(\bar{\Delta}'(x) \geq 0\),
and we derive the required inequality
\begin{align} \label{eq:derivative inequality v1}
   q_{1}(x)-q_{3}(x) \geq q_{2}(x)+q_{4}(x)
\end{align}
based on Eq.~\eqref{eq:derivative inequality}.
For the LHS of Eq.~\eqref{eq:derivative inequality v1}, written as \(F_{\text{LHS}}\), the sign of its derivative is determined by the numerator, which is a quartic function \(P_{4}(x)\).
The third derivative \(P_{4}^{(3)}(x)\) satisfies \(P_{4}^{(3)}(x) \geq 0\) within Eq.~\eqref{eq:x domain}, and the second derivative \(P_{4}^{(2)}(x)\) increases. Since \(P_{4}^{(2)}(x) \leq P_{4}^{(2)}(\phi) \leq 0\) holds, the first derivative \(P_{4}'(x)\) decreases. With \(P_{4}'(x) \geq P_{4}'(\phi) \geq 0\), \(P_{4}(x)\) increases. Similarly, the RHS is also an increasing function with respect to \(x\).
Thus, the LHS of Eq.~\eqref{eq:derivative inequality v1} satisfies
\(F_{\text{LHS}}(x) \geq F_{\text{LHS}}(x = \kappa)\),
where \(F_{\text{LHS}}(x = \kappa) = 2\), while for the RHS of Eq.~\eqref{eq:derivative inequality v1}, \(F_{\text{RHS}}\), we have
\(F_{\text{RHS}}(x) \leq F_{\text{RHS}}(x = \phi)\),
where
\begin{align}
   &F_{\text{RHS}}(x = \phi) 
   = 1-\frac{e_{2}\left[3\left(N-1\right)\rho^{2}+\left(N-2\right)\rho+1-e_{2}\right]}{\left[\left(N-2\right)\rho+1\right]\left[e_{2}+\left(N-1\right)\rho^{2}-\left(N-2\right)\rho-1\right]}.
\end{align}
A sufficient condition for the validity of \(\bar{\Delta}'(x) \geq 0\) or Eq.~\eqref{eq:derivative inequality v1} is given by
\begin{align} \label{eq:sufficient cond1}
   F_{\text{LHS}}(x = \kappa) \geq F_{\text{RHS}}(x = \phi),
\end{align}
which is true if and only if \((e_{1}, e_2) \in \widetilde{\mathcal{R}}_{4}^\mathrm{B}\). The corresponding boundary is
\begin{align}
   \widetilde{\mathcal{B}}_{4} \!=\! {} &\Bigg\{(e_{1}, e_2) \in [0,1]^{2} : e_{2} = \frac{3\left(N-1\right)\rho^{2}}{2} + \left(N-2\right)\rho + 1  \notag \\
   &\;\;- \frac{\rho\sqrt{\left(N-1\right)\left\{9\left(N-1\right)\rho^{2} + 16\left[\left(N-2\right)\rho + 1\right]\right\}}}{2}\Bigg\}.
\end{align}
To relate \(\widetilde{\mathcal{R}}_{4}^\mathrm{B}\) to \(\mathcal{R}_{4}\), 
it is necessary to examine the topmost boundary \(\mathcal{B}_{1,4}^{(3)}\) of \(\mathcal{R}_{4}\). Due to its implicitness, we expand it to the line segment \(\overline{P_1P_4}\): \(e_{2}=1-\rho_{2}\) based on Proposition~\ref{prop:region one B3 e0star}.
For \(\widetilde{\mathcal{R}}_{4}^\mathrm{B}\), the minimum \(e_{2}\) on \(\widetilde{\mathcal{B}}_{4}\) satisfies \(\min_{e_{2}} \widetilde{\mathcal{B}}_{4} \geq 1 - \rho\), with equality when \(N = 5\). Hence, it follows that \(\mathcal{R}_{4} \subset \widetilde{\mathcal{R}}_{4}^\mathrm{B}\) when \(N \geq 5\).
Thus, we have proven Eq.~\eqref{eq:sufficient cond1}, which leads to Eq.~\eqref{eq:derivative inequality v1}.

For the case \(N = 4\), Eq.~\eqref{eq:derivative inequality v1} can be rewritten as
\begin{align} \label{eq:derivative inequality v2}
   q_{1}(x) - q_{2}(x) \geq q_{3}(x) + q_{4}(x).
\end{align}
Similarly to Eq.~\eqref{eq:derivative inequality v1}, it can be verified that both the LHS and RHS of Eq.~\eqref{eq:derivative inequality v2} are increasing functions with respect to \(x\). For the LHS, we have
\(G_{\text{LHS}}(x) \geq G_{\text{LHS}}(x = \kappa)\),
and for the RHS, we have
\(G_{\text{RHS}}(x) \leq G_{\text{RHS}}(x = \phi)\).
Therefore, we derive another sufficient condition for the validity of \(\bar{\Delta}'(x) \geq 0\) or Eq.~\eqref{eq:derivative inequality v2}
\begin{align} \label{eq:sufficient cond2}
   G_{\text{LHS}}(x = \kappa) \geq G_{\text{RHS}}(x = \phi),
\end{align}
which defines the region \(\bar{\mathcal{R}}_{4}^\mathrm{B}\) in the \((e_{1},e_2)\) plane. The corresponding boundary \(\bar{\mathcal{B}}_{4}\) is given in 
\begin{align*}
& \bar{\mathcal{B}}_{4} = \Big\{ (e_{1}, e_2) \in [0,1]^{2} : \\
& \left(1-\rho\right)^{2}\left[\left(N-2\right)\rho+1\right]^{2}\left[\left(N-1\right)\rho+1\right]^{2}+e_{2}^{3}\left[\left(N-1\right)\rho^{2}-\left(N-2\right)\rho-1\right] \\
&+ e_{2}^{2}\left(\left[\left(N-2\right)\rho+1\right]\left[\left(\left(N-2\right)\rho+1\right)\left(3-\left(N+1\right)\rho^{2}\right)+\rho\left(1-\rho\right)\left(3\rho-1\right)\right]+\rho\left(1-\rho\right)^{3}\right) \\
\refstepcounter{equation}
&-3e_{2}\left(1-\rho\right)\left[\left(N-2\right)\rho+1\right]^{2}\left[\left(N-1\right)\rho+1\right] = 0 \Big\}.\tag{\theequation} \label{eq:bar B4}
\end{align*}
When solving \(\bar{\mathcal{B}}_{4}\) for \(e_{2}\), we encounter a cubic equation. The roots of this cubic equation are tedious, but they are more practical for verification in finite cases, which is why we adopt this method for \(N = 4\).  
By converting the boundary equation into a functional form, it is straightforward to verify that \(\bar{\mathcal{B}}_{4}^{N=4}\) decreases over \(e_{2} \in [0, 1-\rho]\). 
We then have \(\bar{\mathcal{B}}_{4}^{N=4} \geq \bar{\mathcal{B}}_{4}^{N=4}(e_{2} = 1-\rho) = 9 \rho^{3} \left(1-\rho\right)^{2} \left(1+\rho\right) \geq 0\). Therefore, we also have \(\mathcal{R}_{4} \subset \bar{\mathcal{R}}_{4}^\mathrm{B}\) when \(N = 4\). Finally, we complete the proof of Eq.~\eqref{eq:derivative inequality v2} for \(N = 4\) by establishing Eq.~\eqref{eq:sufficient cond2}. Notably, \(\bar{\mathcal{R}}_{4}^\mathrm{B}\) contains a larger feasible region than \(\widetilde{\mathcal{R}}_{4}^\mathrm{B}\), so Eq.~\eqref{eq:sufficient cond2} provides a weaker sufficiency compared to Eq.~\eqref{eq:sufficient cond1}. However, as mentioned earlier, proving Eq.~\eqref{eq:sufficient cond2} for arbitrary \(N > 3\) is exceptionally cumbersome.

When \( N \geq 5 \), Eq.~\eqref{eq:derivative inequality v1} holds, and when \( N = 4 \), Eq.~\eqref{eq:derivative inequality v2} holds. We conclude that for \( N > 3 \), \( \bar{\Delta}'(x) \geq 0 \) holds, and thus for \( \kappa \leq x \leq \phi \), the difference in Eq.~\eqref{eq:intercept x v1} satisfies \( \Delta(x) \geq 0 \). This verifies the validity of Eq.~\eqref{eq:intercept e1 b}. Finally, combining the result for \( N = 3 \), Eq.~\eqref{eq:Dsm det upper} is proved for \( N \geq 3 \).


\bibliographystyle{IEEEtran}
\bibliography{references}

\end{document}